\documentclass[11pt]{article}
\usepackage{tikz,tikz-cd}
\usepackage{subcaption}
\usepackage{pgfplots}\pgfplotsset{compat=1.18}
\usepackage{cancel}

\usepackage[utf8]{inputenc}
\usepackage[T1]{fontenc}
\usepackage[english]{babel}
\usepackage{amsmath,amssymb,amsthm}
\usepackage{thmtools}
\usepackage{graphicx} 
\usepackage[margin=1.in]{geometry}
\usepackage[shortlabels]{enumitem}
\usepackage{todonotes}
\usepackage[bibstyle=numeric, backend=biber, sorting=none, citestyle=numeric-comp, giveninits,url=false,maxbibnames=5]{biblatex} 
\usepackage{csquotes}\MakeOuterQuote"
\usepackage{soul}\setstcolor{red}
\usepackage[colorlinks,citecolor=blue,linkcolor=blue]{hyperref}
\usepackage[capitalize]{cleveref}
\usepackage{comment}

\DeclareUnicodeCharacter{0141}{\Lbar{}}

\AtEveryBibitem{\clearfield{month}}
\AtEveryBibitem{\clearfield{day}}

\newtheorem{theorem}{Theorem}
\numberwithin{theorem}{section}
\newtheorem{lemma}[theorem]{Lemma}
\newtheorem{corollary}[theorem]{Corollary}
\newtheorem{proposition}[theorem]{Proposition}

\newtheorem{definition}[theorem]{Definition}

\newtheorem{introtheorem}{Theorem} 

\newtheorem{introproposition}[introtheorem]{Proposition}

\theoremstyle{definition}

\newtheorem*{problem*}{Problem}
\newtheorem*{example*}{Example}
\newtheorem*{assumption*}{Assumption}
\newtheorem{example}[theorem]{Example}
\newtheorem{remark}[theorem]{Remark}
\newtheorem*{warning*}{Warning}

\newcommand{\ip}[2]{\langle #1,#2\rangle}

\newcommand{\ketbra}[2]{|#1\rangle\langle#2|}

\newcommand{\kettbra}[1]{\ketbra{#1}{#1}}

\DeclareMathOperator{\supp}{supp}
\newcommand{\norm}[1]{\lVert #1\rVert}
\newcommand{\oo}{\infty}
\newcommand{\ox}{\otimes}
\newcommand{\mc}{\mathcal}
\newcommand{\eps}{\varepsilon}

\newcommand{\abs}[1]{\lvert #1 \rvert}
\newcommand{\up}[1]{^{(#1)}}

\DeclareMathOperator{\tr}{tr}

\newcommand{\hide}[1]{}

\def\CC{{\mathbb C}}

\def\H{{\mc H}}

\def\K{{\mathcal K}}

\def\RR{{\mathbb R}}

\def\NN{{\mathbb N}}

\newcommand{\R}{\mc R}
\renewcommand{\Im}{\mathrm{Im}}
\renewcommand{\Re}{\mathrm{Re}}

\DeclareMathOperator{\lin}{span}
\DeclareMathOperator{\id}{id}

\DeclareMathOperator{\Sp}{Sp}

\newcommand*{\1}{\text{\usefont{U}{bbold}{m}{n}1}}
\newcommand{\placeholder}[0]{{\,\cdot\,}}

\renewcommand{\Re}{\mathrm{Re}}

\newcommand{\Lbar}{\L{}}

\newcommand{\qandq}{\quad\text{and}\quad}

\let\Lbar\L
\def\L{{\mathcal L}}

\def\D{{\mathbb D}}

\def\cptp{\overset{\text{\tiny CPTP}}\longleftrightarrow}
\def\notcptp{\cancel{\overset{\text{\tiny CPTP}}\longleftrightarrow}}
\def\ptp{\overset{\text{\tiny PTP}}\longleftrightarrow}

\def\staralg{\mathrm{\textup{*-}alg}}
\def\Jstaralg{\mathrm{\textup{J*-}alg}}
\def\revJstaralg{\mathrm{\textup{RJ*-}alg}}
\def\Jcong{\mathrel{\cong_{\mathrm{\textup{J*}}}}}

\numberwithin{equation}{section}

\def\conv{\mathrm{conv}}
\newcommand{\proj}[1]{[#1]}

\newcommand{\jp}[2]{\{#1,#2\}}
\newcommand{\trip}[3]{\{#1,#2,#3\}}

\title{\vspace{-2cm}Sufficiency and Petz recovery for positive maps}
\author{Lauritz van Luijk$^{1,2}$, Henrik Wilming$^3$}
\date{\footnotesize{
$^1$Perimeter Institute for Theoretical Physics, Waterloo, Ontario, Canada\\
$^2$Institute for Quantum Computing, Waterloo, Ontario, Canada\\
$^3$Leibniz Universität Hannover, Institut für Theoretische Physik, Appelstraße 2, 30167 Hannover, Germany\\
\today
}}
\begin{document}
\maketitle
\abstract{
	We study the interconversion of families of quantum states (``statistical experiments'') via positive, trace-preserving (PTP) maps and clarify its mathematical structure in terms of minimal sufficient Jordan algebras, which can be seen to generalize the Koashi-Imoto decomposition to the PTP setting. In particular, we show that Neyman-Pearson tests generate the minimal sufficient Jordan algebra, and hence also the minimal sufficient *-algebra corresponding to the Koashi-Imoto decomposition. 
    As applications, we show that a) equality in the data-processing inequality for the relative entropy or the $\alpha$-$z$ quantum R\'enyi divergence implies the existence of a recovery map also in the PTP case and b) that two dichotomies can be interconverted by PTP maps if and only if they can be interconverted by decomposable, trace-preserving maps. 
    We thoroughly review the necessary mathematical background on Jordan algebras.
    As a step beyond the finite-dimensional case, we also prove Frenkel's formula for approximately finite-dimensional von Neumann algebras.
}

\clearpage
\tableofcontents
\noindent

\clearpage

\section{Introduction and overview}

Consider a pair of quantum states $(\rho,\sigma)$, also called a ``dichotomy'', on a finite-dimensional Hilbert space\footnote{We restrict to finite dimensions throughout.} $\H$ representing different preparations of the same system. It is a natural and fundamental question to ask: How different, i.e., distinguishable, are the two states?

Let us discuss two different ways to address this question: The first, and maybe most natural one, is via (Bayesian) hypothesis testing.
That is, assuming that $\rho$ is (assumed to be) prepared with some prior probability $p$ and $\sigma$ with some prior probability $1-p$, we perform a binary measurement represented by a positive operator-valued measure $(x,\1-x)$ with effect operator $0\leq x\leq \1$. If the outcome is $0$ (corresponding to $x$), then we guess that $\rho$ was prepared; if the outcome is $1$ (corresponding to $\1-x$) we guess that $\sigma$ is prepared. 
Then $(\rho,\sigma)$ are highly distinguishable if the success probability is large, and are little distinguishable if the success probability is low.
The optimal test $x$ is essentially unique and given by the projector $\proj{\rho> t \sigma}$ onto the positive part of $\rho - t\sigma$, with $t = (1-p)/p$ \cite{helstromDetectionTheoryQuantum1967}, see also \cref{sec:bayes}.
A particularly important subtask of hypothesis testing is \emph{asymmetric} hypothesis testing, where one tries to minimize the error probability to wrongly guess $\rho$ if the actual state is $\sigma$ under the constraint that the probability to correctly identify $\rho$ if it is prepared is lower bounded by some value $1-\eps$. 
The quantum Stein's Lemma \cite{ogawaStrongConverseSteins2000} shows that in the asymptotic limit of many independent copies, the resulting error probability decreases exponentially with rate given by the quantum relative entropy $D(\rho\|\sigma) = \tr(\rho \log\rho) - \tr(\rho\log\sigma)$.

A different way to think about distinguishability is in terms of physical processes that are applied to the system. 
Any physical process, applied after the respective preparations, should only be able to reduce the distinguishability of $\rho$ and $\sigma$. 
Any distinguishability measure $\D$ should hence fulfill the \emph{data-processing inequality} 
\begin{equation}\label{eq:DPI}
	\D(\rho\|\sigma) \ge \D(T^*\rho \| T^*\sigma),
\end{equation}
for any completely positive, trace-preserving (CPTP) map $T^*$.%
\footnote{Throughout this paper, we denote trace-preserving maps by $T^*,S^*,\ldots$ because our techniques are of algebraic nature, making it more natural to regard unital maps, denoted $T,S,\ldots$ as primary objects.
The two points of view are, of course, equivalent since trace-preserving maps are dual to unital maps.} 
Examples are the success probability in hypothesis testing discussed above, or the quantum relative entropy $D$.
In fact the literature of quantum information theory exhibits a whole zoo of so-called \emph{divergences} $\D$ (see \cite{tomamichelQuantumInformationProcessing2016,hiaiQuantumFDivergencesNeumann2021}): Positive functions on pairs of density matrices, which in addition to \eqref{eq:DPI} also fulfill $\D(\rho\|\sigma) = 0$ if and only if $\rho=\sigma$ (and possibly have additional desirable properties, such as additivity under tensor products). 
The problem of determining when one dichotomy can be converted into another one (but not necessarily the other way around) has a long history dating back to Alberti and Uhlmann \cite{albertiProblemRelatingPositive1980}, see also \cite{matsumotoReverseTestCharacterization2010,heinosaariExtendingQuantumOperations2012,buscemiComparisonQuantumStatistical2012a,buscemiQuantumRelativeLorenz2017,gour_quantum_2018,buscemi_information-theoretic_2019} for examples of recent work. 

Two dichotomies $(\rho_1,\sigma_1)$ and $(\rho_2,\sigma_2)$ are clearly equally distinguishable if there are physical processes that turn each pair of preparation procedures into the other pair. 
In other words, if there are CPTP maps $T^*$ and $S^*$ such that
\begin{align}\label{eq:cptp-equivalence}
		(\rho_2,\sigma_2) = (T^*\rho_1,T^*\sigma_1),\quad (\rho_1,\sigma_1) = (S^* \rho_2,S^*\sigma_2). 
\end{align}
In this case we say that the two dichotomies are \emph{CPTP-equivalent}, also denoted as $(\rho_1,\sigma_1)\cptp(\rho_2,\sigma_2)$. We say that they are \emph{PTP-equivalent} if there are positive, trace-preserving maps $T^*,S^*$ such that \eqref{eq:cptp-equivalence} is true. 
Distinguishability measures are constant on CPTP-equivalence classes by \eqref{eq:DPI}.
Matsumoto showed that a dichotomy being less distinguishable in terms of hypothesis testing (lower success probability for all priors $p>0$) does \emph{not} imply that there exists even a PTP map converting one dichotomy to the other \cite{matsumotoExampleQuantumStatistical2014}. Thus, the success probabilities are not sufficient to decide (C)PTP-equivalence.

This already hints at a problem when thinking about distinguishability soley in terms of CPTP-equivalence classes: 
There are pairs of dichotomies $(\rho_1,\sigma_1)$ and $(\rho_2,\sigma_2)$ which are clearly equally distinguishable, but
are \emph{not} CPTP-equivalent. The following simple example is taken from \cite{galke_sufficiency_2024}. Let $\dim(\H) \geq 3$. Then there are pairs $(\rho,\sigma)$ such that
\begin{equation}\label{introeq:introexample}
	(\rho,\sigma) \quad \notcptp \quad (\rho^t,\sigma^t),
\end{equation}
where $(\placeholder)^t$ denotes transposition in some fixed basis.
However, it is clear that, even though the transpose is not completely positive, but merely positive, for any procedure to distinguish $\rho$ from $\sigma$ there exists (at least in principle) a different procedure that distinguishes $\rho^t$ from $\sigma^t$ just as well: Simply take the transpose of all involved operators describing the procedure (including possible auxiliary systems). Moreover, this remains true when taking independent copies, i.e., considering $\rho^{\ox n}$ and $\sigma^{\ox n}$. 

This is also reflected in the behaviour of divergences: All known quantum divergences are invariant under taking the transpose of both arguments. Moreover, as far as we are aware, for all divergences for which \eqref{eq:DPI} has been shown for CPTP maps \emph{and} it has been clarified whether \eqref{eq:DPI} holds for PTP maps, it has turned out that \eqref{eq:DPI} in fact holds for PTP maps. 
The largest class is given by the $\alpha$-$z$-R\'enyi divergences \cite{audenaertAzRenyiRelativeEntropies2015,katoAzRenyiDivergenceNeumann2024,hiaiazRenyiDivergences2024}, which includes Petz-R\'enyi divergences \cite{petzQuasientropiesStatesNeumann1985,petzQuasientropiesFiniteQuantum1986} as well as the minimal (or sandwiched) R\'enyi divergence \cite{wildeStrongConverseClassical2014,muller-lennertQuantumRenyiEntropies2013,beigiSandwichedRenyiDivergence2013,bertaRenyiDivergencesWeighted2018,jencovaRenyiRelativeEntropies2018a}, see also \cite[Appendix E]{galke_sufficiency_2024} for an overview. 
The case of the quantum relative entropy was first shown in \cite{muller-hermesMonotonicityQuantumRelative2017}, with later, independent proof in \cite{jencovaRenyiRelativeEntropies2018a,frenkel_integral_2023}.
In other words, known distinguishability measures cannot distinguish between CPTP-equivalence and PTP-equivalence. 

Generalizing from dichotomies, these observations motivate us to study PTP-equivalence of general \emph{statistical experiments}, i.e., finite sets $(\rho_\theta)_{\theta\in\Theta}$ of density matrices on a common finite-dimensional Hilbert space $\H$, and how it relates to (Bayesian) hypothesis testing.
In the remainder of this section, we provide an overview of our main results.
Without loss of generality, we may and will assume in the following that statistical experiments are always \emph{faithful}, i.e. if $a\geq 0$, then $\tr(a\rho_\theta) = 0$ for all $\theta$ implies $a=0$ (any statistical experiment is CPTP-equivalent to a faithful one, see \cref{sec:experiments}). 
We note here already that \cref{sec:maths} collects the necessary mathematical background that we need to establish our results and may be outside the usual mathematical scope of quantum information theory.

\subsection{Sufficiency and the structure of equivalent statistical experiments}

A central notion of this work is \emph{sufficiency}. 
Fix a faithful statistical experiment $(\rho_\theta)$ (we suppress the label set $\Theta$) on $L(\H)$.
A unital *-algebra\footnote{Since we work in finite dimensions, there is no distinction between von Neumann algebras, C*-algebras, and unital *-algebras on $\H$.}
$A\subset L(\H)$ is called sufficient for the statistical experiment $(\rho_\theta)$ if there exists a unital, completely-positive (UCP) map 
$T: L(\H) \to A$ such that $T^*\rho_\theta = \rho_\theta$ for all $\theta\in \Theta$ \cite{ohya_quantum_1993,jencova_sufficiency_2006}. 
Here $T^*$ is the Hilbert-Schmidt dual defined via $\tr(T^*(a) b) = \tr(a Tb)$.
This means that the states $\rho_\theta$ can be recovered from their restrictions to the subalgebra $A$.
In other words, all information about the statistical experiment is contained in $A$.\footnote{And $T$, but we will later see that there is a canonical choice for $T$.}
Among all CPTP-sufficient *-algebras, there is a minimal one, which we denote by $A_{(\rho_\theta)}$ and call the \emph{minimal sufficient *-algebra} \cite{kuramochi_minimal_2017}.
The minimal sufficient *-algebra is the same (up to isomorphism) as the algebra obtained from the \emph{Koashi-Imoto decomposition} of a statistical experiment \cite{koashi_operations_2002,hayden_structure_2004,kuramochi_minimal_2017}.
Petz neatly characterized sufficient *-algebras \cite{petz_sufficient_1986,petz_sufficiency_1988}, see also \cite{ohya_quantum_1993,jencova_sufficiency_2006}.

In this work, we generalize from sufficient *-algebras. Let $K\subset L(\H)$ be an operator system, i.e., a unital, complex subspace closed under the adjoint. 
We say that $K$ is (C)PTP-sufficient for $(\rho_\theta)$ if there exists unital (completely) positive (U(C)P) map
\begin{align}\label{introeq:sufficiency-K}
	T : L(\H)\to K,\qquad T^*\rho_\theta = \rho_\theta,\qquad \theta\in \Theta.
\end{align}
An operator system is \emph{minimal (C)PTP-sufficient} if it is contained in all other (C)PTP-sufficient operator systems. A conditional expectation onto a *-algebra $A$ is a UCP map $F:L(\H)\to A$ such that $F^2=F$ (see \cref{sec:CEs}).
From general properties of completely positive maps and Petz's theorem, it follows that:
\begin{equation}\label{introeq:CPTP-sufficiency}
\begin{gathered}
	\textit{There is a minimal CPTP-sufficient operator system $A_{(\rho_\theta)}$, it is a *-algebra,} \\
	\textit{and it admits a conditional expectation $F$ such that $F^*\rho_\theta = \rho_\theta$}.
\end{gathered}
\end{equation}
Our first observation, using an argument by \Lbar{}uzak (cp.\ \cref{thm:luczak}), is:
\begin{equation}\label{introeq:PTP-sufficiency}
\begin{gathered}
	\textit{There is a minimal PTP-sufficient operator system $J_{(\rho_\theta)}$, it is a J*-algebra,} \\
    \textit{and it admits a conditional expectation $F$ such that $F^*\rho_\theta=\rho_\theta$}.
\end{gathered}
\end{equation}
Here, a J*-algebra means an operator system $J\subset L(\H)$ that is closed under the \emph{Jordan product}
\begin{align}
			\jp{a}{b} = \frac{1}{2}(ab+ba).
\end{align}
In \cref{sec:basics} we provide an introduction to J*-algebras. The most important difference to *-algebras is that the Jordan product is commutative, but not associative.\footnote{Jordan algebras recently also appeared in the context of hypothesis testing in general probabilistic theories \cite{sonoda_hypothesis_2025}.} Every *-algebra $A\subset L(\H)$ is of course also a J*-algebra. 
A UP map $F:L(\H)\to J$ is a conditional expectation onto $J$ if and only if $F^2 = F$. If $J$ is a *-algebra, then $F$ is automatically completely positive (see \cref{sec:multiplicative-properties} and \cref{sec:CEs}). It is called $(\rho_\theta)$-preserving (or state-preserving for $(\rho_\theta)$) if $F^* \rho_\theta = \rho_\theta$.

\begin{example*} Suppose $\rho = \tfrac{1}{2}(\1+X),\sigma= \tfrac12(\1+Z)$, where $X,Y,Z$ are the Pauli matrices. 
	The two states are clearly irreducible and their minimal sufficient *-algebra is $A_{(\rho,\sigma)} = L(\CC^2)$. 
	However, the two states are also symmetric under transpose and the map $E:x\mapsto \frac{1}{2}(x+x^t)$ is a $(\rho,\sigma)$-preserving conditional expectation onto the J*-algebra of symmetric matrices in $L(\CC^2)$. 
    The minimal sufficient J*-algebra is $J_{(\rho,\sigma)} = \{ x=x^t : x\in L(\CC^2) \} = \lin \{\1,X,Z\}$. 
\end{example*}

Why are the minimal sufficient (C)PTP-sufficient operator systems important for (C)PTP-equivalence?
Two statistical experiments are CPTP-equivalent if and only if the two minimal sufficient *-algebras are isomorphic and there exists an isomorphism that intertwines the expectation values \cite{galke_sufficiency_2024}. 
Our first main result is to establish the same result, but for PTP maps and J*-algebras:
\begin{introtheorem}[cp.\ \cref{thm:ptp-inter}]
    A faithful statistical experiment $(\rho_\theta)_{\theta\in\Theta}$ on $\H$ is PTP-equivalent to a faithful statistical experiment $(\omega_\theta)_{\theta\in\Theta}$ on $\K$ 
    if and only if there exists a J*-isomorphism $\psi: J_{(\omega_\theta)} \to J_{(\rho_\theta)}$ which intertwines the two families of states:
    \begin{equation}
        \tr(\rho_\theta\, \psi(a)) = \tr(\omega_\theta\, a), \qquad a\in J_{(\omega_\theta)},\ \theta\in\Theta.
    \end{equation}
	If $T:L(\K)\to L(\H)$ and $S:L(\H)\to L(\K)$ are interconverting UP maps, then $T|_{J_{(\omega_\theta)}} = \psi$ and $S|_{(\rho_\theta)} = \psi^{-1}$.
\end{introtheorem}
We illustrate the theorem using a further example. 
Let $(\rho,\sigma)$ be a dichotomy with $J_{(\rho,\sigma)}=L(\H)$ (such dichotomies are studied in \cref{sec:more}) such that \eqref{introeq:introexample} holds.
Set $\tilde\rho = \tfrac{1}{2}\rho\oplus\tfrac{1}{2}\rho^t$ and $\tilde\sigma = \tfrac{1}{2}\sigma\oplus\tfrac{1}{2}\sigma^t$. 
Set $\tilde\rho = \tfrac{1}{2}\rho\oplus\tfrac{1}{2}\rho^t$ and $\tilde\sigma = \tfrac{1}{2}\sigma\oplus\tfrac{1}{2}\sigma^t$. 
Then the minimal sufficient *-algebras are given by $A_{(\rho,\sigma)} = L(\H)$ and $A_{(\tilde \rho,\tilde \sigma)} = L(\H)\oplus L(\H)$, which are obviously not isomorphic.
However, we will see that the minimal sufficient J*-algebras fulfill
\begin{align}
	J_{(\tilde \rho,\tilde \sigma)} = \{ a\oplus a^t \ :\  a\in J_{(\rho,\sigma)} \} \Jcong J_{(\rho,\sigma)},
\end{align}
where the isomorphism is a J*-isomorphism that preserves expectation values with respect to the two dichotomies. Explicit interconversion maps are given in \cref{exa:final-trp-doubling}.

One may wonder whether more interesting situations may occur.
For dichotomies, we can show that, in a sense, the transpose is all there is.
Recall that a co-completely positive (coCP) map is a CP map followed by a transposition and that a decomposable map is a sum of a CP and a coCP map. 
\begin{introtheorem}[cp.~\cref{thm:decomposable-interconversion}]
    A dichotomy $(\rho,\sigma)$ on $\H$ and a dichotomy $(\hat\rho,\hat\sigma)$ on $\K$ are PTP-equivalent if and only if they are equivalent via decomposable, trace-preserving maps.
\end{introtheorem}

In \cref{sec:more}, we study further aspects of minimal sufficient J*-algebras of dichotomies.
\cref{sec:symmetries} investigates how the minimal sufficient J*-algebra relates to the existence of unitary and anti-unitary symmetries.
Here a (anti-) unitary symmetry of a dichotomy $(\rho,\sigma)$ is an (anti-) unitary $u$ such that $u\rho u^*=\rho$ and $u\sigma u^*=\sigma$. 
The absence of non-trivial unitary symmetries is reflected in the minimal sufficient *-algebra via 
\begin{align}
    \text{no unitary symmetries}\qquad &\iff \quad A_{(\rho,\sigma)}=L(\H).\quad
\intertext{In contrast, the minimal sufficient J*-algebra captures the absence of both kinds of symmetries:}
    \begin{tabular}{c}
         \text{no unitary and no} \\\text{anti-unitary symmetries}
    \end{tabular}\quad &\iff \quad J_{(\rho,\sigma)}=L(\H).\quad
\end{align}
\cref{sec:many-examples} studies which J*-algebras arise as the minimal sufficient J*-algebras of a dichotomy.
For J*-factors, i.e., J*-algebras with trivial center, this is the case if and only if they are generated by two hermitian elements. 
J*-algebras with this property are classified in \cref{sec:2-gen}.

\subsection{Bayes sufficiency and the minimal sufficient algebra}

We now connect (C)PTP-equivalence to binary hypothesis testing. 
Recall that the optimal tests for binary hypothesis testing with priors $(p,1-p)$ are given by the projectors $\proj{\rho > t\sigma}$ onto the positive part of $\rho - t\sigma$ with $t=(1-p)/p$. 
The projectors $\proj{\rho>t\sigma}$ are called Neyman-Pearson tests, in analogy to the classical case. 
One may wonder whether there is a connection between the Neyman-Pearson tests and minimal sufficient (J)*-algebras. 
The following theorem clarifies this. To state it we denote by
\begin{align}
			K_{(\rho,\sigma)} = \lin\{ \proj{\rho>t\sigma} \}_{t>0} +\CC \cdot \1
\end{align}
the \emph{minimal Bayes-sufficient} operator system for $(\rho,\sigma)$ (cp.\ \cref{sec:bayes}).

\begin{introtheorem}[cp.\ \cref{thm:K-generates}]\label{introthm:K-generates}
    Let $(\rho,\sigma)$ be a faithul dichotomy on $\H$. Then 
	\begin{align}
		K_{(\rho,\sigma)} \subset J_{(\rho,\sigma)} \subset A_{(\rho,\sigma)}
	\end{align}
	and both the minimal sufficient J*-algebra and the minimal sufficient *-algebra are generated by the Neyman-Pearson tests:
    \begin{align}
	    J_{(\rho,\sigma)} = \Jstaralg(K_{(\rho,\sigma)}),\quad 
	    A_{(\rho,\sigma)} = \staralg(K_{(\rho,\sigma)}).
    \end{align}
\end{introtheorem}
In fact the statement generalizes to arbitrary faithful statistical experiments $(\rho_\theta)$ in the following sense. 
If we define $\omega = \frac{1}{|\Theta|}\sum_\theta \rho_\theta$, and set 
\begin{align}
	K_{(\rho_\theta)} = \lin\{\proj{\rho_\theta > t \omega}\}_{t>0,\theta \in \Theta}  + \CC\cdot \1, 
\end{align}
then we also have 
\begin{align}
        J_{(\rho_\theta)} = \Jstaralg(K_{(\rho_\theta)}),\quad A_{(\rho_\theta)} = \staralg(K_{(\rho_\theta)}).
\end{align}
Part of the proof of \cref{introthm:K-generates} shows that for any statistical experiment $(\rho_\theta)$ there is a PTP-equivalent statistical experiment $(\hat\rho_\theta)$ such that $\hat\rho_\theta \in J_{(\rho_\theta)}$. In fact (cp.\ \cref{cor:TPCE-ptp-equivalence}), the $\hat\rho_\theta$, $\theta\in \Theta$, generate the minimal sufficient J*-algebra:
\begin{align}\label{introeq:states-generate-J}
		J_{(\rho_\theta)} = \Jstaralg((\hat\rho_\theta)_{\theta\in\Theta}).
\end{align}
Thus, we can always find a PTP-equivalent version of a statistical experiment, where the minimal sufficient J*-algebra is generated by the density matrices itself. In particular, the density matrices $\hat\rho_\theta$ may be expressed using the Neyman-Pearson tests. 

\subsection{Petz recovery and quantum divergences}
As mentioned above, sufficiency is a statement about the recoverability of a family of states from their marginals on a subalgebra. Petz gave a neat characterization of recoverability in terms of Connes cocycles $\rho^{-it}\sigma^{-it}$, as well as via the Petz recovery map \cite{petz_sufficient_1986,petz_sufficiency_1988,ohya_quantum_1993}.
Suppose $T^*$ is a PTP map $L(\H) \to L(\hat\H)$, $\sigma$ a faithful state on $\H$ and $\hat\sigma  = T^*\sigma$ a faithful state on $\hat\H$. Then the Petz recovery map of $T^*$ relative to $\sigma$ is the PTP map $R_{T,\sigma}^*: L(\hat\H) \to L(\H)$ given by (cp.\ \cref{sec:recovery})
\begin{align}
	R_{T,\sigma}^*(\hat \rho)  = \sigma^{\frac12} T(\hat\sigma^{-\frac12} \hat\rho \hat\sigma^{-\frac12}) \sigma^{\frac12}, 
\end{align}
which always fulfills $R_{T,\sigma}^*\hat\sigma = \sigma$. For the following theorem, we define
\begin{align}
			d_{\rho | \sigma} = \sigma^{-\frac12} \rho \sigma^{-\frac12},
\end{align}
which appeared in the literature before, see e.g. \cite{belavkinCstaralgebraicGeneralization1982,jencova_quantum_2010,liu_layer_2025}. The following theorem provides an algebraic characterization of recoverability.
\begin{introtheorem}[cp.\ \cref{thm:algebraic-recovery}]\label{introthm:algebraic-recovery}
    Let $(\rho,\sigma)$ be a dichotomy on $\H$ with $\sigma$ faithful. Let $T^*:L(\H)\to L(\K)$ be a PTP map and $T$ the corresponding UP map. The following are equivalent:
    \begin{enumerate}[(a)]
        \item There exists some PTP map $S^*:L(\K)\to L(\H)$ such that $S^*T^*\rho = \rho, S^*T^*\sigma=\sigma$.
        \item $T(d_{T^*\rho|T^*\sigma}) = d_{\rho|\sigma}$.
        \item $R_{T,\sigma}^* T^*\rho = \rho$, where $R_{T,\sigma}$ is the Petz recovery map relative to $T$ and $\sigma$.
        \item $T$ restricts to an isomorphism $J_{(T^*\rho,T^*\sigma)}\to J_{(\rho,\sigma)}$.
    \end{enumerate}
\end{introtheorem}

In the case of sufficient *-algebras, Petz also showed that there is a simple way to test whether the equivalent conditions of \cref{introthm:algebraic-recovery} hold:
This is the case if and only if the quantum relative entropy remains invariant. The following theorem generalizes this statement to PTP maps.
We state it using the \emph{Hockey-stick divergences} \cite{hirche_quantum_2023,hirche_quantum_2024} 
\begin{align}
E_t(\rho\|\sigma) = \tr((\rho-t\sigma)^+) - (1-t)^+,
\end{align}
where $x^+$ denotes the positive part of a hermitian operator (and $\lambda^+ =\max \{0,\lambda\}$ for $\lambda\in \RR$). Moreover, we write $\rho\ll\sigma$ if the supporting subspace of $\rho$ is contained in that of $\sigma$.
\begin{introtheorem}[cp.\ \cref{thm:petz-sufficiency}]\label{introthm:petz-sufficiency}
    Let $T^*: L(\H) \to L(\K)$ be PTP map and let $\rho\ll\sigma$ be states on $\H$.
    The following are equivalent:
    \begin{enumerate}[(a)]
        \item $D(T^*\rho\|T^*\sigma) = D(\rho\|\sigma)$,
        \item $E_t(T^*\rho\|T^*\sigma) = E_t(\rho\|\sigma)$ for all $t\ge0$,
        \item $(\rho,\sigma)$ can be recovered from $(T^*\rho,T^*\sigma)$ with a PTP map $R^*:L(\K)\to L(\H)$.
    \end{enumerate}
\end{introtheorem}
To prove \cref{introthm:petz-sufficiency}, we make use of Petz's theorem for CPTP maps, so our result does not yield an independent proof of Petz's result. 
However, our techniques open avenues to an independent proof. For example, an independent proof would immediately follow from a direct argument showing
\begin{equation}
    d_{\rho|\sigma} \in \Jstaralg(K_{(\rho,\sigma)}).
\end{equation}
This could, for instance, be established through an operator layer cake formula for $d_{\rho|\sigma}$, similar the operator layer cake formula for the logarithmic derivative \cite{liu_layer_2025} (see \cref{remark:layer-cake} for a detailed discussion).


Recently, a family of $f$-divergences was introduced using Hockey-stick divergences \cite{hirche_quantum_2024}, defined via
\begin{align}\label{introeq:hockey-stick-f-div}
	D_f(\rho\|\sigma) &= \int_1^\infty \big(f''(t) E_t(\rho\|\sigma) + t^{-3} f''(1/t) E_t(\sigma\|\rho)\big)\,dt, \\
	&= f(0) + \int_0^\infty f'(t) \tr(\rho \proj{\rho>t\sigma}) \,dt \label{introeq:layercake} 
\end{align}
where $f:(0,\infty)\to \RR$ is a convex, twice-differentiable function with $f(1) = 0$ and the second line is the layer-cake representation from \cite{liu_layer_2025}. The relative entropy is recovered for $f(x)= x\log x$.  
It is immediate from the proof of \cref{thm:petz-sufficiency} that \cref{introthm:petz-sufficiency} holds for any such $f$-divergence instead of the relative entropy as long as $f''(t)>0$ for all $t\geq 1$.

A consequence of \eqref{introeq:layercake} is that the quantum relative entropy can be computed from the restriction of the states $\rho,\sigma$ to the Bayes-sufficient operator system $K_{(\rho,\sigma)}$.
In fact, as far as we know, all quantum divergences which satisfy the data-processing inequality for PTP maps can be computed \emph{within} $J_{(\rho,\sigma)}$ in a sense that we explain now. 
Take, for example, the 
$\alpha$-$z$ quantum R\'enyi divergence \cite{audenaertAzRenyiRelativeEntropies2015}, defined as
\begin{align}
	D_{\alpha,z}(\rho\|\sigma) = \frac{1}{\alpha-1} \log\tr \Big(\rho^{\frac{\alpha}{2z}} \sigma^{\frac{1-\alpha}{z}} \rho^{\frac{\alpha}{2z}}\Big)^z.\label{introeq:alpha-z}
\end{align}
Setting $z=\alpha$ yields the \emph{sandwiched} quantum R\'enyi divergence $\tilde D_\alpha$ \cite{muller-lennertQuantumRenyiEntropies2013,wildeStrongConverseClassical2014}, while setting $z=1$, one obtains the usual \emph{Petz} quantum R\'enyi divergence $D_\alpha$ \cite{petzQuasientropiesStatesNeumann1985,petzQuasientropiesFiniteQuantum1986}.
Set $\omega = \frac12(\rho+\sigma)$. We can rewrite $D_{\alpha,z}$ as (assuming, without loss of generality, that $\omega$ is invertible) 
\begin{equation}
  D_{\alpha,z}(\rho\|\sigma) = \frac{1}{\alpha-1} \log \tr(\omega d^{(\alpha,z)}_{\rho\|\sigma}),\qquad 
  d^{(\alpha,z)}_{\rho\|\sigma} = \omega^{-\frac12} \Big(\rho^{\frac{\alpha}{2z}} \sigma^{\frac{1-\alpha}{z}} \rho^{\frac{\alpha}{2z}}\Big)^z \omega^{-\frac12}.
\end{equation}
As we show (cp.~\cref{lem:Lp-properties}), $d^{(\alpha,z)}_{\rho|\sigma}$ is contained in the minimal sufficient J*-algebra $J_{(\rho,\sigma)}$. 
Indeed, this is true for any operator that is a sum of symmetrized products of $\rho$ and $\sigma$ with total exponent $0$. 
Since J*-algebras are closed under functional calculus, we may, in addition, apply functions to such symmetrized products.

Building upon work by Jen\v{c}ov\'{a} \cite{jencovaPreservationQuantumRenyi2017,jencova_renyi_2018,jencova_renyi_2021}, and Hiai and Jen\v{c}ov\'{a} \cite{hiaiazRenyiDivergences2024}, we show:
\begin{introtheorem}[simplified, see Theorems \ref{thm:sandwiched-sufficiency}, \ref{thm:alpha-z-smaller-1}, and \ref{thm:alpha-z-bigger-1}] Let $\rho,\sigma$ be states on $L(\H)$ with full rank, $T^*:L(\H)\to L(\K)$ be a PTP map. The following are equivalent:
\begin{enumerate}[(a)]
    \item  $\tilde D_\alpha(\rho\|\sigma) = \tilde D_\alpha(T^*\rho\|T^*\sigma)$ for some $\alpha \in (\frac12,1)\cup(1,\infty)$.
    \item $D_{\alpha,z}(\rho\|\sigma) = D_{\alpha,z}(T^*\rho\|T^*\sigma)$ for $\alpha\in (0,1)$ and $\max\{\alpha,1-\alpha\} < z$.
    \item $D_{\alpha,z}(\rho\|\sigma) = D_{\alpha,z}(T^*\rho\|T^*\sigma)$ for $\alpha\in (1,\infty)$ and $\max\{\frac\alpha2,\alpha-1\} \leq z \leq \alpha < z+1$.
    \item $(\rho,\sigma)$ can be recovered from $(T^*\rho,T^*\sigma)$ with a PTP map $R^*:L(\K)\to L(\H)$. 
\end{enumerate}
\end{introtheorem}

These observations lend support to the conjecture in \cite{galke_sufficiency_2024} that equality of a sufficiently large set of quantum divergences implies PTP-equivalence, as is true in the commuting case, but wrong when asking for CPTP-equivalence.

\subsection{Beyond finite dimensions}

Our proofs are restricted to finite dimensions, but we believe that most of our results extend to the setting of von Neumann algebras. In fact, we believe that they generalize to the setting of JW*-algebras, the Jordan analog of von Neumann algebras \cite{hanche-olsen_jordan_1984}. 
In particular, it will be interesting to prove the sufficiency result for the quantum relative entropy in this setting. 
We plan to revisit this problem in a future publication. 
A core ingredient for our proof is Frenkel's integral formula, which has not yet been generalized to von Neumann algebras. 
The Hockey stick divergence $E_t(\omega\|\varphi)$ naturally extends to the von Neumann algebraic setting. In \cref{app:frenkel}, we prove:

\begin{introproposition}[see \cref{prop:frenkel-vna}]
    Let $M$ be an approximately finite-dimensional von Neumann algebra, and let $\omega, \varphi$ be normal states on $M$.
    Then
    \begin{equation}
        D(\omega \| \varphi) = \int_1^\oo \left(\frac1t E_t(\omega\|\varphi) + \frac1{t^2} E_t(\varphi\|\omega) \right) dt.
    \end{equation}
\end{introproposition}
In the case $M = L(\H)$ with $\dim \H=\oo$, the statement was shown in \cite{jencova_recoverability_2024}, albeit with a more complicated proof. 
A proof of Frenkel's integral formula for states on general von Neumann algebras will appear in upcoming work \cite{gandalfundco,gereonundhao-chung}.

\subsection{Discussion and Outlook}
Our results support the idea \cite{buscemiComparisonQuantumStatistical2012a,galke_sufficiency_2024} that CPTP maps do not provide the right mathematical framework for thinking about sufficiency and distinguishability of quantum states, even though they certainly provide the right framework for talking about the possible physical processes in quantum theory. 

There are obvious questions left open by our results apart from those mentioned above already. 
Most importantly, in the CPTP-case it is known that \emph{approximate} equality in the data-processing inequality implies \emph{approximate} recoverability \cite{fawziQuantumConditionalMutual2015,wildeRecoverabilityQuantumInformation2015,sutterStrengthenedMonotonicityRelative2016,jungeUniversalRecoveryMaps2018,faulknerApproximateRecoveryRelative2020,faulknerApproximateRecoverabilityRelative2022}, which can be quantified in terms of the fidelity $F(\rho,\sigma) = \norm{\rho^{\frac12}\sigma^{\frac12}}_1$. 
Since the fidelity can be computed on the level of the minimal sufficient J*-algebra (see \cref{introeq:alpha-z} for $\alpha=z=\frac12$), it is only reasonable to expect that the approximate recovery result also generalizes to the PTP setting.  

We have shown that the minimal sufficient J*-algebras of dichotomies are 2-generated and that every 2-generated J*-factor is minimal sufficient for some dichotomy. 
We believe that the latter statement should hold for general J*-algebras.
\\ 

\paragraph{Note added.} 
In the independent work \cite{yamagata_quantum_2026}, the author also considers sufficiency from the point of view of positive maps and Jordan algebras.

\paragraph{Acknowledgements.} We would like to thank Wolfram Bauer, Hao-Chung Cheng, Markus Frembs, Christoph Hirche, Anna Jen\v{c}ov\'{a}, Alexander M\"uller-Hermes, René Schwonnek, and Ole Skodda, for helpful discussions. 
In particular, we thank Anna Jen\v{c}ov\'{a} for suggesting \cref{cor:J-sufficient-CE} and the corrected proof of \cref{cor:general-suff-generated-neyman-pearson}.
We thank Konrad Szyma\'nski for suggesting the argument in Remark \ref{remark:layer-cake}.

Research at Perimeter Institute and the University of Waterloo is supported in part by the Government of Canada through the Department of Innovation, Science and Economic Development and by the Province of Ontario through the Ministry of Colleges and Universities.
This work was supported by the Swedish Research Council under grant no.\ 2021-06594 while the first-named author was at Institut Mittag-Leffler in Djursholm, Sweden, during the 2026 program on Operator Algebras and Quantum Information.

\clearpage

\section{Preliminaries}\label{sec:maths}

\subsection{Basics of J*-algebras}\label{sec:basics}

If $a_1,\ldots a_n$ are operators on a Hilbert space $\H$, we denote their symmetrized product by
\begin{equation}\label{eq:symmetrized-product}
    \trip {a_1}{\ldots}{a_n} = \frac12 \big(a_1a_2\cdots a_n + a_n a_{n-1}\cdots a_1\big)
\end{equation}
While this notation is easy to confuse with that for sets, in the following, it will be clear from the context what is meant.
In the case $n=2$, this gives (half) the anti-commutator, better known as the \emph{Jordan product},
\begin{equation*}
    \jp ab = \frac12 (ab+ba).
\end{equation*}
The Jordan product is commutative, but not associative.
It is fully determined by squares in the sense that\footnote{This follows from adding the equations $(a+b)^2 = a^2 + ab+ba + b^2$ and $-(a-b)^2 = -a^2 + ab+ba - b^2$.}
\begin{equation}\label{eq:square trick}
    \jp ab = \frac14\big( (a+b)^2 - (a-b)^2 \big).
\end{equation}

\begin{definition}\label{def:JSA}
    An \emph{operator system} on a Hilbert space $\H$ is a complex subspace $J\subset L(\H)$, which is is *-invariant and contains the identity.
    A \emph{J*-algebra} is an operator system that is closed under Jordan products.%
    \footnote{\label{foot1} Our notion of a J*-algebra agrees with that of a JC*-algebra or a JW*-algebra on a finite-dimensional Hilbert space (in general, JC*- and JW*-algebra differ in their topological properties similar to C*- and W*-algebras, but the difference disappear in finite-dimensions) \cite{hanche-olsen_jordan_1984}.}
\end{definition}

Thus, a unital *-invariant subspace $J$ is a J*-algebra if and only if it is closed under squares.
In fact, it suffices to check squares of hermitian elements.
Clearly, every *-algebra $A$ on $\H$ is a J*-algebra on $\H$.
In particular, this includes $L(\H)$.
At the end of this subsection, we give examples of J*-algebras that are not *-algebras.

The structure of a J*-algebra is encoded in its hermitian part. 
An operator system $J\subset L(\H)$ is a J*-algebra if and only if its hermitian part $J_h=\{x\in J : x=x^*\}$ is a Jordan algebra of hermitian operators, i.e., a unital real vector space of hermitian operators that is closed under the Jordan product.\footnote{In operator algebra lingo, a Jordan algebra of hermitian operators on a Hilbert space is a concretely represented JC- or a JW-algebra \cite{hanche-olsen_jordan_1984}.}

If $a= a^*$ is a hermitian element in a J*-algeba $J$, then $a^n \in J$ for all $n\in \NN$.
Indeed, this follows inductively because $a^n = \jp{a^{n-1}}a$ and $a\in J$.
Therefore, J*-algebras are closed under the functional calculus.
If $a = \sum_i \lambda_i p_i$ is the spectral decomposition, then $f(a) = \sum_i f(\lambda_i) p_i$ is in $J$ for all functions $f:\Sp(a)\to \CC$.
In particular, the spectral projections lie in $J$.
Thus, a J*-algebra is spanned by its projections
\begin{equation}\label{eq:spanned-by-projections}
    J = \lin \{ p \in J \ :\  p^2 =p = p^*\}
\end{equation}

A J*-algebra $J$ on a Hilbert space $\H$ inherits a positive cone $J^+ = \{ a\in J : a\ge0 \}$ from the positive semi-definite order on $L(\H)$.
It can be characterized algebraically via
\begin{equation}\label{eq:positive-cone-squares}
    J^+ = \{ a^2 \ : \ a=a^*\in J\}.
\end{equation}
Moreover, the positive cone $J^+$ is self-dual with respect to the trace: If $a\in J$, then%
\footnote{To see this, note $\Re(a),\Im(a)\in J$. 
If $\Re(a)=\sum_i \lambda_i p_i$ is the spectral decomposition, then $0\le \tr(a p_i) = \tr(\Re(a)p_i) = \lambda_i$ shows $\Re(a)\ge0$.
A similar argument shows $\Im(a) =0$.}
\begin{equation}\label{eq:self-dual-cone}
    a \in J^+ \qquad \iff\qquad  \tr(ab) \ge0 \quad \forall b\in J^+.
\end{equation}

If $S \subset L(\H)$ is a collection of hermitian operators, we denote by $\Jstaralg(S)$ the linear hull of the identity and elements that can be written as nested Jordan products of elements of $S$.
It can be checked that $\Jstaralg(S)$ is a J*-algebra. In fact, it is the smallest J*-algebra containing $S$:
\begin{equation}
    \Jstaralg(S) = \bigcap \, \{ J \subset L(\H)\ \text{J*-algebra containing $S$.} \}.
\end{equation}
If $S = \{a_1,\ldots, a_n\}$ is a finite set, we write $\Jstaralg(a_1,\ldots,a_n)$ as a short hand.

A unital linear map $T: J_1 \to J_2$ between J*-algebras is a \emph{J*-homomorphism} if it is unital, preserves the adjoint, and the Jordan product.
A J*-homomorphism is necessarily positive.
By \eqref{eq:square trick}, to check whether a hermitian-preserving map $T$ is a J*-homomorphism, it suffices to check that 
$T(a^2) = T(a)^2$ for all $a\in J$.
A \emph{J*-isomorphism} is a bijective J*-homomorphism (the inverse is automatically a J*-homomorphism).
We write $J_1\Jcong J_2$ if a J*-isomorphism $T:J_1\to J_2$.

A useful identity is the \emph{Jordan triple product} formula, which expresses the symmetrized triple product in terms of Jordan products:
\begin{equation}
\label{eq:triple-product}
 \jp{\jp ab}c + \jp{\jp bc}a - \jp{\jp ac} b= \frac{1}{2}(abc + cba) = \trip abc,\qquad a,b,c\in L(\H).
\end{equation}
This shows that J*-algebras are closed under the triple product. 
In particular, a J*-algebra $J$ contains the product $aba=\trip aba$ for all $a,b\in J$.
There exist no corresponding formulae that express symmetrized products of four or more operators in terms of the Jordan product.
A J*-algebra $J\subset L(\H)$ that is closed under higher-order symmetrized products is called \emph{reversible}.
Reversible J*-algebras will play a crucial role in this project.

Next, we consider two classes of examples of J*-algebras.
We will return to these examples on several occasions.

\begin{example}[fixed-points of involutions]\label{exa:fixed-pts}
    Let $A$ be a *-algebra on $\H$ and let $\vartheta$ be an involution on $A$, i.e., a *-antiautomorphism with $\vartheta^2=\id$.
    Then the fixed-point space
    \begin{equation}
       A^\vartheta := \mathrm{Fix}(\vartheta) := \big\{ x \in A  \ : \ \vartheta(x)=x \big\}
    \end{equation}
    is a J*-algebra, in fact, a reversible J*-algebra, on $\H$.%
    \footnote{Indeed, if $\vartheta(a_i)=a_i$, $i=1,\ldots, n$, then 
    \begin{align*}
        \vartheta\big(2\trip{a_1}\ldots{a_n}\big) = \vartheta(a_1\cdots a_n)+\vartheta(a_n \cdots a_1) &= \vartheta(a_n)\cdots\vartheta(a_1)+\vartheta(a_1)\cdots \vartheta(a_n) 
        = 2\trip{a_1}\ldots{a_n}.
    \end{align*}}
    
    We consider concrete examples:
    \begin{enumerate}
	    \item\label{it:fixed-pts1}
            The transposition $t: a \mapsto a^t$ in the standard basis is an involution on $A=L(\CC^n)$.
            The fixed-point J*-algebra is given by the symmetric matrices:
            \begin{equation}\label{eq:symmetric-matrices}
                L(\CC^n)^t = \{ x\in M_n(\CC) \ :\ x= x^t\}.
            \end{equation}
        \item\label{it:fixed-pts2}
		In even dimensions, the direct sum decomposition $\CC^{2n}=\CC^n\oplus\CC^n$, gives rise to a \emph{symplectic involution} on $L(\CC^{2n})$, defined as 
            \begin{equation}\label{eq:symplectic-invo}
                \beta: a \mapsto - \Omega a^t \Omega,
            \end{equation}
            where $\Omega$ is the symplectic matrix, i.e., $\Omega(\xi\oplus\eta) = (-\eta)\oplus \xi$.
            The fixed-point J*-algebra is 
            \begin{equation}\label{eq:symplectic-fixed-pts}
                L(\CC^{2n})^\beta = \left\{ \begin{pmatrix} a_{11} & a_{12} \\ a_{21} & a_{22} \end{pmatrix}
                 \ : \ a_{11} = a_{22}^t,\ a_{12} = - a_{12}^t,\ a_{21}^t = -a_{21}^t \right\}
            \end{equation}
         
    \end{enumerate}
    
\end{example}

\begin{example}[Spin factors]\label{exa:spin-factors}
    Consider a family $s_1, \ldots, s_m$ of hermitian untaries on a Hilbert space $\H$, which are pairwise anti-commuting, i.e., satisfy $\jp{s_j}{s_k} = \delta_{ij}$.
    Then 
    \begin{equation}
        J = \lin \,\{\1,s_1,\ldots, s_m\}
    \end{equation}
    is a J*-algebra. J*-algebras of this form are called \emph{spin factors} (see \cite[Ch.~6]{hanche-olsen_jordan_1984}). 
    
    Let us consider an explicit example on the $n$-qubit Hilbert space $\H = (\CC^2)^{\ox n}$.
    Denoting the Pauli matrices on $\CC^2$ by $X,Y,Z$, we define the Majorana operators
    \begin{equation}\label{eq:spin-factors1}
        s_{2k+1} = Y^{\ox k} \ox X \ox \1^{\ox n-k-1}, \qquad s_{2k} = Y^{\ox k}\ox Z \ox \1^{\ox n-k-1},
    \end{equation}
    for $k=0,\ldots,n-1$.
    For instance, $s_1 = X\ox \1^{\ox n-1}$, $s_2 = Y\ox X\ox \1^{n-2}$, $s_3 = Y\ox Z \ox\1^{n-2}$, and $s_{2n} = Y^{\ox n-1}\ox Z$.
    Then the hermitian unitaries $s_1,\ldots, s_{2n}$ are pairwise anti-commuting.
    Hence, we obtain spin factors
    \begin{equation}\label{eq:spin-factors2}
        V_{2n-1} = \lin \{ \1, s_1,\ldots, s_{2n-1} \}, \qquad V_{2n} =\lin \{ \1, s_1,\ldots, s_{2n} \}.
    \end{equation}
    This construction defines a spin factor $V_n$ for each $n\in\NN$.
    A general spin factor $J$, i.e., the span of a family of pairwise anti-commuting hermitian unitaries, is J*-isomorphic to (exactly) one of these (namely the one with $n =$ the number of non-trivial generators) \cite[Prop.~6.1.5]{hanche-olsen_jordan_1984}.
\end{example}

In \cref{sec:structure-theory}, we will discuss that a general J*-algebra is J*-isomorphic to a direct sum of the examples discussed in this section.

\subsection{Multiplicative properties of positive maps}\label{sec:multiplicative-properties}

In this section we discuss multiplicative properties of positive maps. 
We say that a UP map $T$ admits a faithful invariant state if $T^*\sigma = \sigma$ for some faithful state $\sigma$. 

In contrast to completely positive, unital maps, general UP maps do not fulfill the Kadison-Schwarz inequality. However, they fulfill the Jordan-Schwarz inequality:
\begin{lemma}[Jordan-Schwarz inequality {\cite[Lem.~7.3]{stormer_positive_1963}}]\label{lem:JS-inequality}
     Let $T:L(\H)\to L(\K)$ be a unital, positive map. Then
    \begin{equation}\label{eq:JS-inequality}
        \jp{Ta^*}{Ta} \le T\jp{a^*}a , \qquad a\in L(\H).
    \end{equation}
\end{lemma}
As in the case of a completely positive map, we collect the elements that saturate the Jordan-Schwarz in the \emph{multiplicative domain}:
\begin{definition}[{\cite{stormer_multiplicative_2007}}]
    The \emph{multiplicative domain}\footnote{In the literature, the term multiplicative domain is often reserved for UCP maps. The multiplicative domain of a UP map is called the "definite set" in \cite{stormer_multiplicative_2007,stormer_positive_2013}.} $\mathrm{MD}(T)$ of a UP map $T:L(\H)\to L(\K)$ is defined as
    \begin{equation}
        \mathrm{MD}(T)=\big\{a\in L(\H)\ :\ \jp{Ta^*}{Ta} =T\jp{a^*}a \,\big\}.
    \end{equation}
\end{definition}

If $T:L(\H)\to L(\K)$ is a UCP map, then $\mathrm{MD}(T)$ coincides with the usual multiplicative domain\footnote{By the Scwartz inequality, $T\jp{a^*}a = \jp{Ta^*}{Ta}$ implies both $T(a^*a) = T(a)^*T(a)$ and $T(aa^*) = T(a)T(a)^*$ for UCP maps.}, which is a *-subalgebra of $L(\H)$ (see \cite[Prop.~1.5.7]{brown_ozawa} or \cite[Thm.~5.7]{wolf_quantum_2012}).
The analogous statement holds for the Jordan product if $T$ is merely a UP map:

\begin{proposition}[{\cite[Prop.~2.1.7]{stormer_positive_2013}}]\label{prop:multiplicative_dom}
    Let $T:L(\H)\to L(\K)$ be a UP map.
    Then $\mathrm{MD}(T)$ is a J*-algebra and $T|_{\mathrm{MD}(T)}$ is a J*-homomorphism.
    If $a\in \mathrm{MD}(T)$ and $b\in L(\H)$, then
    \begin{equation}\label{eq:homo} 
        \jp{Ta}{Tb}=T\jp ab, \qquad T(aba) = T(a)T(b)T(a).
    \end{equation}
\end{proposition}%
Note that the right-hand side in \eqref{eq:homo} follows from the left-hand side since $aba = 2\jp{\jp ab}a - \jp{a^2}b$. The following simple corollary will prove to be extremely important for us.
\begin{corollary}\label{cor:projections-homo}
    Let $T:L(\H)\to L(\K)$ be a UP map and let $\{p_i\}_{i\in I}$ be a family of projections on $\H$. Suppose that $T$ maps each $p_i$ to a projection on $\K$.
    Then $T$ restricts to a surjective J*-homomorphism
    \begin{align}
        T : \Jstaralg((p_i)_{i\in I}) \to \Jstaralg((Tp_i)_{i\in I}).
    \end{align}
\end{corollary}
\begin{proof}
    $T(p_i)^2 = T(p_i) = T(p_i^2)$ shows $p_i\in \mathrm{MD}(T)$, which implies $\Jstaralg((p_i)_{i\in I})\subset \mathrm{MD}(T)$. 
    Thus, the restriction of $T$ to $\Jstaralg((p_i)_{i\in I})$ is a J*-homomorphism whose range must be $\Jstaralg((Tp_i))_{i\in I}$.
\end{proof}

Next, we consider fixed-point spaces of UP maps $T:L(\H)\to L(\H)$:
\begin{align}
    \mathrm{Fix}(T) := \{ a\in L(\H) : Ta = a \}.
\end{align}
We say that a state $\sigma$ is invariant under a UP map $T:L(\H)\to L(\H)$ or that $T$ is $\sigma$-preserving, if $T^*\sigma=\sigma$, i.e., if $\sigma$ is a fixed-point of the channel in the Schrödinger picture.

\begin{proposition}[{\cite[]{stormer_positive_2013}}]\label{prop:fixpoint}
    Let $T:L(\H)\to L(\H)$ be a UP map with a faithful invariant state $\sigma$.
    Then the fixed-point space $\mathrm{Fix}(T)$ is a J*-subalgebra of the multiplicative domain $\mathrm{MD}(T)$. 
\end{proposition}
\begin{proof}
    Let $a\in \mathrm{Fix}(T)$. From the Jordan-Schwarz inequality we have $T\jp{a^*}a-\jp{Ta^*}{Ta} = T\jp{a^*}a - \jp{a^*}a\ge0$.
    But
    \begin{align}
        0\leq \tr(\sigma\,(T\jp{a^*}a - \jp{a^*}a)) = \tr(\sigma(\jp{a^*}a - \jp{a^*}a)) = 0.
    \end{align}
    Since $\sigma$ is faithful, we find $T\jp{a^*}a = \jp{Ta^*}{Ta} = \jp{a^*}a$. Thus, $\mathrm{Fix}(T)$ is a subset of $\mathrm{MD}(T)$, and it follows from \cref{prop:multiplicative_dom} that it is also closed under Jordan products.
\end{proof}

Next, we need the fact that the Cesaro means of a UP map converge to a UP projection onto the fixed-point space.
The strongest version of this statement was obtained in the master's thesis \cite[Prop.~3.2]{rauber}.%
\footnote{The slightly weaker version of the statement in the proof of \cite[Thm.~2.2.11]{stormer_positive_2013}, which asserts convergence only along a subnet (but works for general von Neumann algebras), will be sufficient for our purposes.
Alternatively, the statement can be shown by dualizing the corresponding statement for PTP maps in \cite[Prop.~6.3]{wolf_quantum_2012}.}

\begin{lemma}[{\cite[Prop.~3.2]{rauber}, see also \cite[Thm.~2.2.11]{stormer_positive_2013} and \cite[Prop.~6.3]{wolf_quantum_2012}}]\label{lem:cesaro-mean}
    Let $T:L(\H) \to L(\H)$ be a UP map. Then
    \begin{align}
        P = \lim_{n\to \infty} \frac{1}{n}\sum_{k=1}^n T^n
    \end{align}
    converges pointwise and
    defines an idempotent UP map onto the fixed-point space $\mathrm{Fix}(T)$. 
\end{lemma}

\subsection{The KMS inner product and Woronowicz's maximum principle}
\label{sec:KMS}

It will be useful to consider UP maps as operators on a Hilbert space. 

Consider a faithful state $\sigma$ on $\H$. Then we can turn $L(\H)$ into a Hilbert-space using the KMS inner product%
\footnote{Note that $\ip{a}{b}_\sigma = \langle \Omega_\sigma| \bar a\otimes b |\Omega_\sigma\rangle$, where $\Omega_\sigma$ is the canonical purification of $\sigma$ on $\bar \H\otimes \H$.}
\begin{align}
    \ip{a}{b}_\sigma = \tr(\sigma^{\frac12} a^* \sigma^{\frac12} b). 
\end{align}
The KMS inner product is useful because it respects the order interval of the reference state $\sigma$. 
To see this, we define $\Gamma_\sigma =\sigma^{\frac12}(\placeholder)\sigma^{\frac12}: L(\H)\to L(\H)$.
Then $\Gamma_\sigma$ is an isomorphism between the order intervals $[0,\sigma]_{L(\H)}$ and $[0,\1]_{L(\H)}$, where $[0,a]_{L(\H)}= \{ b \in L(\H) : 0\leq b\leq a\}$, and we have
\begin{equation}
    \ip ab_\sigma = \tr(a^* \Gamma_\sigma(b)) =\tr(\Gamma_\sigma(a^*)b), \qquad a,b\in L(\H).
\end{equation}
From the above discussion, it is clear that $\ip{a}{b}_\sigma \geq 0$ whenever $a,b\geq 0$. 
Woronowicz's maximum principle shows that the above properties uniquely single out $\ip{a}{b}_\sigma$:
\begin{theorem}[{\cite[Thms.~1.1 \& 1.2]{woronowicz_selfpolar_1974}}]\label{thm:woronowicz} 
    Let $s:L(\H)\times L(\H) \to \CC$ be a sesquilinear form such that (i) $s(a,a)\geq 0$ for all $a\in L(\H)$, (ii) if $a,b\geq 0$, then $s(a,b)\geq 0$, and (iii) $s(\1,a) = \tr(\sigma a)$. Then
    \begin{align}
        s(a,a) \leq \ip{a}{a}_\sigma, \qquad a\in L(\H). 
    \end{align}
    If, in addition, for every $\omega\in [0,\sigma]$, there is an $a\in L(\H)$ with $\tr(\omega b)= s(a,b)$ for all $b\in L(\H)$, then $s(a,b) = \ip ab_\sigma$.
\end{theorem}

We note that Woronowicz's theorem holds in much more generality than stated here (see also \cite{haagerup_tomita-takesaki_1984}). 
A direct consequence is that UP maps are contractions relative to the KMS inner product:

\begin{corollary}\label{cor:UP-contraction}
    Let $\sigma, \hat\sigma$ be faithful states on $\H$ and $\hat \H$, respectively, and $T:L(\hat \H) \to L(\H)$ be a UP map such that $T^*\sigma = \hat \sigma$.
    Then 
    \begin{align}
        \ip{Ta}{Ta}_\sigma \leq \ip{a}{a}_{\hat\sigma}, \qquad a\in L(\hat\H).
    \end{align} 
\end{corollary}
\begin{proof}
    Set $s(a,b) = \ip{Ta}{Tb}_\sigma$. Then the properties required for \cref{thm:woronowicz} are fulfilled since $s(\1,a) = \tr(\sigma Ta) = \tr(\hat \sigma a)$. 
\end{proof}

\subsection{Conditional expectations. I}
\label{sec:CEs}

Consider a UP map $E:L(\H)\to L(\H)$ that is idempotent, i.e., $E^2 = E$, and has a faithful invariant state $\sigma$: $E^*\sigma = \sigma$.
Since $E$ is idempotent, its range $J = \mathrm{Ran}(T)$ coincides with its fixed-point space $\mathrm{Fix}(E)$. 
Thus, by \cref{prop:fixpoint}, the range is a J*-algebra.
Generally, we define:

\begin{definition}
    Let $J \subset L(\H)$ be a J*-algebra.
    A \emph{conditional expectation} onto $J$ is an idempotent UP map $E:L(\H)\to J$ with $J = \mathrm{Ran}(E)$. If $\sigma$ is a state on $\H$, we call $E$ $\sigma$-preserving and $\sigma$ $E$-invariant if $E^*\sigma = \sigma$.
\end{definition}

Clearly, for a conditional expectation $E:L(\H)\to J$, the fixed-point space $J = \mathrm{Fix}(E)$ is contained in the multiplicative domain. 
It then follows from \cref{prop:multiplicative_dom} that 
\begin{align}\label{eq:bimodule-CE}
    E\jp ab = \jp{Ea}b ,\quad E(bab) = bE(a)b, \quad a\in L(\H),\  b\in J. 
\end{align}
If $E$ admits a faithful invariant state $\sigma$, then it must be faithful.
Indeed, if $a\geq 0$ and $Ea=0$, then $0 = \tr(\sigma Ea) = \tr(\sigma a)$ gives $a=0$.
Conversely, a faithful conditional expectation $E$ admits a faithful invariant state. 
Since we assume the Hilbert space to be finite-dimensional, we can simply choose $\sigma = E^*\tau$, where $\tau = \frac{\1}{\tr(\1)}$ denotes the maximally mixed state.

A conditional expectation is called \emph{trace-preserving} if $\tr(Ea) = \tr(a)$, $a\in L(\H)$.
This is the case if and only if $E^*\1=\1$, or, equivalently, if $E$ is $\tau$-preserving.

\begin{example}[continues=exa:fixed-pts]
    A trace-preserving conditional expectation onto the fixed-point J*-algebra $A^\vartheta$ is given by
    \begin{equation}
        E = \frac12 (\id +\, \vartheta) : A \to A^\theta.
    \end{equation}
\end{example}

Trace-preserving conditional expectations onto J*-algebras always exist (see \cite{edwards_conditional_1986} for an infinite-dimensional version):
\begin{lemma}
    If $J$ is a J*-algebra on $\H$, then there exists a unique trace-preserving conditional expectation $E$ onto it.
    For $a\in L(\H)$, $E(a)$ is the element of $J$ defined by
    \begin{equation}\label{eq:TPCE}
        \tr(E(a)b) = \tr(ab), \qquad b\in J.
    \end{equation}
    If $J_0\subset J$ is a J*-subalgebra with trace-preserving conditional expectation $E_0$, then
    \begin{equation}\label{eq:TPCE-factorizes}
        E_0 \circ E = E_0.
    \end{equation}
\end{lemma}
\begin{proof}
    We define $E$ through \eqref{eq:TPCE}.
    This is well-defined because the trace induces an inner product on $J$.
    By definition, $E$ is a linear map $E:L(\H)\to J$ with $E |_J = \id_J$.
    We check positivity. If $0\le a\in L(\H)$, then \eqref{eq:TPCE} gives $\tr(E(a)b) \ge0$ for $0\le b\in J$.
    By \eqref{eq:self-dual-cone}, this entails $E(a)\ge0$.
    Therefore, $E$ is a conditional expectation.
    We check uniqueness. If $F$ is a trace-preserving conditional expectation onto $J$, $a\in L(\H)$, and $b\in J$, then \eqref{eq:bimodule-CE} implies $\tr(F(a)b) = \tr(\jp{Fa}b) = \tr(F\jp ab) = \tr(\jp ab) = \tr(ab) = \tr(E(a)b)$.
    Hence, we have $E=F$.
    It remains to check \eqref{eq:TPCE-factorizes}.
    If $a\in L(\H)$, then  $\tr( E_0(E(a))b) = \tr(E(a)b) = \tr(ab)=\tr(E_0(a)b)$ for all $b\in J_0$.
\end{proof}
We will later frequently use the fact that trace-preserving conditional expectations are self-dual, i.e., $E=E^*$:
\begin{align}
    \tr(E(a)b) = \tr(E(E(a) b)) = \tr(E(a) E(b)) = \tr(E(a E(b) )) = \tr(a E(b)), \quad a,b\in L(\H),
\end{align}
where we used \eqref{eq:bimodule-CE}.
In particular, if $\sigma$ is a state on $\H$, then $E^*\sigma = E\sigma \in J$ is also a state on $\H$, which assigns the same expectation values to operators in $J$.

The following lemma follows an argument by Haagerup and St\o rmer \cite{haagerup_positive_1995}.
\begin{lemma}\label{lem:CE-selfadjoint}
    Let $\sigma$ be a faithful state, $J\subset L(\H)$ a J*-algebra with $\sigma$-preserving conditional expectation $F:L(\H)\to J$. Then 
    \begin{equation}
        \ip{Fa}{b}_\sigma = \ip{a}{Fb}_\sigma, \qquad a,b\in L(\H).
    \end{equation}
\end{lemma}

\begin{proof}       
    By \cref{cor:UP-contraction}, we have $\ip{Fa}{Fa}_{\sigma} \leq \ip{a}{a}_{F^* \sigma} = \ip{a}{a}_\sigma$, $a\in L(\H)$.
    Thus $F$ is a contraction for the KMS inner product. 
    But any idempotent contraction on a Hilbert space is hermitian.\footnote{We check that the range of $F$ is orthogonal to its null space, which shows the claim since $F$ is idempotent. Let $a=Fa$ be in the range of $F$ and $0\neq b$ be in the null space of $F$.  
    Since $F$ is a contraction, we have $\norm{a}_\sigma = \norm{F(a-\lambda b)}_\sigma \leq \norm{a-\lambda b}_\sigma$, where $\norm{a}_\sigma^2 := \ip{a}{a}_\sigma$. Set $\lambda = \ip{b}{a}_\sigma/\ip{b}{b}_\sigma$. Then
    \begin{align*}
        \norm{a}_\sigma^2 \leq \norm{a- \tfrac{\ip{b}{a}_\sigma}{\ip{b}{b}_\sigma}b}_\sigma^2 = \norm{a}_\sigma^2 - \tfrac{|\ip{a}{b}_\sigma|^2}{\ip{b}{b}_\sigma}.
    \end{align*}
    }
\end{proof}

Using this result, we can express a state-preserving conditional expectation in terms of the trace-preserving conditional expectation. 
The formula below was also noted by Jen\v{c}ov\'{a} \cite{jencova_quantum_2010} for sufficient *-algebras (see also \cite{yamagata_quantum_2026}).

\begin{corollary}\label{cor:uniquenes-CE}
     Let $\sigma$ be a faithful state, and $J\subset L(\H)$ a J*-algebra with $\sigma$-preserving conditional expectation $F:L(\H)\to J$. Then $F$ is unique and can be expressed as
	\begin{align}\label{eq:state-preserving-CE}
         F(a) = (E\sigma)^{-\frac12}\, E(\sigma^{\frac12} a \sigma^{\frac12})\, (E\sigma)^{-\frac12},
     \end{align}
     where $E:L(\H)\to J$ is the trace-preserving conditional expectation.
\end{corollary}
\begin{proof}
        We first show $\ip{a}{b}_{E\sigma} = \ip{a}{b}_\sigma$ for all $a,b\in J$. 
        By polarization, it suffices to show this for $a=b$.  Let $a\in J$. Then $a=Ea=Fa$. Using \cref{cor:UP-contraction} twice we get
        \begin{align*}
            \ip{a}{a}_\sigma = \ip{Ea}{Ea}_\sigma \leq \ip{a}{a}_{E\sigma} = \ip{Fa}{Fa}_{E\sigma} \leq \ip{a}{a}_{F^*E\sigma} = \ip{a}{a}_{\sigma}. 
        \end{align*}
        Thus, equality holds everywhere, and we may conclude $\ip{a}{b}_{E\sigma} = \ip{a}{b}_\sigma$ for all $a,b\in J$.
        Using that $F$ is hermitian for the KMS inner product, we have $\ip{Fa}{b}_{E\sigma} = \ip{Fa}{b}_\sigma = \ip{a}{Fb}_\sigma= \ip{a}{b}_\sigma$ for $a\in L(\H), b\in J$. 
        Therefore,
        \begin{align*}
            \ip{Fa}{b}_{E\sigma}=\ip{EFa}{b}_{E\sigma} = \ip{Fa}{Eb}_{E\sigma} = \ip{Fa}{Eb}_\sigma = \ip{a}{Eb}_\sigma, \qquad a,b\in L(\H).
        \end{align*}
        But 
        \begin{align*}
            \ip{Fa}{b}_{E\sigma} = \tr((E\sigma)^{\frac12}F(a) (E\sigma)^{\frac12} b)
        \quad\qandq\quad
            \ip{a}{Eb}_\sigma = \tr(E(\sigma^{\frac12} a \sigma^{\frac12})b).
        \end{align*}
        Since $b\in L(\H)$ was arbitrary, we find $(E\sigma)^{\frac12}F(a)(E\sigma)^{\frac12} = E(\sigma^{\frac12} a \sigma^{\frac12})$.
\end{proof}

If a conditional expectation $F$ preserves two faithful states $\rho$ and $\sigma$, it follows that
\begin{align}
    (E\sigma)^{-\frac12}\, E(\sigma^{\frac12} a \sigma^{\frac12})\, (E\sigma)^{-\frac12} = (E\rho)^{-\frac12}\, E(\rho^{\frac12} a \rho^{\frac12})\, (E\rho)^{-\frac12}, \quad a\in L(\H).
\end{align}

A technical lemma that we will need later shows that the range of a faithful conditional expectation, in fact, coincides with the multiplicative domain:

\begin{lemma}\label{lem:faithful-CE-multi-domain}
    Let $E: L(\H)\to J\subset L(\H)$ be a faithful conditional expectation. Then $\mathrm{MD}(E) = \mathrm{Fix}(E) = J$. 
\end{lemma}

\begin{proof}
    It is clear that $\mathrm{Fix}(E) = J \subset \mathrm{MD}(E)$. 
    The converse is argued as follows: $E$ is a J*-homomorphism on its multiplicative domain. 
    Since $E$ is assumed faithful, it must be injective. Hence, the image of $E$ is isomorphic to the multiplicative domain. But since $E$ is the identity on its image, this shows that the multiplicative domain coincides with the image.
\end{proof}

Next, we consider the relationship between conditional expectations onto a J*-algebra and conditional expectations onto the *-algebras that it generates.

\begin{lemma}\label{lem:CE-extension}
	Let $J\subset L(\H)$ be a J*-algebra with faithful conditional expectation $F$. Then $A = \staralg(J)$ admits a conditional expectation $\tilde F$ such that $F \circ \tilde F = F$. 
    In particular, this gives $\tilde F^*\rho = \rho$ whenever $F^*\rho = \rho$.
\end{lemma}

\begin{proof}
	The following argument is taken from the proof of \cite[Thm.~2]{luczak_aspects_2021}, see also \cite[Lem.~2.2]{haagerup_positive_1995}: Let $\tau = \frac{\1}{\tr(\1)}$ be the maximally mixed state and set $\omega = F^*\tau$. Then $\omega$ is $F$-invariant and faithful, since $F$ is faithful. 
    Let $p\in J$ be a projection and $s= 2p-\1$, i.e., $s$ is a hermitian unitary. 
    Then $F(sbs) = sF(b)s$ for any $b\in L(\H)$, and hence $F(sb) = F(sbss) = sF(bs)s$.
    Therefore
    \begin{align}
    \tr(\omega bs) = \tr(\tau F(bs)) = \tr(\tau sF(sb)s) = \tr(\tau F(sb)) = \tr(\omega sb), \quad b\in L(\H).
    \end{align}
	Since $p$ was an arbitrary projection, we find $\tr(\omega pb) = \tr(\omega bp)$ for every projection $p\in J$ and any $b\in L(\H)$. Since $J$ is spanned by its projections, we find $\tr(\omega ab)=\tr(a\omega b)$ for all $a\in J$ and $b\in L(\H)$. Hence $[\omega,a]=0$ for all $a\in J$. 
    Since $A$ is generated by $J$, it follows that  $[\omega,a]=0$ for all $a \in A$.
    By Takesaki's theorem \cite{takesakiConditionalExpectationsNeumann1972}, this implies that there exists a unique $\omega$-preserving conditional expectation $\tilde F$ onto $A$.%
    \footnote{Since we work in finite dimensions, one can explicitly write down $\tilde F$ using the decomposition $\H = \oplus_j \H_j \otimes \K_j$ and $A = \oplus_j L(\H_j) \otimes \1$. Then $\omega = \oplus_j \tfrac{p_j}{\dim(\H_j)}\1 \ox \omega_j$ for some faithful probability distribution $(p_j)$ and faithful states $\omega_j$ and hence $\tilde F = \oplus_j \id_j \otimes \tr(\omega_j (\placeholder))\1$.} 
    Since $J=\mathrm{Ran}(F)$ is contained in $A$, we have $\tilde F F = F$. 
    Hence $(F\tilde F)^2 = F\tilde F F\tilde F = FF\tilde F = F\tilde F$, so that $F\tilde F$ is a conditional expectation onto $J$ that leaves $\omega$ invariant.
    Thus, the uniqueness statement in \cref{cor:uniquenes-CE} implies $F\tilde F= F$.
    Now suppose $F^*\rho = \rho$.
	It follows that $\tilde F^*\rho = \tilde F^*F^*\rho = (F\tilde F)^*\rho = F^*\rho = \rho$. Thus 
	$\tilde F$ is $\rho$-preserving. 
\end{proof}
Note that by \cref{lem:CE-extension} any $\sigma$-preserving conditional expectation $F:L(\H)\to J$ factorizes as
\begin{align}\label{eq:CE-factorization}
	F: L(\H) \overset{\tilde F}{\longrightarrow} \staralg(J) \overset{\hat F}{\longrightarrow} J,
\end{align}
where $\hat F = F|_{\staralg(J)}$ is simply the restriction of $F$ onto $\staralg(J)$.

For any two J*-algebras $J_1,J_2\subset L(\H)$, their intersection $J_1\cap J_2$ is again a J*-algebra. 
If $F_i:L(\H) \to J_i$ are conditional expectations that have a common, faithful invariant state $\sigma$ (i.e., are $\sigma$-preserving), then both are hermitian projections when we consider $L(\H)$ as a Hilbert-space with the KMS inner product.
Von Neumann's projection theorem \cite[Thm.~13.7]{neumannFunctionalOperatorsVolume1951} then implies that 
\begin{align}\label{eq:E12}
    F = \lim_{n\to \infty} (F_1F_2)^n  = \lim_{n\to\infty} (F_2 F_1)^n
\end{align}
is a $\sigma$-preserving conditional expectation onto the intersection $J_1\cap J_2$ such that
$F_i \circ F = F = F \circ F_i$.
Thus, $F$ is the unique $\sigma$-preserving conditional expectation onto $J_1\cap J_2$.
Using Halperin's projection theorem \cite{halperin1962product}, the argument generalizes to any finite number of $\sigma$-preserving conditional expectations $F_1,\ldots, F_m$ onto J*-algebras $J_1,\ldots,J_m$. Then
\begin{align}
    F = \lim_{n\to\infty} (F_1\cdots F_m)^n
\end{align}
is the unique $\sigma$-preserving conditional expectation onto $\cap_{i=1}^m J_i$.

\begin{lemma}
    Let $F_1 : L(\H)\to J_1$ and $F_2:L(\H)\to J_2$ be conditional expectations onto J*-algebras $J_1,J_2$ with common, faithful invariant state and such that $F_1  F_2 = F_1$ and $F_2  F_1 = F_2$. Then $J_1=J_2$ and $F_1 = F_2$.
    \begin{proof}
       The relations $F_1  F_2 = F_1$ and $F_2  F_1 = F_2$ imply $F = F_1 = F_2$.
    \end{proof}
\end{lemma}

\subsection{Structure theory of abstract J*-algebras}\label{sec:structure-theory}

So far, we have considered J*-algebras as collections of operators on a given Hilbert space $\H$.
It is sometimes better to regard J*-algebras as abstract objects and to regard the action on a given Hilbert space as a representation.
We make the following non-standard definition:%
\footnote{Our notion of an abstract J*-algebra is known in the mathematics literature as a finite-dimensional JC*- or JW*-algebra (cp.\ \cref{foot1}).
JC*- and JW*-algebras are a special case of the more general notion of JB*- or JBW*-algebra, which do not necessarily admit representations as operators on a Hilbert space \cite{hanche-olsen_jordan_1984}.}

\begin{definition}
    We define an \emph{abstract J*-algebra} 
    as a finite-dimensional *-vector space $J$, equipped with a map $\jp\placeholder\placeholder: J\times J \to J$, and a special element $1=1^*\in J$, such that there exists a Hilbert space $\H$ and an injective *-preserving linear map $\pi : J \to L(\H)$, which is unital, i.e., satisfies $\pi(1)=\1$, and takes the Jordan product on $J$ to the Jordan product of $L(\H)$, i.e.,
        \begin{equation}
        \pi(\jp ab) = \jp{\pi(a)}{\pi(b)}, \qquad a,b\in L(\H).
    \end{equation}
    A tuple $(\pi,\H)$ with these properties is called a \emph{(J*-)representation}.
\end{definition}

If $(\pi,\H)$ is a representation of an abstract J*-algebra $J$, the range $\pi(J)$ is a J*-algebra on $\H$.
Thus, the Jordan product $\jp\placeholder\placeholder$ on an abstract J*-algebra satisfies all algebraic properties enjoyed by the Jordan product on $L(\H)$.
If $(\H_i,\pi_i)$, $i=1,2$, are J*-representation of an abstract J*-algebra $J$, then $\pi_1\circ \pi_2^{-1}$ is an J*-isomorphism between the J*-algebras $\pi_i(J)\subset L(\H_i)$, $i=1,2$ (note that the $\pi_j$ are injective).
Consequently, all properties of J*-algebras that are preserved by J*-isomorphisms can be understood as properties of abstract J*-algebras.

Every finite-dimensional *-algebra $A$ is an abstract J*-algebra in the obvious way, and every concrete J*-algebra $J\subset L(\H)$ on a Hilbert space $\H$ is an abstract J*-algebra by taking $\pi$ as the embedding $J \hookrightarrow L(\H)$.
A simple example of a J*-representation of the *-algebra $L(\CC^n)$, which is not a *-representation, is the following:

\begin{example}\label{exa:J*-rep}
    The map $a \mapsto a\oplus a^t$ is a J*-representation of $L(\CC^n)$ on $\CC^{2n}=\CC^n\oplus\CC^n$.
\end{example}

It makes sense to speak of J*-homomorphisms between abstract J*-algebras.
More generally, UP maps between abstract J*-algebras make sense because the latter have a well-defined positive cone $J^+$ (cp.\ \eqref{eq:positive-cone-squares}) and unit element $1\in J$. 

We will often drop the adjective "abstract" in cases where the distinction between abstract and concretely represented J*-algebras is either irrelevant or clear from the context.
In the following, we discuss the structure of abstract J*-algebras.
We begin by discussing the center and direct sum decomposition into so-called J*-factors. Then, we describe the classification of J*-factors.

The center $Z(J)$ of an abstract J*-algebra $J$ is the subspace of elements $z\in J$ such that $uzu = z$ for all hermitian unitaries $u\in J$, see \cite[Paragraph 2.5.1 \& Lem.~4.3.2]{hanche-olsen_jordan_1984}.
This definition makes sense for abstract J*-algebra because the equation $uzu = z$ can be expressed in terms of the Jordan product, see \cref{sec:basics}.
The following Lemma summarizes several characterizations of central elements for a representation $J\subset L(\H)$.

\begin{lemma}\label{lem:center}
    Let $J\subset L(\H)$ be a J*-algebra and $z\in J$. The following are equivalent ($J'$ denotes the commutant of $J$ in $L(\H)$):
	\begin{enumerate}[(a)]
		\item\label{it:center1} $z\in Z(J)$.
		\item\label{it:center2} $[z,a]=0$ for all $a\in J$, i.e., $z\in J\cap J'$.
		\item\label{it:center3} $z \in Z(\staralg(J)) = J''\cap J'$.
		\item\label{it:center4} $\jp{a}{\jp{z}{b}} = \jp{z}{\jp{a}{b}}$ for all $b,c\in J$.
	\end{enumerate}
\end{lemma}

We remark that \cref{it:center4} is often taken as the definition of the center \cite{hanche-olsen_jordan_1984}.
The Lemma shows that the center is given by $Z(J) = J\cap J'$.
Therefore, the above definition generalizes the definition of the center of a *-algebra.

\begin{proof}
 \ref{it:center1} $\Leftrightarrow$ \ref{it:center2}: $[z,a]=0$ for all $a\in J$ if and only if $0=[z,p]$ for all projection $p\in J$, if and only if $uzu=z$ for all hermitian unitaries $u\in J$ (a hermitian unitary is of the form $u=\1-2p$ for a projection $p$).

	\ref{it:center2} $\Rightarrow$ \ref{it:center3}: Obvious.

	\ref{it:center3} $\Rightarrow$ \ref{it:center4}: Since $z\in J'$ we have $\jp{a}{\jp{z}{b}} = \jp{a}{zb} = \jp{a}{b}z = z\jp{a}{b} = \jp{z}{\jp{a}{b}}$.

	\ref{it:center4} $\Rightarrow$ \ref{it:center2}: If $\ref{it:center4}$ holds we get:
	\begin{align}
	0&= a(zb+bz) + (zb+bz)a - z(ab+ba) - (ab+ba) z = [a,z]b - b[a,z] = [[a,z],b],\quad a,b\in J.\nonumber	
	\end{align}
	Hence $[a,z]\in J'$ and hence $[a,z]\in Z(\staralg(J))$. Since a commutator can only be central in a finite-dimensional *-algebra if it vanishes\footnote{Indeed, if $[a,b]$ is central for elements $a,b$ of a *-algebra $A$, then $\tr([a,b]^*[a,b]) = \tr([[a,b]^*,a]b) = 0$ shows $[a,b]=0$.}, this implies $[a,z]=0$.
\end{proof}

By applying \cref{lem:center} in some J*-representation, we see that the Jordan product is associative on the center $Z(J)$ of an abstract J*-algebra $J$.
Thus, $Z(J)$ is an abelian *-algebra.
Therefore, the center is of form $Z(J) = \lin\,\{p_j\}_{j=1}^n$, for a family of pairwise orthogonal projections $p_i\in J$ such that $\sum_j p_j = \1$.
A J*-algebra $J$ is a \emph{J*-factor} if its center consists of scalars $Z(J)=\CC$ \cite[Sec.~4.6]{hanche-olsen_jordan_1984}.

We need to understand direct sums of J*-algebras.
If $J_i\subset L(\H_i)$, $i=1,\ldots, n$, are J*-algebras, then
\begin{equation}
    \oplus_i J_i := \big\{ \oplus_i a_i \ : \ a_i \in J_i \big\} \ \subset \ L(\oplus_i \H_i)
\end{equation}
is a J*-algebra on $\oplus_i\H_i$.
Up to J*-isomorphism, $\oplus_i J_i$ is independent of the representations of the $J_i$.
Therefore, the direct sum makes sense on the level of abstract J*-algebras:
If $(\pi,\H_i)$ are J*-representations of abstract J*-algebras $J_i$, then the abstract J*-algebra $\oplus_i J_i$ is defined through the representation $(\oplus_i \pi_i,\oplus_i\H_i)$.
It is easy to see that the center of the direct sum is $Z(\oplus_i J_i) = \oplus_i Z(J_i)$.

The key fact underlying the classification is that an abstract J*-algebra $J$ decomposes as a direct sum 
\begin{equation}
    J=\oplus_i J_i
\end{equation}
of J*-factors $J_i$ and this decomposition is unique up to J*-isomorphism \cite[Thm.~5.3.5]{hanche-olsen_jordan_1984}.
The above reduces the classification of abstract J*-algebras to the classification of J*-factors, which we consider next.

The classification of J*-factors is into four infinite families, all of which we have encountered in \cref{sec:basics}:

\begin{theorem}[{\cite{jordan_algebraic_1934}, see also \cite{hanche-olsen_jordan_1984,jacobson_structure_1968}}]\label{thm:classification}
    A J*-factor $J$ is J*-isomorphic to one of the following:
    \begin{enumerate}
        \item \label{it:classifiaction1} 
            the full matrix algebra $L(\CC^n)$, $n\in\NN$;
        \item \label{it:classifiaction2} 
            the J*-algebra $L(\CC^{n})^t$ of symmetric $n\times n$-matrices, $n\in \NN$ (see item \ref{it:fixed-pts1} of \cref{exa:fixed-pts});
        \item \label{it:classifiaction3}
            the J*-algebra $L(\CC^{2n})^\beta$ of $2n\times 2n$-matrices invariant under the symplectic involution $\beta$, $n\in \NN$ (see item \ref{it:fixed-pts2} of \cref{exa:fixed-pts});
        \item\label{it:classifiaction4} 
            A spin factor $V_n$, $n\in \NN$, as constructed in \cref{exa:spin-factors}.
    \end{enumerate}
    Apart from the trivial identification $L(\CC^1) = L(\CC^1)^t \Jcong V_1 \Jcong L(\CC^2)^\beta$, the only overlaps of these families are $V_2 = L(\CC^2)^t$, $V_3 \Jcong L(\CC^2)$, and $V_5 \Jcong L(\CC^4)^\beta$.
\end{theorem}

The classification of J*-factors is usually presented in an alternate form (see \cite{jacobson_structure_1968,hanche-olsen_jordan_1984,jordan_algebraic_1934}):
As mentioned earlier, the literature is mostly concerned with Jordan algebras of hermitian matrices, which are related to the J*-algebras that we consider through complexification (see \cref{sec:basics}).
The four families are then the symmetric matrices over $\RR$, corresponding to item \ref{it:classifiaction2} of \cref{thm:classification}, the hermitian matrices over $\CC$, corresponding to item \ref{it:classifiaction1}, the hermitian matrices over the quaternions $\mathbb H$, corresponding to item \ref{it:classifiaction3}, and the spin factors, defined as in \cref{exa:spin-factors} with the span taken over $\RR$.
The only non-obvious identification is that of the hermitian quaternionic matrices with the hermitian part of the J*-algebra in item \ref{it:classifiaction3} of \cref{thm:classification}.
We explain this connection in \cref{rem:quaternion-stuff} below.

Any *-invariant unital subspace $J\subset A$ of a finite-dimensional *-algebra $A$, which is closed under the Jordan product, is an abstract J*-algebra.
Indeed, this holds because each finite-dimensional *-algebra is *-isomorphic to a *-algebra of operators on a finite-dimensional Hilbert space.
Next, we shall consider a canonical way to embed an abstract J*-algebra into a (finite-dimensional) *-algebra:
By modding out a suitable ideal of the tensor algebra over a J*-algebra (see  \cite[Thm.~7.1.8]{hanche-olsen_jordan_1984} or \cite{jacobson_structure_1968}), one can construct a pair $(\iota,A)$ of a *-algebra $A$ and a J*-embedding, i.e., an injective J*-homomorphism, $\iota : J \to A$ enjoying the following property:
For every J*-homomorphism $\phi:J\to B$ into a *-algebra $B$, there exists a *-homomorphism $\hat\phi:A\to B$ extending $\phi$, i.e., the following diagram commutes:
\begin{equation}\label{eq:universal-prop}
    \begin{tikzcd}
        &J\arrow{rr}{\phi}\arrow{dr}{\iota} & & B &  \\
        & & A \arrow[dotted]{ur}{\exists\hat\phi}
    \end{tikzcd}
\end{equation}

This universal property determines $(\iota,A)$ uniquely up to *-isomorphism.
We refer to $(\iota,A)$ as a \emph{universal enveloping *-algebra}.
If $(\iota,A)$ is a universal enveloping *-algebra, then so is $(\iota,A^{\mathrm{op}})$, where $A^{\mathrm{op}}$ denotes the opposite *-algebra.
Thus, the uniqueness of the enveloping *-algebra up to *-isomorphism implies the existence of a *-antiautomorphism $\vartheta$ on $A$, which is the same as a *-isomorphism $\vartheta:A\to A^{\mathrm{op}}$, with  
\begin{equation}\label{eq:iota-theta-invariant}
    \vartheta \circ \iota = \iota.
\end{equation}
As explained in \cite[Sec.~II.1]{jacobson_structure_1968}, this *-antiautomorphism is an involution, i.e., satisfies $\vartheta^2 = \id$, known as the \emph{canonical involution} on $A$.

We will discuss explicitly the universal enveloping *-algebras of the various kinds of J*-factors in the next subsection, where we use universal enveloping *-algebras to construct a useful class of J*-representations.

\begin{remark}\label{rem:quaternion-stuff}
    We follow \cite[Sec.~3]{idel} to explain how the J*-algebra of symplectically invariant $2n\times 2n$-matrices relates to the Jordan algebra $\mathrm{Herm}_n(\mathbb H)$ of hermitian matrices over the quaternions.
    We denote the three imaginary units of $\mathbb H$ as $i$, $j$, and $k$.
    A matrix $a$ over $\mathbb H$ has a unique decomposition $a = a_r + ia_i + ja_j + ka_k$, where $a_r,a_i,a_j,a_l$ are real matrices.
    The adjoint on $M_n(\mathbb H)$ acts as a transpose and entry-wise conjugation (conjugation on $\mathbb H$ is the real linear map that sends $1,i,j,k$ to $1,-i,-j,-k$.).
    A matrix $a\in M_n(\mathbb H)$ is hermitian if and only if $a_r$ is symmetric and $a_i,a_j,a_k$ are skew-symmetric. 
    The real vector space $\mathrm{Herm}_n(\mathbb H)$ is closed under the Jordan product $\jp ab = \frac12(ab+ba)$, where the product is simply matrix multiplication.
    We claim that the map $\varphi: \mathrm{Herm}_n(\mathbb H) \to L(\CC^{2n})_\mathrm{sa}$ given by
    \begin{equation}
        \varphi(a) = 
        \begin{pmatrix}
            a_r +i a_i & a_j + i a_k \\ a_j - i a_k & a_r - i a_i
        \end{pmatrix} 
    \end{equation}
    is a real linear bijection onto the hermitian part of the fixed-point J*-algebra $L(\CC^{2n})^\beta$ that takes the Jordan product on $\mathrm{Herm}_n(\mathbb H)$ to the Jordan product on $L(\CC^{2n})$.
    The characterization in \eqref{eq:symplectic-fixed-pts} shows that the range of $\varphi$ is contained in $L(\CC^{2n})_h^\beta$.
    It is clear that $\varphi$ is injective, and it can be checked through an explicit calculation that $\varphi$ respects the Jordan product.
    It follows readily from \eqref{eq:symplectic-fixed-pts} that the range of $\varphi$ is $L(\CC^{2n})^\beta_\mathrm{sa}$.
    This shows that the complexification of  $\mathrm{Herm}_n(\mathbb H)$ is an abstract J*-algebra from which $L(\CC^{2n})^\beta$ arises via a representation on $\CC^{2n}$.
\end{remark}

\subsection{Representations of J*-algebras}
\label{sec:reps}

In the previous subsection, we considered properties of J*-algebras that are independent of the chosen representation. 
We also need to understand representation-dependent properties.
In the following we describe part of the representation theory of J*-algebras.

As is well-known, any (finite-dimensional) *-algebra $A$ is *-isomorphic to $A\cong \oplus_j L(\H_j)$ and any *-representation of $A$ on a finite-dimensional Hilbert space is of the form
\begin{equation}
    \H \cong \oplus_j \H_j\ox \L_j, \quad A \cong \oplus_j L(\H_j) \ox \1.
\end{equation}
Thus, a *-representation is, up to unitary equivalence, determined by the multiplicities $(l_j)$, $l_j = \dim \L_j$.
We will see that the corresponding statement for J*-algebras is false.
For instance, the defining representation of $L(\CC^n)$ on $\CC^n$ is not unitarily equivalent to the J*-representation $a\mapsto a\oplus a^t$ on $\CC^{2n}$ (see \cref{exa:J*-rep}) although both J*-representations are multiplicity free.
In this case, the non-equivalence of the two J*-representations is clear from the different Hilbert space dimensions.
A more useful invariant that we can associate with a J*-representation $(\pi,\H)$ is the isomorphism class of the *-algebra $\staralg(\pi(J))$ generated by the representation.
This invariant is insensitive to changes in the multiplicity of the representations.

Let us consider a J*-representation $(\pi,\H)$ of a J*-algebra $J$. 
Let $J =\oplus_i J_i$ be the direct sum decomposition into J*-factors.
Let $p_i$ be the unit of $J_i$ as an element of $J$.
Then $\pi(p_i)$ is a family of projections on $\H$ with $\sum_i p_i=1$, and $(\pi,\H)$ is evidently a direct sum of J*-representations $(\pi_i,\H_i)$, where $\H_i = \pi(p_i)\H$.
Therefore, to understand J*-representations of general J*-algebras, it is sufficient to understand those of J*-factors.

Next, we consider a particularly useful representation, derived from the universal enveloping *-algebra, which we call the \emph{universal representation} (following \cite{idel}).%
\footnote{Our definition of the universal representation of a J*-algebra is not related to the concept of a universal representation of a C*-algebra.}

\begin{definition}\label{def:univ-rep}
    A representation $(\pi,\H)$ of an abstract J*-algebra $J$ is \emph{universal} if $(\pi, A= \staralg(\pi(J))$ is the universal enveloping *-algebra of $J$ and if $A\subset L(\H)$ is a multiplicity-free representation of $A$.
\end{definition}
The uniqueness of the universal enveloping *-algebra and the uniqueness of the multiplicity-free representation of a (finite-dimensional) *-algebra up to unitary equivalence imply the following:

\begin{lemma}\label{lem:unitary-equivalence-univ}
    If $(\pi_i,\H_i)$, $i=1,2$, are universal representations of a J*-algebra $J$, then $\pi_1$ and $\pi_2$ are unitarily equivalent.
\end{lemma}

\begin{example}\label{exa:univ-reps}
    We describe the universal representations of J*-factors.
    The proofs of these claims are given in \cite{jacobson_classification_1949} or \cite[Sec.~3.3]{idel}
    \begin{enumerate}
        \item 
            The universal representation of $L(\CC^n)$ for $n\geq 2$ is the representation $a\mapsto a\oplus a^t$ on $\CC^{2n}$, 
            and the universal enveloping *-algebra is $\staralg(\{a\oplus a^t : a\in L(\CC^n)) = L(\CC^n)\oplus L(\CC^n)$.
        \item 
            The universal representation of the J*-factor $L(\CC^n)^t$ of symmetric matrices on $\CC^n$ for $n\geq 2$ (see \eqref{eq:symmetric-matrices}), is the defining representation on $\CC^n$, and the universal enveloping *-algebra is $\staralg(L(\CC^n)^t) = L(\CC^n)$. 

        \item 
            For $n\geq 3$, the universal representation of the J*-factor $L(\CC^{2n})^\beta$ of $2n\times 2n$ matrices invariant under the symplectic involution $\beta$ (see \eqref{eq:symplectic-fixed-pts}) is the defining representation on $L(\CC^{2n})$, and the universal enveloping *-algebra is $\staralg(L(\CC^{2n})^\beta) = L(\CC^{2n})$.
        
        \item 
            The universal representation of a spin factor $V_n$, $n\in\NN$, is the representation described in \cref{exa:spin-factors}, the universal enveloping *-algebra is $\staralg(V_n) = L((\CC^2)^{\otimes \frac n2})$ if $n$ is even and $\staralg(V_n) = L((\CC^2)^{\otimes \frac{n-1}2})\oplus L((\CC^2)^{\otimes \frac{n-1}2})$ if $n$ is odd.
    \end{enumerate}
\end{example}

The universal enveloping *-algebra and, therefore, the universal representation are compatible with direct sums: If $J_k \subset L(\H_k)$ are J*-algebras with universal enveloping *-algebras $(A_k,\iota_k)$ and universal representations $(\pi_k,\hat\H_k)$, then $(\oplus_k \iota_k, \oplus_k A_k)$ is a universal enveloping *-algebra and $(\oplus_k \pi_k,\oplus_k \hat\H_k)$ is a universal representation of $\oplus_k J_k$.
Therefore, \cref{exa:univ-reps} describes the universal representations of general J*-algebras (see \cref{sec:structure-theory}).

Universal enveloping *-algebras are not only interesting because they give rise to useful representations.
More importantly, they can be used to reduce the representation theory of J*-algebras to that of *-algebras.
If $(\pi,\H)$ is a J*-representation of a J*-algebra $J$ and if $(\pi_u,\H_u)$ is the universal representation, then there is a *-homomorphism $\hat\pi:A\to L(\H)$ with $\hat\pi \circ \iota = \pi$, where $A=\staralg(\pi_u(J))$ is the universal enveloping *-algebra.

The representation theory becomes particularly simple if the universal enveloping *-algebra $A$ is a matrix factor.
In this case, every J*-representation is unitarily equivalent to the universal representation with additional multiplicity:

\begin{corollary}\label{cor:reps-if-A-factor}
    Let $J$ be an abstract J*-algebra whose universal enveloping *-algebra $A$ is a matrix factor. 
    If $(\pi,\H)$ is a J*-representation and if $(\pi_u,\H_u)$ denotes the universal representation, then there exists a unitary $u:\H\to \H_u\ox \L$ for some Hilbert space $\L$ such that
    \begin{equation}
        \pi = u^* (\pi_u \ox \1_\L) u.
    \end{equation}
\end{corollary}

\begin{proof}
    As the universal enveloping *-algebra we can take $(\iota,A) = (\pi_u,L(\H_u))$.
    By the universal property \eqref{eq:universal-prop}, there is a *-homomorphism $\hat \pi :L(\H_u)\to L(\H)$ with $\hat\pi \circ \pi_u = \pi$.
    Thus, there is unitary $u:\H\to \H_u \ox \L$ with $\hat\pi = u^*(\id \ox \1)u$, and we have $\pi = \hat\pi \circ \pi_u = u^*(\pi_u \ox \1)u$.
\end{proof}

Next, we discuss in which representations a J*-algebra is reversible.
Recall that a concretely represented J*-algebra $J\subset L(\H)$ is reversible if it is closed under higher-order symmetrized products, i.e., if for all $a_1,\ldots,a_n\in J$ the operator $\trip{a_1}{\ldots}{a_n}$ is also in $J$.
It is easy to verify that $J$ is reversible if and only if each J*-factor $J_i \subset L(\H_i)$ in the direct sum decomposition $J = \oplus_i J_i$, $\H = \oplus_i \H_i$, is reversible.
Thus, we only have to understand the reversibility of J*-factors.

\begin{proposition}[{\cite[Thms.~5.3.10 \& 6.2.5]{hanche-olsen_jordan_1984}}]\label{prop:reversible}
    Let $J$ be a J*-factor on $\H$. 
    Then, if $J$ is not a spin factor, then it is reversible. 
    If $J \Jcong V_n$ is a spin factor, then:
    \begin{enumerate}[(i)]
        \item If $n=2$ or $n=3$, then it is reversible;
        \item If $n=4$ or $n \ge 6$, then $J$ is irreversible.
    \end{enumerate}
    The spin factor $V_5 \Jcong L(\CC^4)^\beta$ has both reversible and irreversible representations.
\end{proposition}

We have seen that some J*-algebras are reversible in every representation. 
We call such J*-algebras \emph{universally reversible}.
This class of J*-algebras will be important later on.

\begin{proposition}\label{prop:univ-rev}
    Let $J$ be a J*-algebra in its universal representation. 
    Let $\vartheta$ be the canonical involution on the universal enveloping *-algebra $A = \staralg(J)$. 
    The following are equivalent:
    \begin{enumerate}[(a)]
        \item\label{it:univ-rev1} $J$ is reversible as a subalgebra of $A$;
        \item\label{it:univ-rev2} $J$ is universally reversible;
        \item\label{it:univ-rev3} $J = A^\vartheta$.
        \item\label{it:univ-rev4} $J$ contains no direct summands that are J*-isomorphic to spin factors $V_n$, $n\ge 4$. 
    \end{enumerate}
\end{proposition}

For the proof, we need some preparation.
If $K\subset L(\H)$ is an operator system, we define 
\begin{equation}\label{eq:def-RJSA}
	\revJstaralg(K) = \bigcap_{J\supset K} J
\end{equation}
where the intersection is over \emph{reversible} J*-algebras $J\subset L(\H)$ containing $K$. 
We refer to $\revJstaralg(K)$ as the reversible J*-algebra \emph{generated by} $K$.
By definition, it is the smallest reversible J*-algebra containing $K$.

\begin{lemma}\label{lem:rev-hull1}
    Let $J$ be a J*-algebra on a Hilbert space $\H$.
	Let $A=\staralg(J) \subset L(\H)$ be the *-algebra generated by $J$ and suppose that $\vartheta$ is an involution on $A$, which leaves $J$ pointwise fixed.
    Then the fixed-point set $A^\vartheta$ is the reversible J*-algebra generated by $J$.
    \begin{equation}
        A^\vartheta = \revJstaralg(J).
    \end{equation}
\end{lemma}
\begin{proof}
    The inclusion $J\subset A^\vartheta$ holds by assumption.
    We check that $A^\vartheta$ is reversible.
    To see this, note that the symmetrization map $E = \frac12(\id + \vartheta)$ is a conditional expectation of $A$ onto $A^\vartheta$.
    Now, if $a_1,\ldots,a_n\in A^\vartheta$, then $E(a_1\cdots a_n+a_n\cdots a_1)= \frac 12(a_1\cdots a_n + a_n\cdots a_1 +a_n\cdots a_1+a_1\cdots a_n) =  a_1\cdots a_n+a_n\cdots a_1$ proves the reversibility of $A^\vartheta$.
    Thus, we have $A^\vartheta \supset \revJstaralg(J)$.
	Since every element of $A$ is a noncommutative polynomial in $J$, every element of $A^\vartheta=E(A)$ is a symmetrized noncommutative polynomial in $J$.
    Clearly, all symmetrized noncommutative polynomials in $J$ are in every reversible J*-algebra that contains $J$.
    Thus, $A^\vartheta \subset \revJstaralg(J)$.
\end{proof}

\begin{proof}[Proof of \cref{prop:univ-rev}]
    \ref{it:univ-rev1} $\Leftrightarrow$ \ref{it:univ-rev2} $\Leftarrow$ \ref{it:univ-rev3} are clear.
    \ref{it:univ-rev2} $\Rightarrow$ \ref{it:univ-rev3} is shown in \cref{lem:rev-hull1}.
    \ref{it:univ-rev2} $\Leftrightarrow$ \ref{it:univ-rev4} is shown in \cref{prop:reversible}.
\end{proof}

We note the following consequence of \cref{prop:univ-rev}, which will be useful for us later on:

\begin{corollary}\label{cor:rev-hull}
    Let $J$ be a universally reversible J*-algebra and let $A$ be its universal enveloping *-algebra.
    Suppose that a J*-subalgebra $J_0\subset J$ generates $A$ as a *-algebra.
    Then, in any representation, $J \subset L(\H)$ is the reversible J*-algebra generated by $J_0$:
    \begin{equation*}
        J = \revJstaralg(J_0).
    \end{equation*}
\end{corollary}
\begin{proof}
    If we show that $J$ is the reversible subalgebra of $A$ that is generated by $J_0\subset J\subset A$, then the universal property of the universal enveloping *-algebra implies the claim in all representations.
    Let $\vartheta$ be the canonical involution on $A$.
    Since $J$ is universally reversible, by \cref{prop:univ-rev}, we have $J=A^\vartheta$.
    By \cref{lem:rev-hull1}, it follows that $J = A^\vartheta = \revJstaralg(J_0)$.
\end{proof}

Finally, we give an explicit description of general representations of universally reversible J*-algebras.
As discussed above, it suffices to consider J*-factors whose representations are classified by the following:
\begin{proposition}\label{prop:reps}
    Let $J\subset L(\H)$ be a universally reversible J*-factor and set $A=\staralg(J)$. Then exactly one of the following cases holds: 
    \begin{enumerate}[(i)]
        \item\label{it:reps1} There is a unitary $u :  \H \to \CC^n \ox \CC^l$, $n\in \NN$, $l\in \NN$, such that 
            \begin{equation}
                J = u^* ( L(\CC^n) \ox \1) u,\qquad A = u^*(L(\CC^n) \ox \1) u.
            \end{equation}

        \item\label{it:reps2}
        There is a unitary $u :  \H \to (\CC^n \ox \CC^{l_1}) \oplus (\CC^n\ox \CC^{l_2}$), $2\le n\in \NN$, $l_1,l_2\in \NN$, such that
        \begin{align}
            J &= u^* \, \big\{( a \ox \1) \oplus( a^t\ox \1) \ :\ a\in L(\CC^d) \big\} \, u ,\\
            A &=  u^* \, \big\{( a \ox \1) \oplus( b \ox \1) \ :\ a,b\in L(\CC^d) \big\} \, u 
        \end{align}

        \item\label{it:reps3} There is a unitary $u: \H \to \CC^n \ox \CC^l$, $2\le n\in \NN$, $l\in \NN$, such that 
        \begin{equation}
            J = u^*( L(\CC^n)^t \ox \1 ) u,\qquad A = u^*( L(\CC^n) \ox \1 ) u.
        \end{equation}

        \item\label{it:reps4} There is a unitary $u: \H \to \CC^n \ox \CC^l$, $3\le n\in\NN$, $l\in \NN$, such that 
        \begin{equation}
            J = u^*( L(\CC^{2n})^\beta \ox \1 ) u, \qquad A = u^*( L(\CC^{2n}) \ox \1) u.
        \end{equation}
    \end{enumerate}
       Moreover, in cases \ref{it:reps2} -- \ref{it:reps4} $A$ is the universal enveloping *-algebra. In case \ref{it:reps1}, the enveloping *-algebra $A$ is not universal. 
\end{proposition}

\begin{proof}
    \emph{Case 1: $J \Jcong L(\CC^n)$.}
    The claim holds trivially if $n=1$. Let $n\ge 2$. 
    Let $\psi:J\to L(\CC^n)$ be a J*-isomorphism.
    By \cref{exa:univ-reps}, the universal enveloping *-algebra of $J$ is $(L(\CC^n)\oplus L(\CC^n),\iota)$ with $\iota:a \mapsto \psi(a)\oplus \psi(a)^t $.
    By the universal property \eqref{eq:universal-prop}, there is a *-homomorphism $\phi: L(\CC^n)\oplus L(\CC^n) \to L(\H)$ with $\phi \circ \iota |_J = \id_J$.
    Standard facts from the representation theory of *-algebras entail that a (unital) *-homomorphism of $L(\CC^n)\oplus L(\CC^n) \to L(\H)$ is of the form
    \begin{equation*}
        \phi(a\oplus b) = u^*( a \ox \1 \oplus b\ox \1) u, \qquad u :\H \to \CC^n \ox \CC^{l_1} \oplus \CC^{n}\ox \CC^{l_2},
    \end{equation*}
    where $l_1,l_2\in \NN_0$, $u$ is a unitary and where we employ the convention that $\CC^n\ox \CC^0 = 0$ to take care of the case where one of $l_1,l_2$ is zero (e.g., for the case $\phi(a\oplus b) = u^* a u$).
    Thus, we have
    \begin{equation*}
        J = \phi(\iota(J)) = \phi( \{a\oplus a^t \ : \ a\in L(\CC^n) \}) = u^* \big\{ (a\ox\1) \oplus (a^t\ox \1) \ : \ a\in L(\CC^n)\big\} u.
    \end{equation*}
    If both $l_1,l_2$ are nonzero, this entail case \ref{it:reps2}.
    If $l_1$ or $l_2$ are zero, we have $J = u^*(L(\CC^n\ox \1) u$, which is case \ref{it:reps1}.

    \emph{Case 2: $J \Jcong L(\CC^n)^t$, $n\ge 2$, or $J \Jcong L(\CC^{2n})^\beta$, $n\ge 3$.}
    In this case, the universal enveloping *-algebra is a matrix factor (see \cref{exa:univ-reps}).
    By \cref{cor:reps-if-A-factor}, this implies that we are either in case \ref{it:reps3} or in case \ref{it:reps4}.
\end{proof}

\begin{remark}\label{remark:reps-involution} 
    Note that in cases \ref{it:reps3} and \ref{it:reps4} we have (neglecting the unitary rotation $u$) $J = \hat J \ox \1$ with $\hat J \subset L(\L)$ for a Hilbert space $\L$ and $A = L(\L)\ox \1$. Moreover, $L(\L)$ is the universal enveloping *-algebra of $\hat J$, and $\hat J = \hat F(L(\L))$ with $\hat F = \frac{1}{2}(\id +\vartheta)$ for an involution $\vartheta$ on $L(\L)$. Hence $J = \frac{1}{2}(\id + \vartheta) \ox \id (A)$. 
\end{remark}

\subsection{Conditional expectations. II}\label{sec:CEs2}

We now come back to conditional expectations, but in the specific setting of universally reversible J*-algebras. 
Recall from \cref{lem:CE-extension} that every faithful conditional expectation onto a J*-algebra $J$ factorizes through $A = \staralg(J)$. It follows from the direct-sum decomposition $J = \oplus J_i$ into factors, discussed in \cref{sec:reps}, that the conditional expectation onto $A$ also decomposes into a direct sum of conditional expectations $E_i$, each of which factorizes through $A_i = \staralg(J_i)$. 
It thus suffices to consider J*-factors.

\begin{proposition}\label{prop:CE-univ-rev}
    Let $E$ be a faithful conditional expectation onto a universally reversible J*-factor $J\subset L(\H)$. Set $A = \staralg(J)$.
    Then the conditional expectation factors as
    \begin{equation}
        E: L(\H) \xrightarrow{\tilde E} A \xrightarrow{\hat E} J,
    \end{equation}
    where $\tilde E$ is a conditional expectation of $L(\H)$ onto $A$ and $\hat E=E|_A$ is a conditional expectation of $A$ onto $J$.\footnote{We only defined conditional expectations of $L(\H)$ onto J*-subalgebras of $L(\H)$, however, the definition makes sames more broadly: A conditional expectation $E$ of a J*-algebra $J$ onto a J*-subalgebra $J_0$ is a UP map $E:J\to J_0$ with $E|_{J_0}=\id$.}
    Depending on the cases listed in \cref{prop:reps}, we have:
    \begin{itemize}
        \item 
            In case \ref{it:reps1}, the Hilbert space decomposes as $\H =  \L\ox \R$ such that $J=L(\L)\ox \1$ and there is a faithful state $\omega$ on $\R$ such that the conditional expectation is of the form 
            \begin{equation}
                E = \id \ox \tr(\omega\placeholder)\1.
            \end{equation}
            Any $E$-invariant state takes the form $\sigma\ox \omega$ with $\sigma \in L(\L)$.
        \item
            In case \ref{it:reps2}, the Hilbert space decomposes as $\H = (\L\ox\R\up1)\oplus(\L\ox\R\up2)$ such that $J= \{ (a \ox \1)\oplus (a^t\ox \1) : a\in L(\L) \} $, and $A = (L(\L)\ox\1)\oplus(L(\L)\ox \1)$. There exists a $\lambda\in(0,1)$ and two faithful states $\omega\up{i}$ on $\R\up{i}$ such that the conditional expectations take the form
            \begin{equation}
                \tilde E = (\id \ox \tr(\omega\up1 \placeholder)\1)\oplus(\id \ox \tr(\omega\up2 \placeholder)\1
            \end{equation}
            and (for $a,b\in L(\L)$)
            \begin{equation}
                \hat E \big((a\ox \1)\oplus(b\ox \1)\big) = \big((\lambda a+ (1-\lambda)b^t) \ox \1\big)\oplus\big((\lambda a^t + (1-\lambda)b) \ox \1\big).
            \end{equation}
            Any $E$-invariant state takes the form $(\sigma \ox \lambda \omega\up1)\oplus(\sigma^t \ox (1-\lambda)\omega\up 2)$ with $\sigma \in L(\L)$.
        \item 
           In cases \ref{it:reps3} or \ref{it:reps4}, the Hilbert space decomposes as $\H = \L\ox \R$ such that $A=L(\L)\ox\1$ and $J = L(\L)^\vartheta\ox \1$ for an involution $\vartheta$ on $L(\L)$.
           The conditional expectations $\tilde E$ and $\hat E$ take the form
            \begin{align}
                \tilde E = \id \ox \tr(\omega\placeholder)\1,\qquad \hat E = \frac{1}{2}(\id+\vartheta)\ox \id
            \end{align}
            for a faithful state $\omega$ on $\R$. 
            Any $E$-invariant state takes the form $\sigma\ox \omega$, with $\sigma = \vartheta^* \sigma \in L(\L)$.
    \end{itemize}
\end{proposition}
\begin{proof}
    The form of $E$ in case \ref{it:reps1} and of $\tilde E$ in cases \ref{it:reps2} to \ref{it:reps4} follows directly from the form of the generated *-algebras discussed in  \cref{prop:reps} and the standard structure of conditional expectations onto *-algebras. In case \ref{it:reps1} this shows the claim since $A=J$. 
    It remains to clarify the structure of $\hat E$ in the remaining cases.

    Case \ref{it:reps2}: It was shown in \cite[Prop.~6.4,6.5]{stormerDecompositionPositiveProjections1980} (see also \cite[Lem.~4.24]{idel}) that any faithful conditional expectation must be of the given form.

    Cases \ref{it:reps2} and \ref{it:reps3}: We have $J = \hat J \ox \1$ and $A = L(\L) \ox \1$, with $L(\L) = \staralg(\hat J)$.  By \cite[Prop.~6.1]{stormerDecompositionPositiveProjections1980} there is only one conditional expectation $L(\L) \to \hat J$. Clearly, the map $\frac{1}{2}(\id +\vartheta) : L(\L) \to \hat J$ is a faithful conditional expectation, hence the unique one. Thus $\hat E$ must be of the form $\frac{1}{2}(\id + \vartheta) \ox \id$.

    The formula for invariant states follows directly from that of the conditional expectations.
\end{proof}

\subsection{\texorpdfstring{$L^p$}{Lp}-spaces}\label{sec:Lp-spaces}

For our later discussions of various quantum divergences, we develop a minimal amount of $L^p$-space theory for Jordan algebras. 
We refer to \cite{beigiSandwichedRenyiDivergence2013,jencovaPreservationQuantumRenyi2017,muller-hermesMonotonicityQuantumRelative2017,jencovaRenyiRelativeEntropies2018a,hiaiazRenyiDivergences2024} for discussions on how $L^p$-spaces relate to quantum divergences.

In the following, we only cover the absolute minimum that we require for our application. In \cite{arhancet_nonassociative_2024}, a systematic study of non-associative $L^p$-spaces for JBW-algebras has been initiated, but does not seem to cover the crucial Lemma that we need. 
We therefore restrict to the setting of universally reversible J*-algebras and make use of their explicit representation theory, leaving a general discussion of the relevant $L^p$-space theory for future work. 

Let $J\subset L(\H)$ be a J*-algebra which admits a $\sigma$-preserving conditional expectation $F$ for some faithful state $\sigma$.
Consider the map $\Gamma_\sigma: L(\H) \to L(\H)$ defined by
\begin{align}
	\Gamma_\sigma a = \sigma^{\frac12} a \sigma^{\frac12}. 
\end{align}
For $p\in \RR \setminus\{0\}$, we define
\begin{align}
	L^p(J,\sigma) := \Gamma_\sigma^{1/p}(J),\quad \norm{a}_{p,\sigma} = \norm{\Gamma_\sigma^{1/p}a}_p,\quad a\in L(\H),
\end{align}
where $\norm{x}_p = \tr(|x|^p)^{\frac{1}{p}}$ denotes the Schatten $p$-norm. Since $\sigma$ is faithful, we have $L^p(L(\H),\sigma) = L(\H)$. We extend the definition also to $p = \infty$ via $J = L^\infty(J,\sigma)$ (note that $\sigma^0 = \1$).

For $p\ge 1$, $\norm{\placeholder}_{p,\sigma}$ is a norm and a quasi-norm for $p\in(0,1)$ \cite{KALTON20031099}, but this will not be essential for us.
For $p=2$, this norm is precisely the norm that is induced by the KMS inner product (see \cref{sec:KMS}): 
\begin{align}
	\norm{a}_{2,\sigma}^2 = \tr(a^* \sigma^{\frac12} a \sigma^{\frac12}) = \langle a ,a\rangle_\sigma,   \qquad a\in L(\H).
\end{align}
Moreover, for $a,b\in L(\H)$ and $\frac{1}{p} + \frac{1}{q} =1$, we have
\begin{align}\label{eq:pq-pairing}
    \tr( (\Gamma_\sigma^{\frac{1}{q}}a)^*(\Gamma_\sigma^{\frac{1}{p}} b)) = \ip{a}{b}_\sigma. 
\end{align}
The conditional expectation $F: L(\H)\to J$ defines a positive map $L(\H) \to L^p(J,\sigma)$ via
\begin{align}
    F_p := \Gamma^{\frac{1}{p}}_\sigma \circ F \circ \Gamma^{-\frac{1}{p}}_\sigma.
\end{align}
The map $F_p$ fulfills $F_p^2  = F_p$, and $F_p(\sigma^{\frac1p}) = \sigma^{\frac1p}$.
In particular, observe that 
\begin{align}
    L^p(J,\sigma) = \mathrm{Fix}(F_p).
\end{align}
Consider another Hilbert space $\hat\H$. Let $\hat J\subset L(\hat\H)$ be a J*-algebra and let $\hat\sigma$ be a faithful state on $\hat\H$ such that $\hat J$ admits a $\hat\sigma$-preserving conditional expectation $\hat F:L(\hat\H)\to\hat J$.
Then, if $T:L(\hat \H)\to L(\H)$ is a UP map such that $T(\hat J)\subset J$ and $T^*\sigma = \hat \sigma$, the positive map 
\begin{align}
        T_p := \Gamma^{\frac{1}{p}}_\sigma \circ T \circ \Gamma_{\hat\sigma}^{-\frac{1}{p}}
\end{align}
restricts to a  map $L^p(\hat J,\hat\sigma) \to L^p(J,\sigma)$.
In the following, we denote by $R_{T,\sigma}:L(\H)\to L(\hat \H)$ the adjoint with respect to the KMS inner product of  $T$, so that
\begin{align}
    \ip{a}{T\hat b}_\sigma = \ip{R_{T,\sigma}a}{\hat b}_\sigma,\quad a\in L(\H),\ \hat b\in L(\hat \H).
\end{align}
$R_{T,\sigma}$ is known as the Petz recovery map and will be discussed in detail in \cref{sec:recovery}. 
Its (Hilbert-Schmidt) adjoint can be explicitly expressed as 
\begin{align}
    R^*_{T,\sigma} = \Gamma_\sigma \circ T \circ \Gamma_{\hat\sigma}^{-1} = \Gamma^{\frac{1}{q}}_\sigma \circ T_p\circ  \Gamma_{\hat \sigma}^{-\frac{1}{q}} = T_1
\end{align}
for any $p,q\in\RR$ with  $\frac{1}{p}+\frac{1}{q}=1$. Since $(\Gamma^{\frac{1}{p}}_\sigma)^* = \Gamma^{\frac{1}{p}}_\sigma$ it follows that
\begin{align}
    T_p^* = (R_{T,\sigma})_q.
\end{align}
Since $F$ is hermitian with respect to the KMS inner product (cp.\ \cref{lem:CE-selfadjoint}), it follows that $F_1 = F^*$. 
Hence, $L^1(J,\sigma)$ is precisely the fixed-point space of $F^*$.

The following lemma will be essential for us. 
\begin{lemma}\label{lem:Lp-properties}
	Let $J\subset L(\H)$ be a universally reversible J*-algebra admitting a $\sigma$-preserving conditional expectation for a faithful state $\sigma$.  Let $0\neq p,q \in \RR$ and $\tfrac{1}{p}+\tfrac{1}{q} = \tfrac{1}{r}$.
	Then the following properties hold:
    \begin{enumerate}
        \item\label{it:Lp-powers} $a \in L^p(J,\sigma)$ if and only if $a^p \in L^1(J,\sigma)$.
        \item\label{it:Lp-module} If $a\in L^p(J,\sigma), b\in L^q(J,\sigma)$, then $\jp{a}{b} \in L^r(J,\sigma)$.
        \item\label{it:Lp-triple} If $a \in L^p(J,\sigma), b\in L^{2q}(J,\sigma)$ then $aba \in L^{r}(J,\sigma)$.
    \end{enumerate}
\end{lemma}
\begin{proof}
	All claims follow from the explicit description of $\sigma$ and $J$ in \cref{prop:CE-univ-rev}: First, by the direct-sum decomposition of $J$ and $\sigma$, it suffices to consider the case where $J$ is a J*-factor. We are left with two cases: In the first case $\H=\L\ox \R$, $J = \hat J\otimes 1$, and $\sigma = \hat\sigma \otimes \omega$ with $\hat\sigma\in \hat J$, so that
	\begin{align}
		L^p(J,\sigma) = \{ \hat a\otimes \omega^{\frac{1}{p}}\ :\ \hat a \in \hat J\}.
	\end{align}
	In the second case, $\H = (\L\ox \R\up1)\oplus(\L\ox \R\up2)$, $J\Jcong L(\L)$, and
	\begin{equation}
    \begin{aligned}
		J &= \{ ( (\hat a \ox \1) \oplus ( \hat a^t \ox \1 )\ :\ \hat a\in L(\L) \},\\ \sigma &= (\hat\sigma \ox \lambda \omega\up1)\oplus \big(\hat\sigma^t \ox (1-\lambda)\omega\up2 \big)
    \end{aligned}
	\end{equation}
	for some $\lambda\in(0,1)$ and $\hat\sigma \in L(\L)$. Hence 
	\begin{align}
		L^p(J,\sigma) = \{\big(\hat a \otimes (\lambda\omega\up1)^{\frac{1}{p}}\big) \oplus \big( \hat a^t \otimes ((1-\lambda)\omega\up2)^{\frac{1}{p}} \ :\ a\in L(\L)\}.
	\end{align}
    In each of the two cases, all items follow by explicit calculation.
\end{proof}

\section{Faithful statistical experiments}
\label{sec:experiments}

We begin by introducing some terminology regarding statistical experiments. 
A \emph{statistical experiment} $(\rho_\theta)_{\theta\in \Theta}$ is simply a collection of states on a Hilbert space $\H$. Here $\Theta$ is some index set, which we assume to be finite (purely for simplicity, all our results can be adapted to the case where $\Theta$ is a general measure space). 
The interpretation of $\Theta$ is as a space of outcomes of an experiment which results in the preparation of state $\rho_\theta$.  The experiment is a \emph{dichotomy} if $|\Theta|=2$.
We will typically denote dichotomies as $(\rho,\sigma)$.

We say that two statistical experiments $(\rho_\theta)$ and $(\sigma_\theta)$ with the same outcome space $\Theta$, but on possibly different Hilbert spaces, are \emph{PTP-equivalent}, written $(\rho_\theta) \ptp (\sigma_\theta)$, if there exists UP maps $S,T$ such that
\begin{align}
        T^*\rho_\theta = \sigma_\theta,\quad S^*\sigma_\theta = \rho_\theta,\qquad \theta\in\Theta.
\end{align}
We also say that the two experiments can be \emph{interconverted} via PTP maps. 
Similarly, they are CPTP-equivalent if they can be interconverted by CPTP maps. In \cite{galke_sufficiency_2024}, it was shown that CPTP-equivalence can be determined from the Koashi-Imoto decomposition \cite{koashi_operations_2002} of a statistical experiment.

Suppose $(\rho_\theta) \ptp (\sigma_\theta)$ with interconverting maps $T$ and $S$. 
Then, by linearity, we have
\begin{align}
    S^* T^* \rho = \rho, \quad \rho\in \conv\{\rho_\theta\}_{\theta\in\Theta},\quad T^*S^* \sigma =\sigma,\quad \sigma\in \conv\{\sigma_\theta\}_{\theta\in\Theta},
\end{align}
where "$\conv$" denotes the convex hull. 
This shows that PTP-equivalence is really a question about the convex hull of a statistical experiment. 

In particular, we can find an enumeration of the extremal points $\rho_{\theta_i}$, $i\in I$, of $\conv\{\rho_\theta\}$ and the extremal points $\sigma_{\theta_i}$, $i \in I$ of $\conv\{\sigma_\theta\}$ such that
\begin{align}
    T^* \rho_{\theta_i} = \sigma_{\theta_i},\quad S^*\sigma_{\theta_i} = \rho_{\theta_i},\quad i \in I.
\end{align}
Set $\overline{\rho_\Theta} = \frac{1}{|I|} \sum_{i\in I}  \rho_{\theta i}$. 
Since 
\begin{align}
    \rho_{\theta_{|I|}} = n \overline{\rho_\Theta} - \sum_{i=1}^{|I|-1} \rho_{\theta i},
\end{align}
it follows that $(\rho_\theta)_{\theta\in \Theta}$ and $(\sigma_\theta)_{\theta\in \Theta}$ are PTP-equivalent if and only if
\begin{align}
    (\rho_{\theta_1},\ldots,\rho_{\theta_{|I|-1}},\overline{\rho_\Theta}) \ \ptp\    ( \sigma_{\theta_1},\ldots,\sigma_{\theta_{|I|-1}},\overline{\sigma_\theta}). 
\end{align}
Note that $ \rho_{\theta_i} \ll \overline{\rho_\Theta}$. 
Hence, for questions regarding PTP-equivalence, one can, without loss of generality, assume that any statistical experiment contains a state whose support contains the support of all the other states.

\begin{definition} 
    We call $(\rho_\theta)_{\theta\in \Theta}$ \emph{faithful} if
    for any $a\in L(\H)$, $\tr(\rho_\theta a^*a) = 0$ for all $\theta\in \Theta$ implies $a=0$.
\end{definition}

Let us define the support of a statistical experiment as
\begin{align}
    \supp \,(\rho_\theta)_{\theta\in\Theta} = \bigvee_{\theta\in\Theta} \supp \rho_\theta,
\end{align}
where, for a collection $p_i$, $i\in I$, of projections, the projection $\vee_i p_i$ is the projection onto the linear subspace spanned by the ranges of the projections $p_i$.

\begin{lemma}\label{lem:faithfulness}
    The following are equivalent:
    \begin{enumerate}[(a)]
        \item The statistical experiment $(\rho_\theta)$ is faithful;
        \item $\supp(\rho_\theta)_{\theta\in\Theta} = \1$;
        \item $\conv\{\rho_\theta\}_{\theta\in\Theta}$ contains a faithful state;
        \item The state $\sum_\theta \mu(\theta) \rho_\theta$ is faithful for any faithful probability measure $\mu$ on $\Theta$.
    \end{enumerate}
\end{lemma}
\begin{proof} 
    Immediate. 
\end{proof}

The following lemma shows that for questions regarding PTP-equivalence, we can always assume faithfulness. 
\begin{lemma}\label{lem:inter-to-faithful}
    Any statistical experiment $(\rho_\theta)$ is CPTP-equivalent to the faithful statistical experiment resulting from restricting $(\rho_\theta)$ to the subspace $\K = \supp(\rho_\theta)_\theta\H$.
\end{lemma}
\begin{proof}
        Write $s=\supp(\rho_\theta)_{\theta}$ and denote by $\iota:L(\K)\to L(\H)$ the natural embedding, so that $\iota(\1)=s$. Then $T(x) = \iota(x) + (\1-s)\tr(\sigma \iota(x))$ for $x\in L(\K)$ and $S(y) = \iota^{-1}(sys)$ for $y\in L(\H)$ define UP maps that interconvert the two statistical experiments. Both maps are clearly UCP maps.
\end{proof}
Combined with the discussion above, we conclude that we can always assume that a statistical experiment contains a faithful state.

\section{CPTP and PTP sufficiency}
\label{sec:sufficiency}
Sufficiency encodes the idea that a strict subset of all possible measurements already encodes all the relevant information of a statistical experiment $(\rho_\theta)_{\theta\in \Theta}$. 
As discussed in the introduction, there are different ways to formalize this notion.
If $T$ is a $(\rho_\theta)$-preserving UCP map on $L(\H)$, i.e., a UCP map such that
\begin{equation}\label{eq:T-preserves-states}
    T^*\rho_\theta = \rho_\theta,\qquad\theta\in \Theta,
\end{equation}
then, for every observable $x \in L(\H)$ and $\theta\in \Theta$, the observable $Tx$ has the same expectation value as $x$:
\begin{equation*}
    \tr(\rho_\theta\, Tx) = \tr(T^*\rho_\theta\,x) = \tr(\rho_\theta x),\qquad \theta\in\Theta.
\end{equation*}
Thus, if one only looks at expectation values, it is sufficient to consider observables in the range of $T$.

The range $K=\mathrm{Ran}(T)$ of a UP map $T$ is an \emph{operator system} (see \cref{def:JSA}).
An operator system is linearly spanned by the effect observables it contains, i.e., $K = \lin\, [0,\1]_K$, where $[0,\1]_K = \{ x\in K : 0\le x\le \1\}$.
We make the following definition:

\begin{definition}\label{def:sufficiency}
    An operator system $K\subset L(\H)$ is \emph{PTP-sufficient} (resp.\ \emph{CPTP-sufficient}) for a statistical experiment $(\rho_\theta)$ if there exists a $(\rho_\theta)$-preserving UP (resp.\ UCP) map $T$ on $L(\H)$ whose range is contained in $K$.
\end{definition}

We will later consider a third notion of sufficiency for operator systems, defined in terms of Bayesian hypothesis testing.
Clearly, CPTP-sufficiency implies PTP-sufficiency. 
Unsurprisingly, the converse is false in general.
However, as we will discuss shortly, the two notions become equivalent if $K$ is a *-algebra \cite{luczak_aspects_2021}.

Before we continue our discussion of sufficiency, let us recall that any statistical experiment is equivalent to a faithful one via CPTP maps. 
Hence, we can make the following assumption:

\begin{assumption*}
    Throughout this section, we assume (without loss of generality) that statistical experiments are faithful.
\end{assumption*}

If $A\subset L(\H)$ is a *-algebra, our notion of CPTP-sufficiency agrees with the notion of sufficiency studied by Petz and others \cite{petz_sufficient_1986,petz_sufficiency_1988,,ohya_quantum_1993,jencova_sufficiency_2006,kuramochi_minimal_2017}.
Petz proved that among all CPTP-sufficient *-algebras, there is a (necessarily unique) CPTP-sufficient *-algebra $A_{(\rho_\theta)}$ that is contained in all other CPTP-sufficient *-algebras.
We refer to $A_{(\rho_\theta)}$ as the \emph{minimal sufficient *-algebra}.
It is given by
\begin{equation}\label{eq:MSA}
    A_{(\rho_\theta)} \,= \bigcap \,\big\{ \,\mathrm{Fix}(T) \ : \ \text{$T$ is a $(\rho_\theta)$-preserving UCP map on $L(\H)$}\, \big\}.
\end{equation}
It is clear from \eqref{eq:MSA} that $A_{(\rho_\theta)}$ is not only minimal among CPTP-sufficient *-algebras, but also among all CPTP-sufficient operator systems.
If $T: L(\H) \to A_{(\rho_\theta)}$ is a $(\rho_\theta)$-preserving UCP map, then \eqref{eq:MSA} implies that $T$ restricts to the identity on $A_{(\rho_\theta)}$, so that $T$ must be a conditional expectation onto.\footnote{Such a conditional expectation is necessarily unique (see \cref{cor:uniquenes-CE}).}
Thus, there exists a $(\rho_\theta)$-preserving conditional expectation onto $A_{(\rho_\theta)}$, which is the unique $(\rho_\theta)$-preserving UCP map into $A_{(\rho_\theta)}$.
Let us summarize the above discussion:
\begin{equation}\label{eq:CPTP-sufficiency}
\begin{gathered}
    \textit{There is a minimal CPTP-sufficient operator system, it is a *-algebra,} \\
    \textit{and it admits a state-preserving conditional expectation}.
\end{gathered}
\end{equation}

For the PTP case, we will establish the following Theorem, which is essentially contained in the works \cite{luczak_quantum_2014,luczak_aspects_2021} of A.~\Lbar{}uczak, and partly contained in the Bachelor thesis of O.~Skodda \cite{Ole}:

\begin{theorem}\label{thm:luczak}
    Let $(\rho_\theta)_{\theta\in\Theta}$ be a faithful statistical experiment. Then:
    \begin{enumerate}[(1)]
        \item\label{it:luczak1} 
            There is a minimal PTP-sufficient operator system $J_{(\rho_\theta)}$.
            It is the J*-algebra given by
            \begin{align}\label{eq:min-suff-J}
                    J_{(\rho_\theta)} \,= \bigcap \,\big\{ \mathrm{Fix}(T) \ : \ \text{$T$ is a $(\rho_\theta)$-preserving UP map on $L(\H)$}\, \big\}.
            \end{align} 
            where the intersection is over $(\rho_\theta)$-preserving UP maps $T:L(\H)\to L(\H)$ 
        \item\label{it:luczak2} 
        There is a unique $(\rho_\theta)$-preserving UP map $F$ of $L(\H)$ into $J_{(\rho_\theta)}$, which is a conditional expectation onto $J_{(\rho_\theta)}$.
        \item\label{it:luczak3} 
        The *-algebra generated by $J_{(\rho_\theta)}$ is the minimal sufficient one, i.e.,
        \begin{equation}\label{eq:J-generates-A}
            \staralg(J_{(\rho_\theta)}) = A_{(\rho_\theta)}
        \end{equation}
    \end{enumerate}  
\end{theorem}

Since the statements in the theorem are not explicit in \cite{luczak_quantum_2014,luczak_aspects_2021}
(\Lbar{}uczak's work is concerned with W*-algebras, Jordan algebras only appear as a proving tool), we provide a self-contained proof adapted to our setting below. 

Before we come to the proof, let us discuss some important consequences.
Based on \cref{thm:luczak}, we have the following PTP-analog of \eqref{eq:CPTP-sufficiency}:
\begin{equation}\label{eq:PTP-sufficiency}
\begin{gathered}
    \textit{There is a minimal PTP-sufficient operator system, it is a J*-algebra,} \\
    \textit{and it admits a state-preserving conditional expectation}.
\end{gathered}
\end{equation}

We know that minimal (C)PTP-sufficient operator systems exist.
Thus, in principle, we have a complete understanding of (C)PTP-sufficiency: An operator system is (C)PTP-sufficient if and only if it contains the minimal (C)PTP-sufficient one.
Hence, \cref{thm:luczak} implies:

\begin{corollary}
    A *-algebra $A\subset L(\H)$ is CPTP-sufficient for $(\rho_\theta)$ if and only if it is PTP sufficient for $(\rho_\theta)$.
\end{corollary}

Thus, we can and will, in the following, drop the prefix CPTP when discussing sufficient *-algebras.

\begin{proof}[Proof of \cref{thm:luczak}]
    \Cref{it:luczak1,it:luczak2}: 
    Parts of the following argument are a PTP-version of the argument used in \cite[App.~A]{hayden_structure_2004}.
    It is clear from the definition of PTP-sufficiency that $J_{(\rho_\theta)}$, as defined in \eqref{eq:min-suff-J}, is an operator system that is contained in every other PTP-sufficient operator system.
    Since $(\rho_\theta)$ is faithful, each UP map $T$ appearing in \eqref{eq:min-suff-J} has a faithful invariant state.
    Hence, by \cref{prop:fixpoint}, each $\mathrm{Fix}(T)$ in the intersection is a J*-algebra.
    Thus, $J_{(\rho_\theta)}$ is an intersection of J*-algebras, and, hence, itself a J*-algebra.
    Next, we show that $J_{(\rho_\theta)}$ admits a $(\rho_\theta)$-preserving conditional expectation.
    Since we work in finite dimensions, we can restrict the intersection \eqref{eq:min-suff-J} to a finite set of $(\rho_\theta)$-preserving UP maps $T_1,\ldots, T_m$:
    \begin{equation*}
        J_{(\rho_\theta)} = \mathrm{Fix}(T_1)\cap \ldots \cap \mathrm{Fix}(T_m).
    \end{equation*}
    Let $F_j$ be the conditional expectation onto $\mathrm{Fix}(T_j)$ obtained from $T_j$ as in \cref{lem:cesaro-mean}. 
    Then $F_j$ is $(\rho_\theta)$-preserving, and
    $F = \lim_{n\to \infty}(F_1\cdots F_m)^n$ yields a  $(\rho_\theta)$-preserving conditional expectation onto $J_{(\rho_\theta)}$. 
    Thus, $J_{(\rho_\theta)}$ is PTP-sufficient and, therefore, the minimal PTP-sufficient operator system.
    If $T: L(\H)\to J_{(\rho_\theta)}$ is a $(\rho_\theta)$-preserving UP map, then, by \eqref{eq:min-suff-J}, the restriction of $T$ to $J_{(\rho_\theta)}$ is the identity on $J_{(\rho_\theta)}$, so that $T$ is a conditional expectation onto $J_{(\rho_\theta)}$, which is unique by \cref{cor:uniquenes-CE}.
    This finishes the proof of the first two items.
    As noted by \Lbar{}uzak in \cite{luczak_quantum_2014}, this part of the proof can also be directly deduced by applying the ergodic theorem for von Neumann algebras \cite{thomsen_invariant_1985} to the semigroup of $(\rho_\theta)$-preserving UP maps. 
    Thus, the statements can be generalized to von Neumann algebras.

    \cref{it:luczak3}: 
    Is is clear that every CPTP-sufficient *-algebra is PTP-sufficient and, hence, contains $\staralg(J_{(\rho_\theta)})$.
    We have to show that $A = \staralg(J_{(\rho_\theta)})$ is CPTP-sufficient.
	This is shown by \cref{lem:CE-extension}.    
\end{proof}

In the following, we wish to study the structure of (minimal) sufficient J*-algebras in detail. 

\begin{lemma}\label{lem:Erho}
    Let $E$ be a conditional expectation onto a J*-algebra $J\subset L(\H)$ that is sufficient for $(\rho_\theta)$, and let $T:L(\H)\to J$ be a $(\rho_\theta)$-preserving UP map. Then
    \begin{equation}
        T^*E^*\rho_\theta=\rho_\theta,\qquad \theta\in\Theta.
    \end{equation}
    In particular, we have $(\rho_\theta)\ptp (E^*\rho_\theta)$.
\end{lemma}
\begin{proof}
    This is immediate from $ET=T$ and $T^*\rho_\theta=\rho_\theta$.
\end{proof}

\begin{corollary}\label{cor:J-sufficient-CE}
Let $J\subset L(\H)$ be a J*-algebra, $E:L(\H) \to J$  a conditional expectation and $(\rho_\theta)_\theta$ a statistical experiment on $\H$. 
Then $J$ is sufficient if and only if $E:L(\H) \to J$ admits a UP map $S:L(\H) \to L(\H)$ such that
$S^* E^* \rho_\theta = \rho_\theta$.
\end{corollary}
\begin{proof}
    If $J$ is sufficient for $(\rho_\theta)$, then there exists a $(\rho_\theta)$-preserving UP map $T:L(\H) \to J\subset L(\H)$. By \cref{lem:Erho} we can set $S=T$. 
    Conversely, suppose that $E$ admits a UP map $S$ such that $S^* E^* \rho_\theta = \rho_\theta$. Then $T=ES$ is a UP map $L(\H) \to J$ such that $T^*\rho_\theta = \rho_\theta$. Hence $J$ is sufficient for $(\rho_\theta)$. 
\end{proof}

Recall that any J*-subalgebra $J\subset L(\H)$ admits a unique trace-preserving conditional expectation $E:L(\H) \to J\subset L(\H)$.
Note also that for a trace-preserving conditional expectation $E=E^*$.
 
\begin{lemma}\label{lem:CE-sufficient}
    Let $J\subset L(\H)$ be a sufficient J*-algebra for $(\rho_\theta)$ and $E:L(\H)\to J$ the trace-preserving conditional expectation. Set $\hat \rho_\theta = E\rho_\theta$. 
    Then 
    \begin{equation}\label{eq:tmp}
        \Jstaralg((\hat\rho_\theta)_{\theta\in\Theta})\subset J 
    \end{equation}
    is sufficient for $(\hat \rho_\theta)$ and $(\rho_\theta)$.
\end{lemma}

\begin{proof}
    Let $E_0$ be the trace-preserving conditional expectation onto $\Jstaralg(\hat\rho_\theta)$.
    Then $E_0 = E_0E$ implies $E_0\rho_\theta =E_0 E\rho_\theta = E_0\hat\rho_\theta =\hat\rho_\theta$.
    Let $T$ be a $(\rho_\theta)$-preserving UP map into $J$.
    \cref{lem:Erho} gives $T^*\hat\rho_\theta = \rho_\theta$.
    We set $S = E_0\circ T$ and conclude 
    \begin{equation}
        S^* \rho_\theta 
        = T^* E_0\rho_\theta 
        = T^* \hat\rho_\theta =\rho_\theta.
    \end{equation}
    Thus, $S$ is a $(\rho_\theta)$-preserving UP map into $\Jstaralg(\hat\rho_\theta)$.
    Therefore, $\Jstaralg(\hat\rho_\theta)$ is sufficient.
\end{proof}

\begin{corollary}\label{cor:TPCE-ptp-equivalence}
    Let $J_{(\rho_\theta)}\subset L(\H)$ be the minimal sufficient J*-subalgebra for $(\rho_\theta)$ on $\H$ with trace-preserving conditional expectation $E:L(\H)\to J_{(\rho_\theta)}$, and set $\hat \rho_\theta = E\rho_\theta$. Then $(\rho_\theta)\ptp (\hat \rho_\theta)$ and 
    \begin{align}
            J_{(\rho_\theta)}  = J_{(\hat\rho_\theta)} = \Jstaralg((\hat\rho_\theta )_\theta). 
    \end{align}
\end{corollary}
\begin{proof}
    The previous Lemmas shows that $\Jstaralg((\hat\rho_\theta)_\theta)$ is sufficient for $(\rho_\theta)$ and $(\hat\rho_\theta)$ and is contained in both $J_{(\rho_\theta)}$ and $J_{(\hat\rho_\theta)}$. But since the latter are minimal sufficient, we must have equality.
\end{proof}

Owing to the corollary, we can assume without loss of generality that any statistical experiment $(\rho_\theta)$ is represented by density matrices \emph{that generate the minimal sufficient J*-algebra $J_{(\rho_\theta)}$}. 

Since a J*-algebra with at most three generators is reversible \cite[Cor.~2.3.8]{hanche-olsen_jordan_1984}, we find:
\begin{corollary}\label{cor:reversible}
    The minimal sufficient J*-algebra $J_{(\rho,\sigma)}$ of a dichotomy $(\rho,\sigma)$ is reversible.
\end{corollary}

\begin{remark}[Relation to the Koashi-Iomoto decomposition]\label{rem:KI}
    The Koashi-Imoto theorem \cite{koashi_operations_2002,hayden_structure_2004} states the following:
    If $(\rho_\theta)$ is a statistical experiment on $\H$, then there exists a direct sum decomposition
    \begin{equation}
        \H = \oplus_j \H_j \ox \K_j,\quad \rho_\theta = \oplus_j \,p_{j|\theta}\, \rho_{j|\theta} \ox \omega_j
    \end{equation}
    with $(p_{j|\theta})_j$ probability distributions, and states $\rho_{j|\theta}$ on $\H_j$ and $\omega_j$  on $\K_j$,
    which has the following property: If $T$ is a UCP map on $\H$, then
    \begin{equation}
        T^*\rho_\theta = \rho_\theta
        \quad \iff\quad
        T |_{L(\H_j\ox \K_j)} = \id \ox \, T_j,
    \end{equation}
    where $T_j$ is some UCP map on $\K_j$ with $T_j^*\omega_j=\omega_j$.
    It is known \cite{jencova_sufficiency_2006,kuramochi_minimal_2017} that the Koashi-Imoto decomposition is directly related to the minimal sufficient *-algebra via 
    \begin{equation}
        A_{(\rho_\theta)} = \oplus_j L(\H_j)\ox \1_{\K_j}.
    \end{equation}
    Thus, the Koasi-Imoto decomposition is precisely the decomposition of the minimal sufficient *-algebra into factors, each of which occurs with a certain multiplicity.
    
    As for *-algebras, J*-algebras have a direct sum of simple J*-factors (J*-algebras without non-trivial J*-ideals).
    Thus, we can write $\H = \oplus_k \L_k$ and $J_{(\rho_\theta)} = \oplus_k J_k$ with each $J_k\subset L(\L_k)$ a J*-factor.
    We know from \cref{thm:luczak} that the minimal sufficient J*-algebra $J_{(\rho_\theta)}$ is a subalgebra of $A_{(\rho_\theta)}$, which generates $A_{(\rho_\theta)}$ as a *-algebra.
    Thus, we have
    $\oplus_j L(\H_j)\ox \1_{\K_j} 
    = \oplus_k  \mathrm{\text{*-}alg}(J_k).$
    Therefore, for every $k$, there exists a set $I(k)$ such that 
    \begin{equation}
        \H = \oplus_k \oplus_{j\in I(k)}
        \H_j \ox \K_j,
        \quad
        J_{(\rho_\theta)} = \oplus_k J_k,
        \quad 
        \mathrm{\text{*-}alg}(J_k) = \oplus_{j\in I(k)} L(\H_j)\ox\1_{\K_j}.
    \end{equation}
    Looking at the representation theory of finite-dimensional Jordan algebras \cite{jacobson_structure_1968,hanche-olsen_jordan_1984,idel}, we see that the enveloping *-algebra of a J*-factor is either a full matrix algebra or a direct sum of two matrix algebras.
    Thus, for each $k$, $I(k)$ consists of either one or two elements.
\end{remark}

Let us discuss an example where the minimal sufficient J*-algebra significantly deviates from the minimal sufficient *-algebra. 

\begin{example}
    Let $(\rho_\theta)$ be an irreducible family of states on a Hilbert space $\H$, i.e., a family with the property that the only subspaces $\L\subset \H$ that are jointly invariant under all $\rho_\theta$ are $\L = \{0\}$ and $\L = \H$.
    Then the minimal sufficient *-algebra is 
    \begin{equation}
        A_{(\rho_\theta)} = L(\H).
    \end{equation}
    Indeed, this follows from the Koashi-Imoto decomposition discussed above.
    (In fact, the Koashi-Imoto decomposition also implies the converse, so that
    $A_{(\rho_\theta)}=L(\H)$ if and only if $(\rho_\theta)$ is irreducible.)
    However, if there exists a basis with respect to which $\rho_\theta = \rho_\theta^t$ for all $\theta$,
    then $E = \frac12(\id + (\placeholder)^t)$ is a conditional expectation with $E^*\rho_\theta = \rho_\theta$ for all $\theta$.
    Then, the symmetric matrices are a sufficient J*-algebra, and we have
    \begin{equation}\label{eq:J-sub-symm-sub-A}
        J_{(\rho_\theta)} \subseteq \{x \in L(\H)\ : \ x=x^t \} \subsetneq L(\H) = A_{(\rho_\theta)}.
    \end{equation}
    An example of a family of states with these properties is the dichotomy $(\rho,\sigma)$ on $\H=\CC^2$ with $\rho = \ketbra00 = \frac12(\1+ Z)$ and $\sigma = \ketbra ++ = \frac12(\1+X)$, where $X,Y,Z$ denote the Pauli matrices.
    In this specific example, the first inclusion in \eqref{eq:J-sub-symm-sub-A} becomes an equality.
\end{example}

\begin{example}\label{exa:trp-doubling}
    Let $(\rho_\theta)$ be a family of states on $\H$. 
    For $0 < \lambda < 1$, consider the states $(\lambda\rho_\theta \oplus (1-\lambda) \rho_\theta^t)$ on $\H\oplus \H$, where the transpose is taken in some arbitrary basis on $\H$.
    We claim that the minimal sufficient J*-algebra is given by
    \begin{equation}\label{eq:trp-doubling}
        J_{(\lambda\rho_\theta\oplus (1-\lambda)\rho_\theta^t)} = \{ x\oplus x^t \ : \ x\in J_{(\rho_\theta)}\} \Jcong J_{(\rho_\theta)},
    \end{equation}
    where the isomorphism is an isomorphism of J*-algebras.
    In particular, if $J_{(\rho_\theta)}=L(\H)$, we have
    \begin{equation}\label{eq:trp-doubling2}
        J_{(\lambda\rho_\theta\oplus (1-\lambda)\rho_\theta^t)} = \{ x\oplus x^t \ : \ x\in L(\H)\} \Jcong L(\H),
    \end{equation}
    We postpone the proof of this claim until later, when we can give a short and elegant proof based on the algebraic structure of PTP-interconvertibility.
\end{example}

\section{The algebraic structure of PTP-equivalence}

We are now in a position to clarify the algebraic structure of PTP-equivalence of faithful statistical experiments.
The first part of the following Theorem first appeared in the Bachelor thesis of O. Skodda \cite{Ole}. 

\begin{theorem}\label{thm:ptp-inter}
    Let $(\rho_\theta)_{\theta \in \Theta}$ and $(\rho'_\theta)_{\theta \in \Theta}$ be faithful statistical experiments on Hilbert spaces $\H$ and $\H'$.
    \begin{enumerate}[(i)]
    \item\label{it:ptp-inter1}
        $(\rho_\theta)_{\theta \in \Theta}$ and $(\rho'_\theta)_{\theta \in \Theta}$ are PTP-interconvertible if and only if there is a J*-isomorphism 
        \begin{equation*}
            \psi : J_{(\rho_\theta)} \to J_{(\rho'_\theta)}
        \end{equation*}
        such that 
        \begin{equation}\label{eq:intertwining-iso}
            \tr(\rho_\theta' \psi(a)) = \tr(\rho_\theta a), \qquad a\in J_{(\rho_\theta)},\ \theta\in \Theta.
        \end{equation}
    \item\label{it:ptp-inter2}
        The J*-isomorphism $\psi$ is uniquely determined by \eqref{eq:intertwining-iso}.
    \item\label{it:ptp-inter3}
        If $T,S$ are UP maps such that $T^*\rho_\theta = \rho'_\theta$ and $S^*\rho'_\theta = \rho_\theta$ for all $\theta\in \Theta$,
        then the restriction of $T$ to $J_{(\rho_\theta)}$ is $\psi$ and the restriction of $S$ to $J_{(\rho'_\theta)}$ is $\psi^{-1}$.
        I.e., the following diagram commutes
        \begin{equation}
        \begin{tikzcd}
            L(\H) \arrow{d}{T} & J_{(\rho_\theta)} \arrow[hook']{l}{} \arrow[leftarrow,shift left]{d}{}{\psi^{-1}} \arrow[rightarrow, shift right,swap]{d}{\psi} \arrow[hook]{r}{} & L(\H)  \\
            L(\H') & \arrow[hook']{l}{} J_{(\rho_\theta')} \arrow[hook]{r}{} & L(\H') \arrow{u}{S}
        \end{tikzcd}
        \end{equation}
    \end{enumerate}
\end{theorem}

For the proof, we need a characterization of J*-isomorphisms in terms of the order structure.

\begin{lemma}\label{lem:order-iso}
    Let $J_1,J_2$ be J*-algebras, and let $T:J_1\to J_2$ be a *-preserving unital linear map.
    Then $T$ is a J*-isomorphism if and only if it is an order isomorphism, i.e., $T$ is a linear bijection such that $T(J_1^+)=J_2^+$.
\end{lemma}
\begin{proof}
    If $T$ is a J*-isomorphism, then \eqref{eq:positive-cone-squares} shows that $T(J_1^+)= J_2^+$. 
    To show the converse, it is enough to establish that $T$ is a J*-homomorphism.
    Since $T$ is *-preserving, \eqref{eq:square trick} shows that we only have to prove that $(Tx)^2 = T(x^2)$ for hermitian $x\in J_1$.
    Let $x= \sum_i \lambda_i p_i$ be the spectral decomposition of a hermitian element $x=x^*\in J_1$ and denote by $[0,\1]_{J} = \{ a : 0\leq a \leq \1, a\in J\}$ the unit interval of a J*-algebra $J$.
    Since $T$ restricts to an affine bijection of the convex sets $[0,\1]_{J_1}$ and $[0,\1]_{J_2}$, it maps extreme points onto extreme points.
	The extreme points of the unit interval of a J*-algebra are precisely the projections.\footnote{For every positive element $a\geq 0$, we have $a = \sum_j \lambda_j p_j$ with eigenvalues $\lambda_j\geq 0$ and orthogonal projections $p_j\in J$, see \cref{sec:basics}.}
    Thus, $T$ maps projections to projections.
    Set $q_i = Tp_i$.
    If $i\ne j$, then $p_i$ and $p_j$ are orthogonal. Thus, $p_i + p_j$ is a projection, and, hence, $q_i + q_j = Tp_i + Tp_j = T(p_i + p_j)$ is a projection.
    The sum of two projections is a projection if and only if the projections are orthogonal.
    Therefore $q_i$ and $q_j$ are orthogonal, so that $Tx = \sum_i \lambda_i q_i$ must be the spectral decomposition of $Tx$.
    In particular, we have $(Tx)^2 = \sum_i \lambda_i^2 q_i = \sum_i \lambda_i^2 Tp_i = T(x^2)$.
\end{proof}

\begin{proof}[Proof of \cref{it:ptp-inter1} of \cref{thm:ptp-inter}]
    We let $J\up\prime$ denote the minimal sufficient J*-algebras of $(\rho_\theta\up\prime)$ and let $F\up\prime:L(\H\up\prime) \to J\up\prime$ denote the $(\rho_\theta\up\prime)$-preserving conditional expectations onto it.
    
    Assume  $(\rho_\theta)\ptp(\rho_\theta')$ and let $T:L(\H)\to L(\H')$, $S:L(\H')\to L(\H)$ be interconverting UP maps.
    If necessary, we replace $T$ and $S$ by the maps $F'\circ T$ and $F\circ S$ to obtain interconverting UP maps whose ranges are contained in $J'$ and $J$, respectively.
    The UP maps $T\circ S$ and $S\circ T$ satisfy $(T\circ S)^*\rho_\theta = \rho_\theta$ and $(S\circ T)^*\rho_\theta'=\rho_\theta'$.
    Thus, by \cref{thm:luczak}, $ (T\circ S)|_{J'} = \id_{J'}$ and $S\circ T|_J = \id_{J}$.
    We define $\psi = T|_J : J \to J'$ and $\psi' = S|_{J'} :J' \to J$.
    Then $\psi$ and $\psi'$ are unital positive maps that are inverses of each other.
    Therefore, $\psi$ is an order isomorphism and $\psi' = \psi^{-1}$.
    By \cref{lem:order-iso}, $\psi$ is also a J*-isomorphism.
    \Cref{eq:intertwining-iso} follows from the construction of $\psi$.

    Conversely, suppose that an isomorphism $\psi$ as in the statement is given. 
    Define the UP maps $T = \psi^{-1} \circ F'$ and $S = \psi\circ F$. Then
    \begin{align}
        \tr(\rho'_\theta S(a)) = \tr(\rho'_\theta \psi(F(a))) = \tr(\rho_\theta F(a)) = \tr(\rho_\theta a),\quad a\in L(\H),
    \end{align}
    and similarly $T^*\rho_\theta = \rho_\theta'$.
\end{proof}

\begin{lemma}\label{lem:restricted-iso}
    Let $\hat J \subset L(\hat\H), J \subset L(\H)$ be J*-subalgebras. Let $T:\hat J \to L(\H)$ be a UP map and $E:L(\H) \to J$ a conditional expectation such that $\alpha = E\circ T$ is a J*-isomorphism $\hat J \to J$. Then $T(\hat J) = J$ and $T=\alpha$.
\end{lemma}
\begin{proof}[Proof of \cref{lem:restricted-iso}]
  Consider the UP maps $\beta = T \circ \alpha^{-1} : J \to T(\hat J)$ and $F = \beta \circ E: L(\H) \to T(\hat J)$. 
Then we have
\begin{align}
    E\circ F &= E, \quad F^2=F, \quad 
    F\circ E = F,\quad F\circ T = T.
\end{align}
Now let $\sigma$ be a faithful state on $L(\H)$.
By the first equality, $E^*\sigma$ is a faithful $F$-and $E$-invariant state. By the second equality, $F$ is a conditional expectation onto its range, which is a J*-algebra $F(L(\H))\subset T(\hat J)$. By the last equality, we have $T(\hat J) \subset F(L(\H))$. 
Hence $F$ is a conditional expectation onto $T(\hat J)$.
By a standard argument, we have $E = \lim_{n\to\infty} (E\circ F)^n = \lim_{n\to\infty} (F\circ E)^{n} = F$ and hence $J = T(\hat J)$ as well as $\alpha = E\circ T = F\circ T = T$.
\end{proof}

\begin{proof}[Proof of \cref{it:ptp-inter2,it:ptp-inter3} of \cref{thm:ptp-inter}]
    We begin by showing the uniqueness of $\psi$.
    Let $\psi_1,\psi_2$ be J*-isomorphisms intertwining the states as specified in \cref{thm:ptp-inter}, and let $E : L(\H)\to J_{(\rho_\theta)}$ be the $(\rho_\theta)$-preserving conditional expectation.
    Then $T = \psi_2^{-1}\circ \psi_1\circ E$ is a UP map on $L(\H)$ with $T^*\rho_\theta=\rho_\theta$.
    By \cref{thm:luczak}, this implies that $T$ restricts to the identity on $J_{(\rho_\theta)}$. Thus, we have $\psi_1^{-1} = \psi_2^{-1}$ and, hence, $\psi_1=\psi_2$.
    
    Let $E$ and $E'$ be the conditional expectations onto $J_{(\rho_\theta)}$ and $J_{(\rho'_\theta)}$, respectively.
    The proof of \cref{thm:ptp-inter} shows that $ET$ and $E'S$ restrict to mutually inverse J*-isomorphisms between the two minimal sufficient J*-algebras. The result now follows from \cref{lem:restricted-iso}.
\end{proof}

Using \cref{thm:ptp-inter}, we can now provide the missing proof in \cref{exa:trp-doubling}:

\begin{example}[continues=exa:trp-doubling]\label{exa:final-trp-doubling}
    We identify $L(\H\oplus \H)$ with the space $M_2(L(\H))$ of $L(\H)$-valued $2\times 2$-matrices in the natural way.
    Then the families $(\rho_\theta)$ and $(\lambda\rho_\theta\oplus(1-\lambda)\rho_\theta^t)$ are interconverted by the PTP maps $T:L(\H)\to L(\H\oplus\H)$ and $S:L(\H\oplus\H)\to L(\H)$ given by
    \footnote{Indeed, we have
    $\tr(\tfrac12(\rho_\theta\oplus\rho_\theta^t)\,Ta)
        = \tfrac12 (\tr(\rho_\theta a) + \tr(\rho_\theta^t a^t)) = \tr(\rho_\theta a)$, $a\in L(\H)$,
    and 
        $\tr(\rho_\theta\, Sb)
        = (\tr(\rho_\theta (\lambda b_{11} + (1-\lambda)b_{22}^t))
        = \lambda \tr(\rho_\theta b_{11}) + (1-\lambda)\tr(\rho_\theta^t b_{22})
        = \tr(\lambda\rho_\theta\oplus(1-\lambda)\rho_\theta^t) \,b)$, $b\in M_2(L(\H))=L(\H\oplus\H)$.
    }
    \begin{equation*}
        Ta = a \oplus a^t, \qquad Sb = \lambda b_{11}+(1-\lambda) b_{22}^t ,\qquad a\in L(\H),\ b = [b_{ij}]\in M_2(L(\H)).
    \end{equation*}
    Thus, by \cref{thm:ptp-inter}, the minimal sufficient J*-algebra of $(\lambda\rho_\theta\oplus (1-\lambda) \rho_\theta^t)$ is the image of $J_{(\rho_\theta)}$ under $T$, which is exactly the right-hand side of \eqref{eq:trp-doubling}.
\end{example}

We also mention the following result, but postpone the proof until Section \cref{sec:standard-rep}.
Recall that a positive map is called \emph{decomposable} if it is the sum of a CP and a coCP map (a CP map followed by transposition in some basis).

\begin{theorem}\label{thm:decomposable-interconversion}
    Two dichotomies $(\rho,\sigma)$ and $(\tau,\omega)$ are PTP-interconvertible if and only if there exist decomposable UP maps $T$ and $S$ such that $(T^*\rho,T^*\sigma)=(\tau,\omega)$ and $(S\tau,S\omega)=(\rho,\sigma)$.
\end{theorem}

\section{Petz recovery maps}
\label{sec:recovery}
We now take a small detour and study recovery maps. 
In the following, $\sigma$ is always a faithful state and $T:L(\hat\H)\to L(\H)$ is a UP map such that $\hat \sigma = T^*(\sigma)$ is faithful as well. We will discuss later that this is essentially without loss of generality, see \cref{remark:faithfulness-petz}.

The Petz recovery map $R=R_{T,\sigma}:L(\H)\to L(\hat\H)$ of $T$ with respect to $\sigma$ is defined as 
\begin{equation}\label{eq:def-recovery-map}
    R(a) = \hat\sigma^{-\frac12} T^*\big(\sigma^{\frac12}x \sigma^{\frac12}\big) \hat\sigma^{-\frac12}, \qquad a\in L(\H).
\end{equation}
Note that $R^*(\hat\sigma)=\sigma$ by construction.
Equivalently, we can define $R$ as the adjoint of $T$ with respect to the KMS inner products relative to $\sigma$ and $\hat\sigma$ (see \cref{sec:KMS}):
\begin{equation}
    \ip{\hat a}{T(a)}_{\hat\sigma} = \ip{R(\hat a)}{a}_{\sigma}, \qquad a\in L(\H),\ \hat a\in L(\hat\H).
\end{equation}
In particular, this shows that the Petz recovery map of $R$ (relative to $\hat \sigma$) is again $T$,
\begin{align}
	R_{R,\hat\sigma}a &= (R^*\hat\sigma)^{-\frac12} R^*(\hat\sigma^{\frac12} a\hat\sigma^{\frac12}) (R^*\hat\sigma)^{-\frac12} = Ta,
\end{align}
and that the Petz recovery map behaves well under the composition of maps. 
Given a state $\rho$ on $\H$, our aim in this section is to characterize when $R_{T,\sigma}$ recovers $\rho$ from $\hat\sigma$. 
The following observation is essential for us.

\begin{lemma}\label{lem:observable-recovery}
    Let $a\in L(\hat\H)$. The following are equivalent:
    \begin{enumerate}[(a)]
        \item\label{it:observable-recovery1} there exists a UP map $S:L(\H)\to L(\hat\H)$ with $ST(a)=a$ and $S^*\hat\sigma = \sigma$;
        \item\label{it:observable-recovery2} $\ip{T(a)}{T(a)}_{\sigma}=\ip aa_{\hat\sigma}$;
        \item\label{it:observable-recovery3} $R_{T,\sigma} T(a) = a$.
    \end{enumerate}
\end{lemma}
\begin{proof}
    \ref{it:observable-recovery3} $\Rightarrow$ \ref{it:observable-recovery1} is clear.
    \ref{it:observable-recovery1} $\Rightarrow$ \ref{it:observable-recovery2}:
    Denote by $\norm\placeholder_{2,\sigma}$ the norm induced by the KMS inner product.
    By \cref{cor:UP-contraction}, both $T$ and $S$ are contractions between the respective KMS Hilbert spaces.
    Thus, we have $\norm{a}_{2,\hat\sigma} =\norm{ST(a)}_{2,\hat\sigma}\le \norm{T(a)}_{2,\sigma} \le \norm{a}_{2,\hat\sigma}$.
    Hence, we have $\norm{T(a)}_{2,\sigma}=\norm{a}_{2,\hat\sigma}$, which is just another way of writing to \ref{it:observable-recovery2}.

    \ref{it:observable-recovery2} $\Rightarrow$ \ref{it:observable-recovery3}:
    It is a general fact that if $A$ is a contraction between Hilbert spaces and $\norm{A\xi}=\norm{\xi}$ for some vector $\xi$, then $A^*A\xi=\xi$.
    Thus, the claim follows because $R_{T,\sigma}$ is the adjoint operator of $T$ with respect to the KMS inner products.
\end{proof}

In Petz's original work on recovery maps \cite{petz_sufficiency_1988}, the Connes cocycles $\rho^{it}\sigma^{-it}$ play an essential role. However, these cannot be constructed from $\rho$ and $\sigma$ using only the Jordan product. 
To state our main theorem characterizing Petz recovery for UP maps, we use the operator
\begin{equation}
    d_{\rho|\sigma} = \sigma^{-\frac12}\rho\sigma^{-\frac12},
\end{equation}
which can be seen as a symmetrized version of $\rho^{\frac12}\sigma^{-\frac12}$.
It should not be confused with the relative modular operator $\rho\otimes\sigma^{-1}$. 
By construction, we have
\begin{equation}\label{eq:R-of-Delta}
    R_{T,\sigma} d_{\rho|\sigma} = d_{T^*\rho |T*\sigma}. 
\end{equation}
and 
\begin{equation}\label{eq:dual-of-petz}
    R_{T,\sigma}^*(\hat\rho) = \sigma^{\frac12} T(d_{\hat\rho| T^*\sigma}) \sigma^{\frac12}
\end{equation}
for any density operator $\hat\rho\in L(\hat \H)$.

\begin{definition}\label{def:sufficient-map}
    We say that a UP map $T:L(\hat \H) \to L(\H)$ is \emph{sufficient for $(\rho,\sigma)$} if there exists a UP map $S:L(\H) \to L(\hat \H)$ such that $S^* T^*\rho = \rho$ and $S^* T^*\sigma=\sigma$. 
\end{definition}

\begin{theorem}\label{thm:algebraic-recovery}
    Let $(\rho,\sigma)$ be a dichotomy on $\H$ with $\sigma$ faithful, $T:L(\hat\H)\to L(\H)$ be a UP map such that $\hat\sigma = T^*\sigma$ is faithful and set $\hat \rho = T^*\rho$.
    The following are equivalent:
    \begin{enumerate}[(a)]
        \item\label{it:algebraic-recovery1} 
        $T$ is sufficient for $(\rho,\sigma)$.
        
        \item\label{it:algebraic-recovery2} 
        $R_{T,\sigma}^*(\hat\rho) = \rho$, were $R_{T,\sigma}$ is the Petz recovery map of $T$ relative to $\sigma$.
        \item\label{it:algebraic-recovery3} $T(d_{\hat\rho|\hat\sigma}) \le d_{\rho|\sigma}$, or $T(d_{\hat\rho|\hat\sigma}) \ge d_{\rho|\sigma}$.
        \item\label{it:algebraic-recovery4} $T(d_{\hat\rho|\hat\sigma}) = d_{\rho|\sigma}$.
        \item\label{it:algebraic-recovery5} $T$ restricts to an isomorphism $J_{(\hat\rho,\hat\sigma)} \to J_{(\rho,\sigma)}$. 
    \end{enumerate}
    If, in addition, $\rho$, $\hat\rho$ are full rank states, then the above statements are equivalent to the analogous statements for $(\sigma,\rho)$ instead of $(\rho,\sigma)$.
\end{theorem}
\begin{proof}
    We set $d = d_{\rho|\sigma}$, $\hat d=d_{\hat\rho|\hat\sigma}$, and $R = R_{T,\sigma}$.
    Note that $R(d)=\hat d$ (see \eqref{eq:R-of-Delta}).
    \ref{it:algebraic-recovery2} $\Rightarrow$ \ref{it:algebraic-recovery1}, and \ref{it:algebraic-recovery4} $\Rightarrow$ \ref{it:algebraic-recovery3} are trivial.

    \ref{it:algebraic-recovery1} $\Rightarrow$  \ref{it:algebraic-recovery4}:
    Write $P = R_{S,\hat\sigma}$.
    Then we have $P(\hat d)= d$ by \cref{eq:R-of-Delta} and, hence, $P R( d)= d$. Since $P^*(R^*\hat \sigma) = P^*\sigma = P^* S^*\hat\sigma = \sigma$, 
    \cref{lem:observable-recovery} shows $T R( d) =  d$, where we used that the Petz recovery map of $R$ is $T$.
    Due to \eqref{eq:R-of-Delta}, we conclude $T(\hat d)= d$.

    \ref{it:algebraic-recovery3} $\Rightarrow$ \ref{it:algebraic-recovery4}:
    First, assume the "$\le$" inequality. Notice that
    \begin{align}
        0 \leq \tr(\sigma( d - T(\hat d))) = \tr(\sigma d) - \tr(\hat\sigma\hat d) = 0.
    \end{align}
    Since $\sigma$ is faithful, we thus have $T(\hat d) =  d$.
    The case of the other inequality follows analogously.

    \ref{it:algebraic-recovery4} $\Rightarrow$ \ref{it:algebraic-recovery2}:
    By definition of the Petz recovery map, we have $R^*\hat\sigma=\sigma$ and (by \cref{eq:dual-of-petz})
    \begin{equation}
        R^*\hat\rho  = \sigma^{\frac12}T(\hat d)\sigma^{\frac12} = \sigma^{\frac12} d\sigma^{\frac12} = \rho.
    \end{equation}

    \ref{it:algebraic-recovery5} $\Leftrightarrow$ \ref{it:algebraic-recovery1} is shown in \cref{thm:ptp-inter}. 

    Assume now that $\rho$ and $\hat\rho$ have full rank. 
    The symmetry of \ref{it:algebraic-recovery1} under the exchange $\rho\leftrightarrow \sigma$ implies that we may exchange the states in the other statements as well.
\end{proof}
Now suppose that $J\subset L(\H)$ is a J*-algebra that admits a $\sigma$-preserving conditional expectation $F$. 
Eq.\ \eqref{eq:state-preserving-CE} shows that 
\begin{align}
		F = R_{E,\sigma},\quad E = R_{F,E\sigma},
\end{align}
where $E$ is the trace-preserving conditional expectation onto $J$.

\begin{corollary}\label{cor:d-in-J}
	Let $J\subset L(\H)$ be a J*-algebra with $\sigma$-preserving conditional expectation $F$. 
	Then $J$ is sufficient for $(\rho,\sigma)$ if and only if $d_{\rho|\sigma}\in J$ if and only if $\rho \in L^1(J,\sigma)$.
	In this case, we have
	\begin{align}
		d_{\rho|\sigma} = d_{E\rho|E\sigma},
	\end{align}
	where $E$ is the trace-preserving conditional expectation onto $J$.
	In particular, $d_{\rho|\sigma} \in J_{(\rho,\sigma)}$.
\end{corollary}
\begin{proof}
	If $J$ is sufficient for $(\rho,\sigma)$, then $F$ is $(\rho,\sigma)$-preserving and the identity is a recovery map for $F$. Hence, by \cref{thm:algebraic-recovery}, $Fd_{F^*\rho|F^*\sigma} = F d_{\rho|\sigma} = d_{\rho|\sigma}$. Therefore $d_{\rho|\sigma}\in J$. 
	Conversely, suppose that $d_{\rho|\sigma} \in J$ and denote by $E$ the trace-preserving conditional expectation onto $J$. Since $F$ is the recovery map of $E$ relative to $\sigma$ and $d_{\rho|\sigma}\in J$, we know that 
	\begin{align}
		d_{\rho|\sigma} = F d_{\rho|\sigma} = d_{E \rho|E\sigma}.
	\end{align}
	Hence $Ed_{E\rho|E\sigma} = d_{\rho|\sigma}$, so that the recovery map $F$ of $E$ fulfills $F^* E\rho = \rho$. Since $F^*E^* = (EF)^* = F^*$, we find $F^*\rho=\rho$ and $J$ is sufficient for $(\rho,\sigma)$. 
    That $d_{\rho|\sigma} \in J$ if and only if $\rho \in L^1(J,\sigma)$ follows immediately from the definition of $L^1(J,\sigma)$ in \cref{sec:Lp-spaces}.
\end{proof}

Let now $\hat\rho = E\rho,\hat\sigma = E\sigma$, where $E$ is the trace-preserving conditional expectation onto $J_{(\rho,\sigma)} = \Jstaralg(\hat\rho,\hat\sigma)$.
Note that the triple product $\{ \hat\sigma^{\frac12}, d_{\hat\rho|\hat\sigma},\hat\sigma^{\frac12}\} = \hat\rho$  can be constructed using just the Jordan product.
As a consequence of the previous corollary and $J_{(\rho,\sigma)} = J_{(\hat\rho,\hat\sigma)} = \Jstaralg(\hat\rho,\hat\sigma)$ we thus find that $J_{(\rho,\sigma)}$ is also generated by $\hat\sigma$ and $d_{\hat\rho|\hat\sigma}$.

Recall from \cref{sec:CEs} that the fixed-point space of a UP map $T:L(\H) \to L(\H)$ with faithful invariant state $\sigma$ is a J*-algebra which admits a $\sigma$-preserving conditional expectation given by the Cesaro mean of $T$ (cp.\ \cref{lem:cesaro-mean}).
This yields the following corollary:

\begin{corollary}\label{cor:fixed-point-recovery}
    Let $\sigma$ be a faithful state on $\H$ and let $T:L(\hat\H)\to L(\H)$ be a UP map.
    Set $W = T \circ  R_{T,\sigma}$. Then the fixed-point space $\mathrm{Fix}(W)$ is a J*-algebra.
    Consider the $\sigma$-preserving conditional expectation $F:L(\H) \to \mathrm{Fix}(W)$ and the trace-preserving conditional expectation $E:L(\H) \to \mathrm{Fix}(W)$.
    The following are equivalent for a state $\rho$ on $\H$:
    \begin{enumerate}[(a)]
    \item $T$ is sufficient for $(\rho,\sigma)$.
    
    \item $\mathrm{Fix}(W)$ is sufficient for $(\rho,\sigma)$.
        \item $W^* \rho = \rho$.
        \item $F^*\rho = \rho$.
        \item $E W^* E\rho = E\rho$.
        \item $W d_{\rho|\sigma} = d_{\rho|\sigma}$, i.e., $d_{\rho|\sigma} \in \mathrm{Fix}(W)$.
        \item $W d_{E\rho|E\sigma} = d_{E\rho|E\sigma}$.
    \end{enumerate}
\end{corollary}

\begin{remark}[Faithfulness]\label{remark:faithfulness-petz}
    In \cref{sec:experiments}, we saw that $(\rho,\sigma) \ptp (\rho',\sigma')$ is equivalent to $(\rho, \tfrac{1}{2}(\rho+\sigma))\ptp (\rho', \tfrac{1}{2}(\rho'+\sigma'))$. This shows that in a situation where $\sigma$ is not faithful, but the dichotomy $(\rho,\sigma)$ is faithful we may always replace $\sigma$ with $\tfrac{1}{2}(\rho+\sigma)$ in the equivalent statements \ref{it:algebraic-recovery2} -- \ref{it:algebraic-recovery5} in \cref{thm:algebraic-recovery}.

     If the dichotomy in question is not faithful, we can always first interconvert to a faithful dichotomy and then apply \cref{thm:algebraic-recovery}. 
\end{remark}

\section{The standard representation of a statistical experiment}
\label{sec:standard-rep}

In this section, we introduce a particularly useful representation of a statistical experiment, which we refer to as the standard representation.\footnote{Unrelated to the notion of standard representation of a von Neumann algebra.}
We believe that this representation is of general interest for statistical experiments.
We use it to prove two results on dichotomies. 
The first result, which will be essential for our later study of Bayesian hypothesis testing, is that the minimal sufficient J*-algebra of any faithful dichotomy is a universally reversible J*-algebra.
Second, we will provide the proof of \cref{thm:decomposable-interconversion}.

We have seen in \cref{sec:experiments} that for questions related to PTP-equivalence we can assume without loss of generality that a statistical experiment is faithful (and in fact that it contains a faithful state). 
\cref{thm:ptp-inter} together with \cref{cor:TPCE-ptp-equivalence} shows that PTP-equivalence classes are completely described by the minimal sufficient J*-algebra $J_{(\rho_\theta)}$ and the reduced states $E\rho_\theta\in J_{(\rho_\theta)}$, where $E$ is the trace-preserving conditional expectation onto $J_{(\rho_\theta)}$.

We claim that any faithful representation $\pi: J_{(\rho_\theta)} \to L(\hat\H)$ yields a PTP-equivalent dichotomy $(\hat\rho_\theta)$ on $\hat\H$:
To see this, set $\hat J = \pi(J)$. Let $\hat E:L(\hat\H) \to \hat J$ be the trace-preserving conditional expectation onto $\hat J$ and let $F:L(\H) \to J_{(\rho_\theta)}$ be the $(\rho_\theta)$-preserving conditional expectation.
We can then define UP maps
\begin{align}
    T = \pi^{-1}\circ \hat E:L(\hat\H) \to J_{(\rho_\theta)} \subset L(\H),\quad S = \pi\circ F: L(\H)\to \hat J \subset L(\hat\H),
\end{align}
where $\pi^{-1}$ is restricted to $\hat J$. 
Define a statistical experiment $(\hat\rho_\theta) = (T^*\rho_\theta)$ on $\hat\H$. 
Since $\hat E \circ \pi(a) = \pi(a)$ for any $a\in J_{(\rho_\theta)}$ we have $\hat E \circ \pi \circ F = \pi\circ F$. It follows that $S^* \circ T^* = (T\circ S)^* = F^*$ and hence $S^*\hat \rho_\theta = F^* \rho_\theta = \rho_\theta.$  Thus $(\rho_\theta) \ptp (\hat\rho_\theta)$, $\hat J = J_{(\hat \rho_\theta)}$ and $\hat\rho_\theta \in \hat J$.

Let now $(\rho_\theta)$ be a, not necessarily faithful, statistical experiment on $\H$. 
By first restricting to the support of the experiment (cp.\ \cref{lem:inter-to-faithful}), converting via the trace-preserving conditional expectation on the resulting minimal sufficient J*-algebra and then using a universal representation of the resulting minimal sufficient J*-algebra (cp.\ \cref{def:univ-rep}), we obtain a PTP-equivalent statistical experiment $(\hat\rho_\theta)$. 
We refer to $(\hat\rho_\theta)$ as a \emph{standard representation} of the statistical experiment $(\rho_\theta)$.
Its importance stems from the following list of properties:

\begin{proposition}\label{prop:standard-rep}
    The standard representation $(\hat\rho_\theta)$ of a statistical experiment $(\rho_\theta)$ has the following properties:
    \begin{enumerate}
        \item\label{it:univ1} $(\hat\rho_\theta)$ is faithful and if $\rho_\theta$ is faithful for some $\theta\in \Theta$, so is $\hat\rho_\theta$.
        
        \item\label{it:univ2} 
        $J_{(\hat\rho_\theta)} = \pi(J_{(\rho_\theta)}) = \Jstaralg((\hat\rho_\theta)_\theta)$ 
        and the unique $(\hat\rho_\theta)$-preserving conditional expectation onto $J_{(\hat\rho_\theta)}$ is the trace-preserving one.

        \item\label{it:univ3} 
        $J_{(\hat\rho_\theta)} \subset L(\H)$ is in its universal representation. The generated *-algebra is the minimal sufficient *-algebra $A_{(\hat\rho_\theta)} = \staralg(J_{\hat\rho_\theta}) = \staralg((\hat\rho_\theta)_{\theta})$.
         
        \item\label{it:univ4} If $(\sigma_\theta)$ is another statistical experiment (possibly on another Hilbert space) with standard representation $(\hat\sigma_\theta)$, then $(\rho_\theta) \ptp (\sigma_\theta)$ if and only if $(\hat \rho_\theta)$ and $(\hat\sigma_\theta)$ are \emph{unitarily} equivalent.  
    \end{enumerate}
\end{proposition}
\begin{proof}
    By construction of $(\hat\rho_\theta)$ we can without loss of generality assume that $(\rho_\theta)$ is a faithful statistical experiment with $E\rho_\theta = \rho_\theta$, where $E$ is the trace-preserving conditional expectation onto $J_{(\rho_\theta)}$, which coincides with the $(\rho_\theta)$-preserving conditional expectation $F$. 
    We denote by $\pi:J_{(\rho_\theta)} \to L(\hat\H)$ the universal representation. We then have 
    \begin{align}\label{eq:proof-std-rep}
        \tr(\hat\rho_\theta \hat a) = \tr(\rho_\theta \pi^{-1} (\hat E (\hat a))),\quad \hat a\in L(\hat\H). 
    \end{align}

     \cref{it:univ1}: Suppose $0\le \hat a\in L(\hat\H)$ and
    $0= \tr(\hat\rho_\theta\, \hat a)$ for all $\theta\in\Theta$. 
    Since $(\rho_\theta)$ is faithful this requires $\pi^{-1}\hat Ea =0$.
    Since $\pi^{-1}$ is an isomorphism when restricted to the range of $\hat E$, this requires $\hat Ea=0$. But $\hat E$ is faithful (because it is trace-preserving), and therefore $a=0$. Hence $(\hat\rho_\theta)$ is faithful.
    
    \cref{it:univ2}: By \cref{thm:ptp-inter}, the PTP-equivalence of $(\rho_\theta)$ and $(\hat\rho_\theta)$ implies that $J_{(\hat\rho_\theta)}=\pi(J_{(\rho,\sigma})$.
    The statement that the $(\hat\rho_\theta)$-preserving conditional expectation is trace-preserving holds by construction, and the equality $\Jstaralg((\hat\rho_\theta)_\theta)= J_{(\hat\rho_\theta)}$ is shown in \cref{cor:TPCE-ptp-equivalence}.
    
    \cref{it:univ3}: This follows by construction from \cref{thm:luczak} and \cref{it:univ2}. 

    \cref{it:univ4}:
    Clearly if $(\hat\rho_\theta)\subset L(\H)$ and $(\hat\sigma_\theta) \subset L(\K)$ are unitarily equivalent, then $(\rho_\theta)$ and $(\sigma_\theta)$ are PTP-equivalent.
    To see the converse, let $(\rho_\theta) \ptp (\sigma_\theta)$ and hence $(\hat\rho_\theta)\ptp(\hat\sigma_\theta)$, where $\hat\rho_\theta \in L(\hat\H)$ and $\hat\sigma_\theta\in L(\hat\K)$ are the standard representations of the statistical experiments.
    \Cref{thm:ptp-inter} shows that there is an isomorphism $\psi : J_{(\hat\rho_\theta)}\to J_{(\hat\sigma_\theta)}$, which intertwines the expectation values of the two statistical experiments.
    Since both J*-algebras are in their universal representations, by \cref{lem:unitary-equivalence-univ} the J*-isomorphism $\psi$ is unitarily implemented.
\end{proof}

\begin{proposition}\label{prop:MSJA-uni-rev}
    Let $(\rho,\sigma)$ be a faithful dichotomy on a Hilbert space $\H$. 
    Then $J_{(\rho,\sigma)}$ is a universally reversible J*-algebra.
\end{proposition}

\begin{proof}
    By \cref{prop:standard-rep}, the universal representation $J_{(\hat\rho,\hat\sigma)}$ of $J_{(\rho,\sigma)}$ is generated by two hermitian elements, hence reversible by \cref{cor:reversible}. By \cref{prop:univ-rev} $J_{(\rho,\sigma)}$ is reversible in every representation.
\end{proof}

We use the universal representation to prove \cref{thm:decomposable-interconversion}.
The following result is due to St\o rmer (see \cite{ha_positive_2002} for a stronger version of the statement).
\begin{lemma}[{\cite[Cor.~7.3]{stormerDecompositionPositiveProjections1980}}]\label{lem:reversible-decomposable-CE}
    Let $J\subset L(\H)$ be a J*-algebra and $E:L(\H)\to J$ a faithful conditional expectation. Then $E$ is decomposable if and only if $J$ is reversible. 
\end{lemma}

\begin{lemma}\label{lem:univ-decomposable}
    The interconversion between a faithful dichotomy $(\rho,\sigma)$ and its standard representation $(\hat\rho,\hat\sigma)$ can always be achieved with decomposable PTP maps.
\end{lemma}

\begin{proof}
    By the universal property of the universal embedding $J_{(\hat\rho,\hat\sigma)}\subset A_{(\hat\rho,\hat\sigma)}$, the J*-isomorphism $\pi^{-1} : J_{(\hat\rho,\hat\sigma)} \to J_{(\rho,\sigma)}$ lifts to a *-homomorphism $\pi : A_{(\hat\rho,\hat\sigma)} \to A_{(\rho,\sigma)}$ (where we used that $A_{(\rho,\sigma)}$ is generated by $J_{(\rho,\sigma)}$, see \cref{thm:luczak}).
    By \cref{cor:reversible}, $J_{(\hat\rho,\hat\sigma)}\subset L(\hat\H)$ is reversible.
    By \cref{lem:reversible-decomposable-CE}, this implies that the trace-preserving conditional expectation $\hat E: L(\hat\H) \to J_{(\hat\rho,\hat\sigma)}$ is decomposable.
    Then $T =  \pi^{-1} \circ \hat E$ satisfies $(T^*\rho,T^*\sigma)=(\hat\rho,\hat\sigma)$ and is decomposable as a composition of decomposable maps.
    
    We cannot use the same argument to construct the converse map.
    However, we know that there exists a UP map $S:L(\H)\to L(\hat\H)$ with $(S^*\hat\rho,S^*\hat\sigma)=(\rho,\sigma)$.
    By \cref{thm:algebraic-recovery}, we can take $S = R_{T,\sigma}$ as the Petz recovery map. 
    The definition of the Petz recovery map makes it evident that decomposability of $T$ implies decomposability of $R_{T,\sigma}$ (see \eqref{eq:def-recovery-map}).
    This finishes the proof.
\end{proof}

\begin{proof}[Proof of \cref{thm:decomposable-interconversion}]
   We only need to prove that PTP-equivalence implies that the two dichotomies can be mapped into each other via decomposable maps.
    Consider the following diagram:
    \begin{equation}
    \begin{tikzcd}
        (\rho,\sigma) \arrow[leftrightarrow]{r}{} \arrow[leftrightarrow]{d}{} & (\tau,\omega) \arrow[leftrightarrow]{d}{} \\
        (\hat\rho,\hat\sigma) \arrow[leftrightarrow]{r}{} & (\hat\tau,\hat\omega) 
    \end{tikzcd}
    \end{equation}
    By \cref{lem:univ-decomposable}, the vertical interconversions can be achieved with decomposable PTP maps, and, by \cref{prop:standard-rep}, the lower horizontal interconversion can be achieved with a unitary.
    Thus, the upper horizontal conversion is indeed possible with decomposable PTP maps.
\end{proof}

\section{Sufficiency and Bayesian hypothesis testing}

\label{sec:bayes}
Suppose a quantum system is known to be either in the quantum state $\rho$ or in the state $\sigma$.
To distinguish the two cases, one performs a binary measurement $(x,\1-x)$, where $x$ (resp.\ $\1-x$) is associated with $\rho$ (resp.\ $\sigma$).
There are four possible cases, two error cases (mistaking $\rho$ for $\sigma$ or $\sigma$ for $\rho$) and two success cases (correctly returning $\rho$ or correctly returning $\sigma$).
If $(\rho,\sigma)$ occur (or are believed to occur) with probabilities $(p,1-p)$ with $p\in(0,1)$, the success probability is
\begin{equation}\label{eq:P-succ1}
    p \tr(\rho x) + (1-p)\tr(\sigma(\1-x)) 
       = (1-p) +p\tr((\rho - \tfrac{1-p}{p}\sigma)x))
\end{equation}
Maximizing the success probability over all possible measurements, we obtain
\begin{equation}\label{eq:P-succ2}
 P_{\text{succ.},p}(\rho,\sigma) = 
1-p +  p\tr((\rho-\tfrac{1-p}{p}\sigma)^+)  
 = p + (1-p)\tr((\sigma-\tfrac{p}{1-p}\rho)^+),
\end{equation}
where the second equality follows from (writing $t=(1-p)/p$)
\begin{equation*}
    pt +  p\tr((\rho-t\sigma)^+)  
 = p\left(t + t\tr((\sigma-\tfrac{1}{t}\rho)^-)\right) 
 = p\left(1 + t(\tr(\sigma-\tfrac{1}{t}\rho)^+)\right) 
 = p + pt\tr((\sigma-\tfrac{1}{t}\rho)^+).
\end{equation*}
Let $s = \supp(\rho,\sigma)$ be the support projection of the dichotomy $(\rho,\sigma)$. Then for any test $x$
\begin{align}
    \tr((\rho-t\sigma)x) = \tr(s(\rho-t\sigma)s x) = \tr((\rho-t\sigma) sxs).
\end{align}
Therefore, tests can always be restricted to the subspace $\K = \supp(\rho,\sigma)$ without affecting the success probability. This means that Bayesian hypothesis is only concerned with the subspace $\K$. 
In other words, we can assume without loss of generality that all dichotomies are faithful.

In the following, we will express the success probability using the \emph{Hockey stick divergence}, defined as \cite{hirche_quantum_2023,hirche_quantum_2024}
\begin{equation}\label{eq:hockey-stick-def}
    E_t(\rho\|\sigma) = \max_{0\leq x\leq 1} 
    \tr((\rho-t\sigma)x) - (1-t)^+
  =
    \tr((\rho-t\sigma)^+) - (1-t)^+,\qquad t\ge 0.
\end{equation}
The term $(1-t)^+=\max\{1-t,0\}$ in the definition ensures that $E_t(\rho\|\rho) = 0$ for all $t\ge 0$. 

Using the Hockey stick divergence, the success probability takes the form
\begin{equation}
P_{\text{succ.},p}(\rho,\sigma) = \begin{cases} 
(1-p)+pE_{(1-p)/p}(\rho\|\sigma)),&\text{if } p \leq \tfrac12\\
p + p E_{(1-p)/p}(\rho\|\sigma),&\text{if } p\geq \tfrac12.
\end{cases}
\end{equation}
Since $E_t$ fulfills
$E_t(\rho\|\sigma) = tE_{1/t}(\sigma\|\rho)$ \cite{hirche_quantum_2024}, the apparent asymmetry between $\rho$ and $\sigma$ is only superficial.
Evidently, the Hockey stick divergences and the success probability fulfill the data-processing inequality: If $T:L(\K)\to L(\H)$ is a UP map, we have 
\begin{align}
    E_t(T^*\rho\|T^*\sigma) - E_t(\rho\|\sigma) = \max_{0\leq T(x) \leq \1} \tr((\rho-t \sigma)T(x)) - \max_{0\leq x \leq \1} \tr((\rho-t \sigma)x) \leq 0.
\end{align}
In the context of Bayesian hypothesis testing, the probability distribution $(p,1-p)$ indicating the (believed) probabilities of the states $(\rho,\sigma)$, is called the \emph{prior}.
We define a notion of sufficiency by asking that the optimal success probability is achieved with measurements in $K$:

\begin{definition}\label{def:bayes-suff}
    An operator system $K\subset L(\H)$ is sufficient for Bayesian hypothesis testing, \emph{Bayes-sufficient} for short,
    for a dichotomy $(\rho,\sigma)$ if, for all priors, the optimal success probabilities can be achieved with measurements in $K$.%
    \footnote{\label{footnote:check} Bayes-sufficiency has appeared in the literature before, e.g., in \cite{jencova_quantum_2010} under the name "2-sufficiency".}
\end{definition}

By \cref{eq:P-succ1,eq:P-succ2}, an operator system $K$ is Bayes-sufficient precisely when the variational formula definition of the Hockey stick divergence can be restricted to $K$, i.e.,
\begin{equation}\label{eq:E-opt-over-K}
    E_t(\rho\|\sigma) = \max_{\substack{x\in K\\ 0\le x\le 1}} \tr((\rho-t\sigma)x)-(1-t)^+, \qquad t>0.
\end{equation}

We make a short detour introducing certain families of projections that depend on the joint spectral properties of hermitian operators.
To continue, we need introduce some notation.
If $\bullet$ is a translation-invariant binary relation on $\RR$ (e.g., "$=$" or "$\le $") and if $a,b\in L(\H)$ are hermitian operators, we define a projection $[a\bullet b]$ on $\H$ via the functional calculus:
\begin{equation}
    [a \bullet b] := \chi_{\bullet}(a-b),
\end{equation}
where $\chi_\bullet$ is the characteristic function of the set of numbers $\lambda\in\RR$ such that $\lambda \bullet 0$.
By construction, one has $[a+\lambda \bullet b] = [a\bullet b-\lambda]$, $\lambda\in\RR$, and $[\lambda a \bullet b] = [a\bullet \lambda^{-1} b]$ for $\lambda>0$.
Let us consider some examples:
$[a = b]$ denotes the projection onto $\ker(a-b)$, and $[a>b]$ the projection onto the positive part of $a-b$.
By the functional calculus, we have $[a\ge b] = [a=b] + [a>b]$. 

We return to the properties of Hockey stick divergences.
They are related to the families of projections discussed above via
\begin{equation}
    E_t(\rho\|\sigma) = \tr\big((\rho-t\sigma) \,\proj{\rho>t\sigma}\big) -(1-t)^+= \tr\big((\rho-t\sigma) \,[\rho\ge t\sigma]\big)-(1-t)^+, \qquad t>0.
\end{equation}
We will refer to the projections $\proj{\rho>t\sigma}$ and $[\sigma>t\rho]$, $t>0$, as the \emph{Neyman-Pearson tests}.
We see that the variational formula for $E_t(\rho\|\sigma)$ attains its optimum at the Neyman-Pearson test $\proj{\rho>t\sigma}$.
We need to understand the structure of general optimizers. The following Lemma is taken from \cite{jencova_quantum_2010}:\footnote{The faithfulness assumption in \cite{jencova_quantum_2010} is not used in the proof of this statement. }

\begin{lemma}\label{lem:bayes-optimizers}
    Let $0\le x\le \1$ be an effect, and let $0<p<1$. 
    The following are equivalent:
    \begin{enumerate}[(a)]
        \item the success probability is optimal with respect to the prior $(p,1-p)$,
        \item $\tr((\rho-t\sigma)x) = E_t(\rho\|\sigma)+(1-t)^+$ for $t= \frac{1-p}p$,
        \item $\proj{\rho>t\sigma} \le x \le \proj{\rho\ge t\sigma}$ for $t= \frac{1-p}p$. 
    \end{enumerate}
\end{lemma}

Next, we show that a unique minimal Bayes-sufficient operator system exists:

\begin{proposition}\label{prop:min-bayes-suff}
    An operator system $K\subset L(\H)$ is Bayes sufficient for a  dichotomy $(\rho,\sigma)$ if and only if it contains the Neyman-Pearson tests:
    \begin{equation}
        \proj{\rho>t\sigma} \in K, \qquad t>0.
    \end{equation}
    In particular, there is a minimal Bayes-sufficient operator system:
    \begin{equation}\label{eq:min-bayes-suff}
        K_{(\rho,\sigma)} = \lin \big\{ \proj{\rho>t\sigma}\ : \ t>0\,\big\} + \CC\cdot \1.
    \end{equation}
\end{proposition}

The scalars only appear in \eqref{eq:min-bayes-suff} because we consider operator systems, which are unital by definition.
For the proof of \cref{prop:min-bayes-suff}, we need the following Lemma, essentially taken from \cite{liu_layer_2025}:

\begin{lemma}\label{lem:pencil}
    For any pair of hermitian operators $a,b\in L(\H)$, the map $t\mapsto [a>tb]$ is right continuous, while $t\mapsto [a\ge tb]$ left continuous.
    The two functions have finitely many (at most $\dim(\H)$) points of discontinuity and coincide at all other points
    \begin{equation}
        [a>tb] = [a \ge tb].
    \end{equation}
\end{lemma}
\begin{proof}
    As noted in the proof of \cite[Lem.~2.1]{liu_layer_2025}, the spectral decomposition 
    \begin{equation}
        a-tb = \sum_i \lambda_i(t) \kettbra{v_i(t)},\quad t\in \RR,
    \end{equation}
    with analytic eigenpairs $(\lambda_i(t),v_i(t))$ has the property that each eigenvalue function $\lambda_i(t)$ is monotonically decreasing in $t$.
    Let $\bullet$ denote either "$>$" or "$\ge$".
    We have
    \begin{equation}
        [a \bullet tb] = \sum_i \chi_\bullet(\lambda_i(t)) \, \kettbra{v_i(t)},\qquad t>0,
    \end{equation}
    where $\chi_>$ and $\chi_\ge$ are the characteristic function of $(0,\oo)$ and $[0,\oo)$, respectively.
    Since $\chi_>$ is left continuous and the $\lambda_i(t)$ are continuous and monotonically decreasing, $[a>tb]$ is right continuous in $t$. 
    The left continuity of $[a\ge tb]$ follows analogously.
    The points of discontinuity are the $t$ values for which (at least) one of the eigenvalue functions $\lambda_i(t)$ is zero.
    Away from points of discontinuity, the two cases coincide.
\end{proof}

\begin{proof}[Proof of \cref{prop:min-bayes-suff}]
    Clearly, $K$ is Bayes-sufficient if $\proj{\rho>t\sigma}\in K$ for all $t>0$. We show the converse:
    If $K$ is sufficient, \cref{lem:bayes-optimizers} implies that, for each $t>0$, there is an operator $x=x^*\in K$ with $\proj{\rho>t\sigma}\le x\le \proj{\rho \ge t\sigma}$.
    By \cref{lem:pencil}, $\proj{\rho>t\sigma}=[\rho\ge t\sigma]$ for all but finitely many $t>0$.
    Thus, Bayes-sufficiency requires $\proj{\rho>t\sigma}\in K$ for all but finitely many $t>0$, but $\proj{\rho>t\sigma}$ is right continuous, so that $\proj{\rho>t\sigma}\in K$  must hold for all $t>0$.
\end{proof}

In the remainder of this section, we study how the minimal Bayes-sufficient operator system $K_{(\rho,\sigma)}$ relates to sufficient (J)*-algebras. 
It is clear that, for an operator system $K\subset L(\H)$, CPTP-sufficiency implies PTP-sufficiency.
However, it is not immediately clear that PTP-sufficiency also implies Bayes-sufficiency.
We will show the following theorem:

\begin{theorem}\label{thm:K-generates}
    Let $(\rho,\sigma)$ be a faithful dichotomy on a Hilbert space $\H$. 
    Then the minimal Bayes-sufficient operator system fulfills
     \begin{equation}\label{eq:min-suff-inclusions}
         K_{(\rho,\sigma)} \subset J_{(\rho,\sigma)} \subset A_{(\rho,\sigma)}
     \end{equation}
     and 
    \begin{equation}\label{eq:K-generates}
        J_{(\rho,\sigma)} = \Jstaralg(K_{(\rho,\sigma)}), 
        \qquad A_{(\rho,\sigma)} = \staralg(K_{(\rho,\sigma)}).
    \end{equation}
\end{theorem}

Combining \cref{prop:min-bayes-suff} and \cref{thm:K-generates}, we immediately obtain the following:

\begin{corollary}
     Let $(\rho,\sigma)$ be a faithful dichotomy on $\H$ and $J\subset L(\H)$ be a J*-algebra (in particular, $J$ could be a *-algebra).
     The following are equivalent:
     \begin{enumerate}[(a)]
         \item $J$ is sufficient;
         \item $J$ is Bayes-sufficient;
         \item $[\rho>t\sigma]\in J$ for all $t>0$.
     \end{enumerate}
\end{corollary}

In fact, the following corollary shows that the minimal sufficient J*-algebra of any statistical experiment is generated by Neyman-Pearson tests.
\begin{corollary}\label{cor:general-suff-generated-neyman-pearson}
    Let $(\rho_\theta)_{\theta\in \Theta}$ be a faithful statistical experiment, $\mu:\Theta\to[0,1]$ a faithful probability distribution and set $\omega = \sum_\theta \mu(\theta) \rho_\theta$. 
    Then
    \begin{align}
        J_{(\rho_\theta)} = \Jstaralg((\proj{\rho_\theta > t \omega})_{t>0,\theta\in\Theta}).
    \end{align}
\end{corollary}

\begin{proof}
Set $J =  \Jstaralg((\proj{\rho_\theta > t \omega})_{t>0,\theta\in\Theta})$. Since $J_{(\rho_\theta)}$ is sufficient, we have $\proj{\rho_\theta > t \omega}\in J_{(\rho_\theta,\omega)}$ for each fixed $\theta$. Since $J_{(\rho_\theta,\omega)} \subset J_{(\rho_\theta)_\theta}$, we have $J \subset J_{(\rho_\theta)}$. 
It remains to show that $J$ is sufficient. 
Let $E$ be the trace-reserving conditional expectation onto $J$. Since $J$ is sufficient for each faithful dichotomy $(\rho_\theta,\omega)$, by \cref{cor:J-sufficient-CE} there exists a recovery map $S_\theta$ such that $S_\theta^* E^* \rho_\theta = \rho_\theta$ and $S_\theta^* E^* \omega = \omega$. By \cref{thm:algebraic-recovery} we can choose $S_\theta = R_{E,\omega}=:R$ independently of $\theta$.  Thus $R^* E^* \rho_\theta = \rho_\theta$ for all $\theta$ and $J$ is sufficient by \cref{cor:J-sufficient-CE}.
\end{proof}

For the proof of \cref{thm:K-generates}, we need some preparations.
We start with the following:

\begin{lemma}\label{lem:Pt-preservation}
    Let $T:L(\K) \to L(\H)$ be a UP map, let $(\rho,\sigma)$ be a dichotomy on $\H$ and set $\hat\rho = T^*\rho$, $\hat\sigma = T^*\sigma$.
    The following are equivalent:
    \begin{enumerate}[(a)]
        \item\label{it:Pt-preservation1} $E_t(\rho\|\sigma) = E_t(\hat \rho\|\hat\sigma)$ for all $t> 0$,
        \item\label{it:Pt-preservation2}
        $T\proj{\hat\rho>t \hat\sigma} = \proj{\rho>t\sigma}$ for all $t>0$.
    \end{enumerate}
 \end{lemma}
\begin{proof}
    \ref{it:Pt-preservation2} $\Rightarrow$ \ref{it:Pt-preservation1}:
    Let $t\geq 0$. We have 
    \begin{align}
        E_t(\hat\rho\|\hat\sigma) +(1-t)^+ &= \tr(T^*(\rho-t\sigma) \proj{\hat\rho>t\hat\sigma} ) = \tr((\rho-t\sigma)T\proj{\hat\rho>t\hat\sigma}) \nonumber\\
        &= \tr((\rho-t\sigma)\proj{\rho>t\sigma}) = E_t(\rho\|\sigma)+(1-t)^+ .
    \end{align}
    
    \ref{it:Pt-preservation1} $\Rightarrow$ \ref{it:Pt-preservation2}:
    Clearly $0\le T\proj{\hat\rho>t\hat\sigma}\le \1$. 
    For $t\ge0$, we have
    \begin{align}
        E_t(\rho,\sigma) &\ge \tr((\rho-t\sigma) T\proj{\hat\rho>t\hat\sigma}) -(1-t)^+\nonumber\\
        &= \tr(T^*(\rho-t \sigma) \proj{\hat\rho>t\hat\sigma}) -(1-t)^+
        = E_t(\hat\rho,\hat\sigma) = E_t(\rho,\sigma).
    \end{align}
    Thus, equality holds.
    By \cref{lem:bayes-optimizers}, this implies 
    \begin{equation}
        \proj{\rho>t\sigma}\le T\proj{\hat\rho>t\hat\sigma} \le [\rho\ge t\sigma],\qquad t>0.
    \end{equation}
    By \cref{lem:pencil}, the lower and upper bounds coincide for all but finitely many $t>0$. Thus, $T[\hat \rho>t\hat\sigma]=[\rho>t\sigma]$ holds for all but finitely many $t>0$, so that the right continuity (see \cref{lem:pencil}) implies the equality holds for all $t>0$.
\end{proof}

\begin{corollary}\label{cor:restriction-iso-on-K}
	Let $T:L(\K) \to L(\H)$ be a UP map, and let $(\rho,\sigma)$ be a faithful dichotomy on $\H$ such that $(\hat\rho,\hat\sigma)=(T^*\rho,T^*\sigma)$ is faithful as well. 
    Suppose that $E_t(\rho\|\sigma) = E_t(\hat\rho\| \hat\sigma)$ for all $t> 0$. Then $T$ restricts to a J*-isomorphism $\psi:\Jstaralg(K_{(\hat\rho,\hat\sigma)}) \to \Jstaralg(K_{(\rho,\sigma)})$ with
    \begin{equation}\label{eq:restriction-iso-on-K}
        \psi(\proj{\hat\rho>t\hat\sigma}) =\proj{\rho>t\sigma},\qquad t>0.
    \end{equation}
\end{corollary}
\begin{proof}
	From \cref{lem:Pt-preservation} we find that $T\proj{\hat\rho > t \hat\sigma} = \proj{\rho>t \sigma}$. 
	Hence, the Neyman-Pearson tests $\proj{\hat\rho > t\hat\sigma}$ are in the multiplicative domain of $T$. 
	Thus $T$ restricts to a J*-homomorphism $\psi$ from $\Jstaralg(K_{(\hat\rho,\hat\sigma)})$ into $L(\H)$. 
	Its range is $\Jstaralg(K_{(\rho,\sigma)})$ because it maps generators to generators. 
	Since $T^*$ maps the faithful state $\tfrac{1}{2}(\rho+\sigma)$ to the faithful state $\tfrac{1}{2}(\hat\rho+\hat\sigma)$, $T$ is a faithful map. Hence, its restriction $\psi$ is a J*-isomorphism.  
\end{proof}


\begin{lemma}\label{lem:bayes-suff}
    Let $(\rho,\sigma)$ be a faithful dichotomy on $\H$, and let $E$ be a faithful conditional expectation onto a J*-algebra $J \subset L(\H)$. Set $(\hat\rho,\hat\sigma)=(E^*\rho,E^*\sigma)$.
    The following are equivalent:  
    \begin{enumerate}[(a)]
        \item\label{it:bayes-suff1} $E_t(\hat\rho \| \hat\sigma)=E_t(\rho \| \sigma)$ for all $t>0$;
        \item\label{it:bayes-suff2} $\proj{\rho>t\sigma} \in J$ for all $t>0$.
    \end{enumerate}
    If these hold, then 
    \begin{equation}
        \proj{\rho>t\sigma} = \proj{\hat\rho>t \hat\sigma},\qquad t>0.
    \end{equation}
\end{lemma}
\begin{proof}
    \ref{it:bayes-suff2} $\Rightarrow$ \ref{it:bayes-suff1}:
    By data-processing, we have $E_t(\hat\rho\|\hat\sigma) \le E_t(\rho\|\sigma)$. 
    The converse inequality readily follows from our assumption and the variational formula for the Hockey stick divergence (see \eqref{eq:hockey-stick-def}):
    \begin{align*}
        E_t(\rho\|\sigma) +(1-t)^+ = \tr((\rho-t\sigma) \proj{\rho>t\sigma})
        &= \tr((\rho-t\sigma) E\proj{\rho>t\sigma})\\
        &= \tr((\hat\rho-t\hat\sigma)\proj{\rho>t\sigma}) \le E_t(\hat\rho\| \hat\sigma) +(1-t)^+ .
    \end{align*}

    \ref{it:bayes-suff1} $\Rightarrow$ \ref{it:bayes-suff2}:
    It suffices to show $\proj{\rho>t\sigma}\in \mathrm{Ran}(E)$.
    Indeed, by \cref{lem:Pt-preservation}, we have
    \begin{equation}\label{eq:siudfhpweiu}
        E(\proj{\hat\rho>t \hat\sigma}) = \proj{\rho>t\sigma}, 
        \qquad t>0.
    \end{equation}

    We now check the last claim:
    From \eqref{eq:siudfhpweiu}, we learn that $\proj{\hat\rho> t \hat\sigma}$ is in the multiplicative domain of $E$.
    By \cref{lem:faithful-CE-multi-domain}, the multiplicative domain of $E$ is $J$.
    Thus, using \eqref{eq:siudfhpweiu}, we get 
    $\proj{ \hat\rho > t \hat\sigma} =E(\proj{ \hat\rho > t \hat\sigma}) = \proj{\rho>t\sigma}$.
\end{proof}

\begin{corollary}\label{cor:K-in-J}
    Let $(\rho,\sigma)$ be a faithful dichotomy on $\H$, then
    \begin{equation}
        \proj{\rho>t\sigma} \in J_{(\rho,\sigma)}, \qquad t>0.
    \end{equation}
    In particular, we have
        $K_{(\rho,\sigma)} \subset \Jstaralg(K_{(\rho,\sigma)})\subset J_{(\rho,\sigma)}$.
\end{corollary}
\begin{proof}
    Apply \cref{lem:bayes-suff} to the $(\rho,\sigma)$-preserving conditional expectation $F$ onto $J_{(\rho,\sigma)}$.
\end{proof}

\begin{lemma}\label{lem:J-of-K-is-rev}
    The J*-algebra $\Jstaralg(K_{(\rho,\sigma)})$ generated by $K_{(\rho,\sigma)}$ is reversible.
\end{lemma}
\begin{proof}
    Let $E:L(\H)\to \Jstaralg(K_{(\rho,\sigma)})$ be a faithful conditional expectation. 
    Since evidently $\proj{\rho>t\sigma}\in \Jstaralg(K_{(\rho,\sigma)})$ for all $t>0$, \cref{lem:bayes-suff} gives $\proj{E^*\rho > t E^*\sigma} = \proj{\rho>t\sigma}$ for $t > 0$.
    We show reversibility by showing
    \begin{equation}\label{eq:iufhawie}
        \Jstaralg(K_{(\rho,\sigma)}) =\Jstaralg(K_{(E^*\rho,E^*\sigma)}) = J_{(E^*\rho,E^*\sigma)}
    \end{equation}
    and appealing to \cref{cor:reversible}.
    The first equality follows from $\proj{E^*\rho > t E^*\sigma} = \proj{\rho>t\sigma}$.
    To see $\Jstaralg(K_{(\rho,\sigma)})  \supseteq J_{(E^*\rho,E^*\sigma)}$, we only have to check that $\Jstaralg(K_{(\rho,\sigma)})$ is sufficient for $(E^*\rho,E^*\sigma)$, but this is clear since $E$ is a $(E^*\rho,E^*\sigma)$-preserving conditional expectation onto $\Jstaralg(K_{(\rho,\sigma)})$.
    \cref{cor:K-in-J} shows the converse inclusion $\Jstaralg(K_{(E^*\rho,E^*\sigma)}) \subset J_{(E^*\rho,E^*\sigma)}$.
\end{proof}

Another ingredient that we need for the proof of \cref{thm:K-generates} is Frenkel's integral formula, which relates the quantum relative entropy to the Hockey stick divergence \cite{frenkel_integral_2023,hirche_quantum_2024} 
\begin{equation}\label{eq:Frenkel}
    D(\rho\|\sigma) = \int_1^\oo \bigg(\frac1t E_t(\rho\|\sigma) + \frac1{t^2} E_t(\sigma\|\rho)\bigg) \,dt.
\end{equation}
Note that Frenkel's formula implies that the relative entropy satisfies the data-processing inequality for PTP-maps, see also \cite{muller-hermesMonotonicityQuantumRelative2017}.

\begin{proof}[Proof of \cref{thm:K-generates}]
The inclusion $K_{(\rho,\sigma)}\subset J_{(\rho,\sigma)}$ is shown in \cref{cor:K-in-J}, and the inclusion $J_{(\rho,\sigma)} \subset A_{(\rho,\sigma)}$ is shown in \cref{thm:luczak}. Thus, we have already established \eqref{eq:min-suff-inclusions}. It remains to show \cref{eq:K-generates}.
We proceed in three steps

\emph{Step 1.}
    We show that $A=\staralg(K_{(\rho,\sigma)})$ is sufficient.
    Let $E$ be a faithful conditional expectation onto $A$. 
    Then $E$ is completely positive since $A$ is a *-algebra. 
    Since $E\proj{\rho>t\sigma}= \proj{\rho>t\sigma}$ for all $t>0$, \cref{lem:bayes-suff} implies $E_t(E^*\rho\|E^*\sigma) = E_t(\rho\|\sigma)$ for all $t>0$.
    By Frenkel's formula \eqref{eq:Frenkel}, we have $D(E^*\rho\|E^*\sigma)=D(\rho\|\sigma)$.
    Since $E^*$ is CPTP, the equality $D(E^*\rho\|E^*\sigma)=D(\rho\|\sigma)$ and Petz's theorem imply that there is a recovery map \cite{petz_sufficient_1986,jencova_sufficiency_2006}.
    In particular, since the range of $E$ is $A$, $A$ is a sufficient *-algebra.

\emph{Step 2.}
    We show $A_{(\rho,\sigma)} = \staralg(K_{(\rho,\sigma)})$.
    In step 1, we showed that $\staralg(K_{(\rho,\sigma)}$ is sufficient. Thus, we have $A_{(\rho,\sigma)}\subseteq \staralg(K_{(\rho,\sigma)})$. 
    Since we already established $K_{(\rho,\sigma)} \subseteq A_{(\rho,\sigma)}$, the result follows.

\emph{Step 3.}
    We show $\Jstaralg(K_{(\rho,\sigma)})=J_{(\rho,\sigma)}$ under the additional assumption that the dichotomy $(\rho,\sigma)$ on $\H$ is in its own standard representation. Note that $\Jstaralg(K_{(\rho,\sigma)})\subset J_{(\rho,\sigma)}$ by \cref{cor:K-in-J}, and that  $J_{(\rho,\sigma)}$ is universally reversible by \cref{prop:MSJA-uni-rev}. 
    Our additional assumption guarantees that the minimal sufficient *-algebra $A_{(\rho,\sigma)}$ is the universal enveloping *-algebra of $J_{(\rho,\sigma)}$ (cp.\ \cref{prop:standard-rep}).
    By the previous step, $\Jstaralg(K_{(\rho,\sigma)})$ generates the same *-algebra as $J_{(\rho,\sigma)}$.
    As $J_{(\rho,\sigma)}$ is in its universal representation, \cref{cor:rev-hull} asserts that $J_{(\rho,\sigma)}$ is the reversible J*-algebra generated by $\Jstaralg(K_{(\rho,\sigma)})$.
    However, by \cref{lem:J-of-K-is-rev}, the latter is already a reversible J*-algebra.
    Therefore, \cref{cor:rev-hull} implies  $\Jstaralg(K_{(\rho,\sigma)}) = J_{(\rho,\sigma)}$.

\emph{Step 4.} 
    We finish the proof by removing the standard representation assumption in step 3.
    Let $T:L(\hat\H)\to L(\H)$, $S:L(\H)\to L(\hat\H)$ be UP maps interconverting the dichotomy $(\rho,\sigma)$ on $\H$ with its standard representation $(\hat\rho,\hat\sigma)$ on $\hat\H$.
    By \cref{thm:ptp-inter}, $T$ restricts to a J*-isomorphism $\psi:J_{(\hat\rho,\hat\sigma)}\to J_{(\rho,\sigma)}$ whose inverse is the restriction of $S$.
    By \cref{cor:restriction-iso-on-K}, we have $\psi(K_{(\hat\rho,\hat\sigma)}) = K_{(\rho,\sigma)}$.
    By step 3, we have $\Jstaralg(K_{(\hat\rho,\hat\sigma)})=J_{(\hat\rho,\hat\sigma)}$.
    As $\psi$ is a J*-isomorphism, we have
    \begin{equation}
        \Jstaralg(K_{(\rho,\sigma)}) 
        = \Jstaralg(\psi(K_{(\hat\rho,\hat\sigma)}))
        = \psi(\Jstaralg(K_{(\hat\rho,\hat\sigma)}))
        = \psi(J_{(\hat\rho,\hat\sigma)}) 
        = J_{(\rho,\sigma)}.
    \end{equation}
\end{proof}

\section{Divergences and recovery}

\subsection{Relative entropy}
Petz showed that if $T$ is a UCP map (or, more generally, a unital 2-positive map) and $D(\rho\|\sigma) = D(T^*\rho\|T^*\sigma)$, then $R_{T,\sigma}^* T^*\rho = \rho$ and hence $(\rho,\sigma) \cptp (T^*\rho,T^*\sigma)$, where $R_{T,\sigma}$ is the Petz recovery map (see \cref{sec:recovery}) \cite{petz_sufficiency_1988}. 
Thus, the equality case in the data-processing inequality implies that $T$ is sufficient relative to $(\rho,\sigma)$ (cp.~\cref{def:sufficient-map}). 

We now generalize this statement to UP maps instead of UCP maps. 
\begin{theorem}\label{thm:petz-sufficiency}
Let $(\rho,\sigma)$ be a dichotomy on $\H$ with $\rho\ll\sigma$ and $T:L(\hat\H) \to L(\H)$ be a UP map. The following are equivalent:
\begin{enumerate}[(a)]
    \item\label{it:DPI-equality-relent} $D(\rho\|\sigma) = D(T^*\rho\|T^*\sigma)$;
    \item\label{it:DPI-equality-hockey} $E_t(\rho\|\sigma) = E_t(T^*\rho\|T^*\sigma)$ for all $t>0$;
    \item\label{it:DPI-equality-petzmap} 
    $T$ is sufficient for $(\rho,\sigma)$.
\end{enumerate}
\end{theorem}
\begin{proof}
     In the following we write $(\hat\rho,\hat\sigma) = (T^*\rho,T^*\sigma)$. Without loss of generality, we may assume that $\sigma$ and $\hat\sigma$ are  faithful (otherwise we may interconvert with faithful dichotomies). This implies that $T$ is faithful: Suppose $x\geq 0$ and $T(x)=0$, then $x=0$ follows from $0 = \tr(\sigma T(x)) = \tr(\hat\sigma x)$.
     
      \ref{it:DPI-equality-petzmap} $\Rightarrow$ \ref{it:DPI-equality-relent}: This is immediate from the data-processing inequality of the relative entropy for positive maps. 
  
    \ref{it:DPI-equality-relent} $\Rightarrow$ \ref{it:DPI-equality-hockey}: 
    Since the Hockey stick divergences fulfill the data-processing inequality, equality in the data-processing inequality for the relative entropy implies $E_t(\rho\|\sigma) = E_t(\hat\rho \| \hat\sigma)$ and $E_t(\sigma\| \rho) = E_t(\hat\sigma \| \hat\rho)$ for almost all $t\geq 1$ by \cref{eq:Frenkel}. By continuity, we get equality for all $t\geq 1$. Since $E_t(\sigma\|\rho) = t E_{1/t}(\rho\|\sigma)$, we get $E_t(\rho \| \sigma) = E_t(\hat \rho\|\hat\sigma)$ for all $t>0$. 

    \ref{it:DPI-equality-hockey} $\Rightarrow$ \ref{it:DPI-equality-petzmap}: 
	By \cref{thm:K-generates} we have $J_{(\rho,\sigma)} = \Jstaralg(K_{(\rho,\sigma)})$ and $J_{(\hat\rho,\hat\sigma)} = \Jstaralg(K_{(\hat\rho,\hat\sigma)})$. By \cref{cor:restriction-iso-on-K} we find that $T$ restricts to a J*-isomorphism $J_{(\hat\rho,\hat\sigma)} \to J_{(\rho,\sigma)}$.
	\Cref{it:DPI-equality-petzmap} now follows from item \ref{it:algebraic-recovery5} of \cref{thm:algebraic-recovery}.
\end{proof}

\begin{remark}\label{rem:curious-formula}
We note  a, maybe non-obvious, formula for $d_{\rho|\sigma}$ that may be of independent interest. 
Given a faithful state $\sigma$, we define a UP map by
\begin{align}
	\varphi_{\sigma}(a) = \int_0^1 \sigma^{s-\frac12}a\sigma^{\frac12-s}\, ds = \int_0^1 \{\sigma^{s-\frac12},a,\sigma^{\frac12-s}\}\,ds.
\end{align}
Consider the logarithmic derivative
\begin{align}\label{eq:operator-layer-cake}
	L_\sigma(\rho) :=	
    \lim_{t\to 0} \frac{\log(\sigma+t \rho) - \log(\sigma)}{t} = \int_0^\infty  \proj{\rho> t\sigma} \,dt, 
\end{align}
where the right-most equality was shown in \cite{chengErrorExponentsQuantum2025}, and in fact follows from \eqref{eq:Frenkel} \cite{cheng_operator_2025}.
Lieb has shown \cite{liebConvexTraceFunctions1973} that the logarithmic derivative is an invertible linear map, with inverse 
\begin{align}
	L^{-1}_\sigma: a \mapsto \int_0^1 \sigma^s a \sigma^{1-s} \,ds =  \sigma^{\frac12}\varphi_{\sigma}(a)\sigma^{\frac12}.
\end{align}
We thus have
\begin{align}
	d_{\rho|\sigma} &=	\sigma^{-\frac12} \rho \sigma^{-\frac12} \nonumber\\
	&=\sigma^{-\frac12}\big(L^{-1}_\sigma (L_\sigma(\rho))\big)\sigma^{-\frac12}\nonumber\\
	&= \varphi_\sigma(L_\sigma(\rho)) \nonumber\\
	&= \int_0^\infty \varphi_\sigma(\proj{\rho>t\sigma}) \,dt \nonumber\\
	&= \int_0^\infty  \int_0^1 \{\sigma^{s-\frac12},\proj{\rho>t\sigma},\sigma^{\frac12-s}\}\,ds\,dt.\label{eq:d-formula} 
\end{align}
Since $d_{\rho|\sigma} = d_{E\rho|E\sigma}$ for the trace-preserving conditional expectation $E$ onto $J_{(\rho,\sigma)}$, we may exchange $\rho$ and $\sigma$ with $E\rho$ and $E\sigma$ in the formulae above.   
\Cref{eq:d-formula} then shows how to explicitly calculate $d_{\rho,\sigma}$ using only $K_{(\rho,\sigma)}$ and $E\sigma$.
\end{remark}

\begin{remark}\label{remark:layer-cake}
	The proof of \cref{thm:petz-sufficiency} uses \cref{thm:K-generates}, which in turn uses Petz's sufficiency theorem for completely positive maps to establish that $\Jstaralg(K_{(\rho,\sigma)})$ is a sufficient J*-algebra.
    Our result, therefore, does not provide an independent proof of Petz's theorem.
    However, if one could establish that $\Jstaralg(K_{(\rho,\sigma)})$ is sufficient by alternative means, our argument would provide an independent proof of Petz's theorem. 

	One way one might hope to do this is by deriving an integral representation of the form
	\begin{align}
		d_{\rho|\sigma}=\sum_{k} \int_0^\infty \cdots \int_0^\infty f_k(t_1,\ldots,t_k) \{p_{t_1},\ldots,p_{t_k}\}\, dt_1\cdots dt_k,
	\end{align}
	for some $(\rho,\sigma)$-independent functions $f_k:[0,\infty)\to \RR$.
	Here, we we write $p_t = \proj{\rho > t\sigma}$  (recall that $\{a_1,\ldots,a_n\}$ denotes a symmetrized product and not a set). 
	Indeed, since equality of the Hockey stick divergences implies that $T$ restricts to a J*-isomorphism such that $Tp_t = \proj{T^*\rho > t T^*\sigma} =: \hat p_t$, we would get
	\begin{align}
		Td_{\hat\rho|\hat\sigma} &= \sum_{k} \int_0^\infty \cdots \int_0^\infty f_k(t_1,\ldots,t_k) T\{p_{t_1},\ldots,p_{t_k}\}\, dt_1\cdots dt_k\nonumber\\ 
		&=\sum_{k} \int_0^\infty \cdots \int_0^\infty f_k(t_1,\ldots,t_k) \{Tp_{t_1},\ldots,Tp_{t_k}\}\, dt_1\cdots dt_k \nonumber\\
		&= d_{\hat\rho|\hat\sigma}.
	\end{align}
	and the result would follow from \cref{thm:algebraic-recovery}. 
    We note that the above is not possible with $f_k =0$ for all $k>1$. 
    Indeed, this follows from the following argument of Konrad Szyma\'nski: 
    By considering abelian dichotomies one sees that $f_k=0$ for $k>1$ forces $f_1(t)=1$ for all $t>0$. By the operator layer cake formula \eqref{eq:operator-layer-cake}, this implies the false statement that $d_{\rho|\sigma}$ equals the logarithmic derivative.

	Note that our results imply that any operator $x \in J_{(\rho,\sigma)} = \Jstaralg(K_{(\rho,\sigma)})$ has an integral-representation as above with $x$-dependent functions $f_k(x)$. These can be viewed as higher-order \emph{layer-cake representations}, see \cite{liu_layer_2025,chengErrorExponentsQuantum2025,cheng_operator_2025}.   
    In particular, if $E:L(\H) \to J_{(\rho,\sigma)}$ is the trace-preserving conditional expectation, the PTP-equivalent states $E\rho$ and $E\sigma$ have an integral representation as above with state-dependent functions $f_k(\rho), f_k(\sigma)$. 
\end{remark}

\subsection{Sandwiched R\'enyi divergence}
We now prove a Petz-sufficiency statement for the sandwiched R\'enyi divergence, defined as \cite{muller-lennertQuantumRenyiEntropies2013,wildeStrongConverseClassical2014}
\begin{align}
\tilde D_\alpha (\rho\|\sigma)  = \frac{1}{\alpha-1}\log \tr\left[\big(\sigma^\frac{1-\alpha}{2\alpha}\rho\sigma^\frac{1-\alpha}{2\alpha}\big)^\alpha\right],\quad \alpha \in (\tfrac12,\infty)\setminus\{1\}
\end{align}
if $\rho\ll \sigma$ or $\alpha \in (\frac12,1)$ and set to $+\infty$ otherwise. The definition can be extended to $\alpha=1,\infty$ by taking limits. Then $\tilde D_1(\rho\|\sigma) = D(\rho\|\sigma)$ is the relative entropy. 
In the following we set
\begin{align}
    \tilde Q_\alpha(\rho\|\sigma) := \tr\left[\big(\sigma^\frac{1-\alpha}{2\alpha}\rho\sigma^\frac{1-\alpha}{2\alpha}\big)^\alpha\right].
\end{align}
Note that 
\begin{align}
	\norm{d_{\rho|\sigma}}_{\alpha,\sigma}^\alpha 
	= \norm{ \Gamma^{\frac{1}{\alpha}} d_{\rho,\sigma} }_\alpha^\alpha 
	= \tr\big(\sigma^{\frac{1-\alpha}{2\alpha}}\rho \sigma^{\frac{1-\alpha}{2\alpha}}\big)^\alpha = \tilde Q_\alpha(\rho\|\sigma), 
\end{align}
showing that the sandwiched R\'eny divergence is essentially equivalent to the $L^p$-norm of $d_{\rho|\sigma}$ relative to $\sigma$.
The sandwiched R\'enyi divergence fulfills the data-pocessing inequality $\tilde D_\alpha(T^*\rho \| T^*\sigma)\leq \tilde D_\alpha(\rho\|\sigma)$ for any $\alpha \in(\frac12,1)\cup(1,\infty)$ for UCP maps \cite{beigiSandwichedRenyiDivergence2013,frankMonotonicityRelativeRenyi2013} as well as for UP maps \cite{muller-hermesMonotonicityQuantumRelative2017,jencova_renyi_2018,jencova_renyi_2021}.

\begin{theorem}\label{thm:sandwiched-sufficiency} Let $(\rho,\sigma)$ be a dichotomy on $\H$ with $\rho \ll \sigma$. Let $T:L(\hat\H) \to L(\H)$ be a UP map.
The following are equivalent:
\begin{enumerate}[(a)]
    \item $\tilde D_\alpha(T^*\rho,T^*\sigma) = \tilde D_\alpha(\rho,\sigma)$ for some $\alpha \in (\frac12,1)\cup (1,\infty)$.
    \item $T$ is sufficient for $(\rho,\sigma)$.
\end{enumerate}
\end{theorem}

To prove \cref{thm:sandwiched-sufficiency}, we will make use of the connection between sandwiched R\'enyi divergences and $L^p$-norms and follow the proof of the corresponding statement for quantum channels in \cite{jencova_renyi_2018,jencova_renyi_2021}. We warn the reader that the convention for $L^p$-spaces in \cite{jencova_renyi_2018,jencova_renyi_2021} is different from ours. 

For the remainder of this subsection, we consider the setting of  \cref{thm:sandwiched-sufficiency}.
We write $(\hat \rho,\hat\sigma) = (T^*\rho,T^*\sigma)$. 
By introducing pre- and postprocessing UCP maps, we can and will assume without loss of generality that $\sigma$ and $\hat \sigma$ are faithful.

\begin{lemma}\label{lem:L1-sufficient}
	Let $J\subset L(\H)$ be a universally reversible J*-algebra admitting a $\sigma$-preserving conditional expectation $F$. Then 
    $J$ is sufficient for $(\rho,\sigma)$ (i.e., $\rho$ is $F$-invariant) if and only if $\rho \in L^1(J,\sigma)$. 
    Moreover, if $\rho$ is faithful, these equivalent condition imply $L^1(J,\sigma) = L^1(J,\rho)$.
\end{lemma}
\begin{proof}
    The first part of the statement is shown in \cref{cor:d-in-J}. The only statement left to prove is $L^1(J,\sigma) = L^1(J,\rho)$.
    This follows from the explicit characterization of $F$-invariant states in \cref{prop:CE-univ-rev}. 
\end{proof}

For the following Lemma, recall from \cref{cor:fixed-point-recovery} that a map $T$ is sufficient for a pair of states $(\rho,\sigma)$ with $\sigma$ and $\hat\sigma$ faithful, if and only if $d_{\rho|\sigma}$ is in the fixed-point J*-algebra $\mathrm{Fix}(W)$ of $W=T\circ R_{T,\sigma}$. 

\begin{lemma}\label{lem:L1-sufficient-2}
    Suppose $T$ is sufficient for $(\rho,\sigma)$ and $\mu \in L^1(J_{(\rho,\sigma)},\sigma)$ is a state. Then $T$ is sufficient for $(\mu,\sigma)$. 
\end{lemma}

\begin{proof}
    $J_{(\rho,\sigma)}$ is universally reversible and admits a $\sigma$-preserving conditional expectation. 
    Thus, the statement follows from \cref{lem:L1-sufficient}.
\end{proof}

The case $\alpha=2$ of \cref{thm:sandwiched-sufficiency} has been shown before by Jen\v{c}ov\'{a}\cite[Lem.~8]{jencovaPreservationQuantumRenyi2017}. 
Since it will be essential for proving \cref{thm:sandwiched-sufficiency}, we provide an independent proof using \cref{thm:algebraic-recovery}.

\begin{lemma}\label{lem:sandwich-2}
The following are equivalent:
\begin{enumerate}[(a)]
    \item $\tilde D_2(\hat\rho,\hat\sigma) = \tilde D_2(\rho,\sigma)$.
    \item $T$ is sufficient for $(\rho,\sigma)$.
\end{enumerate}
\end{lemma}
\begin{proof}
Since $\tilde Q_2(\rho\|\sigma) = \norm{d_{\rho|\sigma}}^2_{2,\sigma}$, 
equality of the sandwiched R\'enyi divergence of order $\alpha=2$ implies
\begin{align}
    \norm{d_{\rho|\sigma}}_{2,\sigma} = \norm{d_{\hat\rho|\hat\sigma}}_{2,\hat\sigma} = \norm{R_{T,\sigma} d_{\rho|\sigma}}_{2,\hat\sigma} \leq \norm{d_{\rho|\sigma}}_{2,R_{T,\sigma}^*\hat\sigma} = \norm{d_{\rho|\sigma}}_{2,\sigma},
\end{align}
where we used \cref{cor:UP-contraction} for the inequality and $R_{T,\sigma}^* T^*\sigma = \sigma$. 
Using that $T$ is the recovery map of $R_{T,\sigma}$, \cref{lem:observable-recovery} then shows
\begin{align}
    T d_{\hat\rho|\hat\sigma} = T R_{T,\sigma} d_{\rho|\sigma} = d_{\rho|\sigma},
\end{align}
which implies $R_{T,\sigma}^*\hat \rho = \hat \rho$ by \cref{thm:algebraic-recovery}.
\end{proof}

We will make use of a second ingredient that was shown by Jen\v{c}ov\'{a}using complex interpolation techniques for non-commutative $L^p$-spaces, see also \cite{beigiSandwichedRenyiDivergence2013,muller-hermesMonotonicityQuantumRelative2017,hiaiazRenyiDivergences2024}.

\begin{lemma}[{\cite[Proof of Thm.~7]{jencovaPreservationQuantumRenyi2017}}]\label{lem:jencova1}
Equality in the DPI,  $\tilde D_\alpha(\rho\|\sigma) = \tilde D_\alpha(\hat\rho\|\hat\sigma)$, for $\alpha>1$ implies 
    \begin{align}
        \tilde D_2(\omega \| \sigma) = \tilde D_2(T^* \omega \| \hat\sigma),
    \end{align}
    where $\omega = \omega_0/\tr(\omega_0)$ with
    \begin{align}\label{eq:def-omega-SW}
            \omega_0
            = \sigma^{\frac14} \Big(\sigma^{\frac{1-\alpha}{2\alpha}}\rho\sigma^{\frac{1-\alpha}{2\alpha}}\Big)^{\frac\alpha2} \sigma^{\frac14} = \Gamma_\sigma^{\frac12} \Big( \big(\Gamma_\sigma^{\frac{1}{\alpha}}d_{\rho|\sigma}\big)^{\frac{\alpha}{2}} \Big).
    \end{align}
\end{lemma}

\begin{proof}[Proof of \cref{thm:sandwiched-sufficiency} for $\alpha>1$]
    By \cref{lem:sandwich-2,lem:jencova1}, $T$ is sufficient for $(\omega,\sigma)$. 
    A direct calculation shows
    \begin{align}
        \rho = \tr(\omega_0) \,\Gamma_\sigma^{1-\frac1\alpha}\big((\Gamma_\sigma^{\frac12} d_{\omega|\sigma} )^{\frac2\alpha}\big). 
    \end{align}
    \cref{lem:Lp-properties} and $d_{\omega|\sigma}\in J_{(\omega,\sigma)}$ show $\rho \in L^1(J_{(\omega,\sigma)},\sigma)$.
    By \cref{lem:L1-sufficient-2} this implies that $T$ is sufficient for $(\rho,\sigma)$. 
\end{proof}

For the case $\frac12<\alpha<1$ we follow proof of Jen\v{c}ov\'{a} for 2-positive maps \cite{jencova_renyi_2021}.
We define 
\begin{equation}
    \mu =\frac{1}{\tilde Q_\alpha(\rho\|\sigma)} \big(\Gamma_\sigma^{1/\alpha} d_{\rho|\sigma}\big)^\alpha,\qquad \hat\mu =\frac{1}{\tilde Q_\alpha(\hat\rho\|\hat\sigma)} \big(\Gamma_{\hat\sigma}^{1/\alpha} d_{\hat\rho|\hat\sigma}\big)^\alpha.
\end{equation}
It follows from \cref{lem:Lp-properties} that $\mu \in L^1(J_{(\rho,\sigma)},\sigma)$.
Hence, $d_{\mu|\sigma} \in J_{(\rho,\sigma)}$.

\begin{lemma}\label{lem:rho-sigma-suff-mu-sigma} $T$ is sufficient for $(\rho,\sigma)$ if and only if $T$ is sufficient for $(\mu,\sigma)$.
\end{lemma}
\begin{proof}
    If $T$ is sufficient for $(\rho,\sigma)$, it follows from \cref{lem:L1-sufficient-2} and the above discussion that it is sufficient for $(\mu,\sigma)$. 
    Conversely, assume that $T$ is sufficient for $(\mu,\sigma)$. By the explicit form of $\mu$, we have
    \begin{align}
        \rho = \tilde Q_\alpha(\rho\|\sigma)^{1/\alpha} \Gamma_\sigma^{1-\frac{1}{\alpha}}\mu^{\frac1\alpha}. 
    \end{align}
     \cref{lem:Lp-properties} now shows that $\rho \in L^1(J_{(\mu,\sigma)},\sigma)$. It follows that $d_{\rho|\sigma} \in J_{(\mu,\sigma)}$. Since $T$ is assumed to be sufficient for $(\mu,\sigma)$, we have $J_{(\mu,\sigma)} \subset \mathrm{Fix}(W)$. Hence $T$ is sufficient for $(\rho,\sigma)$ by \cref{cor:fixed-point-recovery}.
\end{proof}

In the following we set $\gamma = \alpha/(\alpha-1)>1$ and choose $\beta$ such that $\frac{1}{\gamma} + \frac{1}{\beta} = 1$.
Recall from \cref{sec:Lp-spaces} that $T_\gamma$ and $R_{T,\sigma}$ are defined as
\begin{align}\label{eq:petz-Lp-formula}
    T_\gamma =  \Gamma_\sigma^{\frac1\gamma}\circ T \circ \Gamma_{\hat\sigma}^{-\frac1\gamma},\qquad
    R_{T,\sigma}^* = \Gamma_\sigma^{\frac{1}{\beta}} \circ T_\gamma \circ \Gamma_{\hat\sigma}^{-\frac{1}{\beta}}.
\end{align}

\begin{lemma}\label{lem:crucial-sandwiched}
    $T$ is sufficient with respect to $(\rho,\sigma)$ if 
    \begin{align}
       \tilde Q_\alpha(\rho\|\sigma) = \tilde Q_\alpha(\hat\rho\|\hat\sigma),\qandq \mu^{\frac1\gamma} = T_\gamma({\hat\mu}^\frac{1}{\gamma}). 
    \end{align}
\end{lemma}
\begin{proof} 
    Set $\omega_0  = \Gamma_\sigma^{\frac{1}{\beta}} \mu^{\frac1\gamma}$ and $\hat\omega_0 = \Gamma_{\hat\sigma}^{\frac{1}{\beta}} {\hat\mu}^\frac{1}{\gamma}$. 
    Then \cref{eq:petz-Lp-formula} and the assumption $\mu^{\frac1\gamma} = T_\gamma({\hat\mu}^\frac{1}{\gamma})$ show
    \begin{align}
        R_{T,\sigma}^*(\hat \omega_0) = \Gamma_\sigma^{\frac{1}{\beta}}\circ T_\gamma({\hat\mu}^{\frac1\gamma}) = \omega_0.
    \end{align}
    Since $R^*_{T,\sigma}$ is trace-preserving we have $\tr(\hat\omega_0) = \tr(\omega_0)$. We also have $R_{T,\sigma}^*(\hat{\sigma}) = \sigma$. Set $\hat\omega =\hat\omega/\tr\hat\omega_0$ and $\omega = \omega/\tr\omega_0$. Then $R^*_{T,\sigma} \hat\omega = \omega$.  Moreover
    \begin{align}
    \tilde Q_\gamma(\hat\omega\|\hat\sigma) = \tr\big((\Gamma_{\hat\sigma}^{\frac1\gamma} d_{\hat\omega|\hat\sigma})^\gamma\big) = \frac{\tr \mu}{(\tr\hat\omega_0)^\gamma} = \frac{1}{(\tr \hat\omega_0)^\gamma}
    \end{align}
    and similarly
    \begin{align}
    \tilde Q_\gamma(\omega\|\sigma) = \frac{1}{(\tr\omega_0)^\gamma} =\frac{1}{(\tr \hat\omega_0)^\gamma} =  \tilde Q_\gamma(\hat\omega\|\hat\sigma).
    \end{align}
    It hence follows from the $\alpha>1$ case of \cref{thm:sandwiched-sufficiency} that $T^* \omega = \hat\omega$ and $T^* \sigma = \hat\sigma$. Hence $T$ is sufficient for $(\omega,\sigma)$. 
    As in the proof of \cref{lem:rho-sigma-suff-mu-sigma}, it follows that $T$ is sufficient for $(\mu,\sigma)$ and hence sufficient for $(\rho,\sigma)$. 
\end{proof}

\begin{proof}[Proof of \cref{thm:sandwiched-sufficiency} for $\alpha \in (\frac12,1)$] 
Recall that we assume without loss of generality that $\sigma$ is faithful. First, consider the case where $\rho$ is also faithful. 
In \cite[Proof of Thm.~5.1 for $\psi\sim\varphi$]{jencova_renyi_2021}, Jen\v{c}ov\'{a}shows for general positive maps that if $\rho$ and $\sigma$ have full rank, then equality in the data-processing inequality implies $T_\gamma({\hat\mu}^{\frac1\gamma}) = \mu^{\frac1\gamma}$. 
Hence, the assumptions of \cref{lem:crucial-sandwiched} are fulfilled, and we conclude that $T$ is sufficient for $(\rho,\sigma)$. 

In the case that $\rho$ does not have full rank, we can use the same continuity arguments as in \cite[Proof of Thm.~5.1]{jencova_renyi_2021} to reduce the claim to the full rank case. 
\end{proof}

\subsection{\texorpdfstring{$\alpha$-$z$}{alpha-z}-R\'enyi divergence}
In this section, we discuss sufficiency statements for the $\alpha$-$z$-R\'enyi divergence \cite{audenaertAzRenyiRelativeEntropies2015}, building upon the work by Hiai and Jen\v{c}ov\'{a} \cite{hiaiazRenyiDivergences2024}. 
The $\alpha$-$z$-R\'enyi divergence is the two-parameter divergence defined as 
\begin{align}
    D_{\alpha,z}(\rho\|\sigma) = \frac{1}{\alpha-1}\log Q_{\alpha,z}(\rho\|\sigma),\quad Q_{\alpha,z}(\rho\|\sigma) = \tr\big(\sigma^{\frac{1-\alpha}{2z}}\rho^{\frac{\alpha}{z}} \sigma^{\frac{1-\alpha}{2z}} \big)^z,\quad 0<\alpha\neq 1, z>0. 
\end{align}
It reduces to the sandwiched R\'enyi divergence for $\alpha=z$ and the Petz R\'enyi divergence for $z=1$. 
The $\alpha$-$z$ R\'enyi divergence fulfills the data-processing inequality for PTP maps if and only if it fulfills the data-processing inequality for CPTP maps if and only if \cite{audenaertAzRenyiRelativeEntropies2015,carlenInequalitiesQuantumDivergences2018,zhangWignerYanaseDysonConjectureCarlenFrankLieb2020,hiaiQuantumFDivergencesNeumann2021,katoAzRenyiDivergenceNeumann2024,hiaiazRenyiDivergences2024}
\begin{align}\label{eq:alpha-z-range}
    0<\alpha<1,\quad \max\{\alpha,1-\alpha\}\leq z,\quad\text{or}\quad \alpha>1,\quad \max\{\alpha/2,\alpha-1\}\leq z \leq \alpha. 
\end{align}
The following is the exact analog of \cite[Thm.~4.5]{hiaiazRenyiDivergences2024} for positive instead of 2-positive maps.
\begin{theorem}\label{thm:alpha-z-smaller-1}
Let $0<\alpha<1$ and $\max\{\alpha,1-\alpha\} \leq z$. Let $\rho,\sigma$ be states on $\H$ and $T:L(\hat \H)\to L(\H)$ a unital, positive map. 
Assume that either $\alpha <z$ and $\sigma \ll \rho$ or $1-\alpha<z$ and $\rho \ll \sigma$. 
Then $T$ is sufficient with respect to $(\rho,\sigma)$ if and only if $D_{\alpha,z}(\rho\|\sigma) = D_{\alpha,z}(T^*\rho\|T^*\sigma)$.
\end{theorem}
In the following, whenever $\alpha$ and $z$ are fixed,  we set
    \begin{align}\label{eq:pqr}
    p = \frac{z}{\alpha},\qquad r = \frac{z}{1-\alpha},\qquad q = -r = \frac{z}{\alpha - 1}.
\end{align}

\begin{proof}[Proof of \cref{thm:alpha-z-smaller-1}] 
    We only need to show that equality in the data-processing inequality implies sufficiency of $T$ for $(\rho,\sigma)$. Following \cite{hiaiQuantumFDivergencesNeumann2021}  we  assume $\alpha<z$ and $\sigma \ll \rho$, since otherwise we can exchange the roles of $p,r$ and $\rho,\sigma$ using the equality $Q_{\alpha,z}(\rho\|\sigma) = Q_{1-\alpha,z}(\sigma\|\rho)$. Then $p>1$ and we have
    \begin{align}
        \frac1p + \frac1r = \frac1z.
    \end{align}
    By pre- and post-processing, we can assume without loss of generality that both $\rho$ and $\hat \rho=T^*\rho$ are faithful. 
    Set
    \begin{align}
        \mu_0 = (\rho^\frac{1}{2p}\sigma^{\frac1r}\rho^\frac{1}{2p})^z,\qquad \omega_0 = \rho^\frac{p-1}{2p}\mu_0^{\frac1p} \rho^{\frac{p-1}{2p}}. 
    \end{align}
    Similarly, we define $\hat\mu_0$ and $\hat\omega_0$ relative to $\hat\rho = T^*\rho$ and $\hat\sigma = T^*\sigma$. 
    Then $\tr\mu_0 = Q_{\alpha,z}(\rho\|\sigma) = \tilde Q_{\alpha,z}(\hat\rho\|\hat\sigma) =\tr \hat\mu_0$. Define $\omega = \omega_0/\tr\omega_0$ and similarly for $\hat\omega$.
    In \cite[Proof of Thm.~4.5]{hiaiazRenyiDivergences2024} it is shown that $T^*\omega_0 = \hat\omega_0$. 
    In particular $\tr\omega_0 = \tr\hat\omega_0$, hence also $T^*\omega = \hat\omega$.\footnote{This part of the proof does not require $T$ to be 2-positive.}
    By construction, we have
    \begin{align}
     \tilde Q_p(\omega\|\rho) =   \frac{\tr \mu_0 }{(\tr\omega_0)^p}. 
    \end{align}
    Using the equality of $\alpha$-$z$-R\'enyi divergences, we thus find
    \begin{align}
        \tilde Q_p(\omega\|\rho)  = \tilde Q_p(\hat\omega\|\hat\rho). 
    \end{align}
    Hence \cref{thm:sandwiched-sufficiency} shows that $T$ is sufficient for $(\omega,\rho)$. 
    By \cref{lem:Lp-properties}, we have
    \begin{align}
        \mu \propto \big(\Gamma_\rho^{\frac1p} d_{\omega|\rho}\big)^p \in L^1(J_{(\omega,\rho)},\rho)\quad\Rightarrow\quad \sigma \propto \big( \Gamma_\rho^{-\frac1p} \mu^{\frac1z}\big)^r \in L^1(J_{(\omega,\rho)},\rho).  
    \end{align}
   By \cref{lem:L1-sufficient-2} it follows that $T$ is sufficient for $(\sigma,\rho)$. 
\end{proof}

The following is the exact analog of \cite[Thm.~4.7]{hiaiazRenyiDivergences2024} for positive instead of 2-positive maps. 

\begin{theorem}\label{thm:alpha-z-bigger-1}
    Let $\alpha>1$ and $\max\{\alpha/2,\alpha-1\} \leq z\leq \alpha < z+1$. Suppose $\rho \ll\sigma$ are states on $\H$ and $T:L(\hat \H)\to L(\H)$ is a UP map. Then $T$ is sufficient for $(\rho,\sigma)$ if and only if $D_{\alpha,z}(T^*\rho\|T^*\sigma) = D_{\alpha,z}(\rho\|\sigma)$. 
\end{theorem}
\begin{proof}
    We only need to show that equality in the data-processing inequality implies sufficiency of $T$ for $(\rho,\sigma)$.
    By pre- and post-processing, we can assume that $\sigma$ and $\hat\sigma = T^*\sigma$ are faithful. 
    We use $p,q,r$ as defined in \cref{eq:pqr}, so that $p\in (\frac12,1]$ and $q >  1$. Define
    \begin{align}
    w = \left(\sigma^{\frac{1-\alpha}{2z}} \rho^{\frac{\alpha}{z}} \sigma^{\frac{1-\alpha}{2z}} \right)^{\alpha-1} 
    = \big(\Gamma_\sigma^{-\frac1q} \rho^{\frac1p}\big)^{\alpha-1},
    \quad \omega_0 = \Gamma^{1-\frac{1}{q}}_\sigma w,
    \quad \omega = \frac{\omega_0}{\tr\omega_0}.
\end{align}
Similarly, define $\hat w$, $\hat\omega_0$ and $\hat\omega$ relative to $\hat\sigma$ and $\hat\rho = T^*\rho$. 
Note that 
\begin{align}
    \tr w^q = Q_{\alpha,z}(\rho\|\sigma), 
\end{align}
so that equality in the DPI corresponds to $\tr w^q = \tr \hat w^q$.
In \cite[Lem.~4.6]{hiaiazRenyiDivergences2024} it is shown that this equality also implies $T_{q}\hat w = w$. It follows that
\begin{align}
    R_{T,\sigma}^* \hat\omega_0 =\Gamma^{1-\frac{1}{q}}_\sigma T_q  \Gamma^{-(1-\frac{1}{q})}_{\hat\sigma} \hat\omega_0 = \omega_0.
\end{align}
In particular $\tr \hat\omega_0 = \tr\omega_0$. 
Moreover, we have
\begin{align}
    \tilde Q_q(\omega\| \sigma) = \frac{\tr w^q}{\tr(\omega_0)^q} = \frac{\tr\hat w^q}{\tr(\hat\omega_0)^q} = \tilde Q_q(\hat\omega\|\hat\sigma).  
\end{align}
Thus, by \cref{thm:sandwiched-sufficiency}, $R_{T,\sigma}$ is sufficient for $(\hat\omega,\hat\sigma)$ and hence $T$ is sufficient for $(\omega,\sigma)$.
A short calculation shows
\begin{align}
    \rho^{\frac1p} \propto \Gamma_\sigma^{\frac1q} w^{\frac{1}{\alpha-1}} \propto \Gamma_\sigma^{\frac1q} \big((\Gamma_\sigma^{\frac1q }d_{\omega|\sigma})^{\frac{1}{\alpha-1}}\big).
\end{align}
Since $q(\alpha-1) = z$ and $\frac{1}{z} + \frac{1}{q} = \frac{1}{p}$, it follows from \cref{lem:Lp-properties} that $\rho \in L^1(J_{(\omega,\sigma)},\sigma)$. Hence, by \cref{lem:L1-sufficient-2}, it follows that $T$ is sufficient for $(\rho,\sigma)$.
\end{proof}
As emphasized in \cite{hiaiazRenyiDivergences2024}, the additional constraint $\alpha<z+1$ is necessary, even in the case of completely positive maps.

\appendix

\section{More about minimal sufficient J*-algebras}\label{sec:more}

\subsection{Jordan symmetries vs unitary symmetries}\label{sec:symmetries}

An (anti-)unitary symmetry of a dichotomy $(\rho,\sigma)$ on a Hilbert space $\H$ is an (anti-)unitary operator $u$ on $\H$ such that
\begin{equation}\label{eq:anti-unitary-symmetry}
    u\rho u^*=\rho, \qquad u\sigma u^*=\sigma.
\end{equation}
A unitary symmetry is trivial if $u \propto 1$ is a scalar.
The absence of non-trivial unitary symmetries is easily characterized:

\begin{lemma}\label{lem:unitary-symmetries}
    Let $(\rho,\sigma)$ be a faithful dichotomy.
    The following are equivalent:
    \begin{enumerate}[(a)]
        \item\label{it:unitary-symmetries1} $(\rho,\sigma)$ only have trivial unitary symmetries; 
        \item\label{it:unitary-symmetries2} $L(\H)$ is the minimal sufficient *-algebra;
        \item\label{it:unitary-symmetries3} the Koashi-Imoto decomposition $\H = \oplus_j \K_j\ox \L_j$ consists of a single block $\K$ with trivial multiplicity, i.e., $\dim\L=1$.
        \item\label{it:unitary-symmetries4} the pair $(\rho,\sigma)$ is irreducible, i.e., all jointly invariant subspaces are trivial.
    \end{enumerate}
\end{lemma}
\begin{proof}
    \ref{it:unitary-symmetries1} $\Leftrightarrow$ \ref{it:unitary-symmetries4} follows from the standard fact that a pair $(a,b)$ of hermitian operators on a finite-dimensional Hilbert space is irreducible if and only if the only unitaries commuting with both $a$ and $b$ are the trivial ones. 
    \ref{it:unitary-symmetries2} $\Leftrightarrow$ \ref{it:unitary-symmetries3} is explained in \cref{rem:KI}.
    It follows from the Koashi-Imoto decomposition (see \cref{rem:KI}) that the group of unitary symmetries of $(\rho,\sigma)$ is the unitary group of the commutant of the minimal sufficient *-algebra $A_{(\rho,\sigma)}$.
    This unitary group is trivial if and only if $A_{(\rho,\sigma)}=L(\H)$. This shows \ref{it:unitary-symmetries1} $\Leftrightarrow$ \ref{it:unitary-symmetries2}, which finishes the proof.
\end{proof}

\begin{proposition}\label{prop:J-symmetries}
    Let $(\rho,\sigma)$ be a faithful dichotomy.
    The following are equivalent:
    \begin{enumerate}[(a)]
        \item\label{it:J-symmetries1} $(\rho,\sigma)$ has no anti-unitary symmetries and only trivial unitary symmetries;
        \item\label{it:J-symmetries2} $L(\H) = J_{(\rho,\sigma)}$ is the minimal sufficient J*-algebra;
        \item\label{it:J-symmetries3} The following three conditions are met:
        \begin{enumerate}[({c.}1)]
            \item\label{it:J-symmetries31} $(\rho,\sigma)$ only has trivial unitary symmetries,
            \item\label{it:J-symmetries32} there is no basis relative to which both $\rho$ and $\sigma$ are real matrices,
            \item\label{it:J-symmetries33} if $\dim(\H) = 2n$ is even, it is not possible to decompose $\H$ as $\H = \CC^n\oplus \CC^n$ in such a way that $\beta(\rho)=\rho$ and $\beta(\sigma)=\sigma$, where $\beta$ is defined as in \cref{exa:fixed-pts}.
        \end{enumerate}
    \end{enumerate}
\end{proposition}

\begin{proof}
    Since all J*-automorphisms of $L(\H)$ are either implemented by unitaries or anti-unitaries (this follows from Wigner's theorem \cite{wigner1931gruppentheorie}), item \ref{it:J-symmetries1} is equivalent to:
    \begin{enumerate}[(a$'$)]\it 
        \item\label{it:J-symmetries1p} 
        The identity is the only J*-automorphism of $L(\H)$ that leaves both $\rho$ and $\sigma$ invariant.
    \end{enumerate}

    \ref{it:J-symmetries1} $\Rightarrow$ \ref{it:J-symmetries3} is clear.

    \ref{it:J-symmetries2} $\Rightarrow$ \ref{it:J-symmetries1p}:
    Let $\vartheta$ be a J*-automorphism of $L(\H)$ with $\vartheta(\rho)=\rho$, $\vartheta(\sigma)=\sigma$.
    Then the fixed-point J*-algebra $\mathrm{Fix}(\vartheta)$ is sufficient for $(\rho,\sigma)$ (indeed, the conditional expectation $E = \frac12(\id + \vartheta)$ onto $J$ is $(\rho,\sigma)$-preserving).
    Hence, $\mathrm{Fix}(\vartheta) = L(\H)$ because any sufficient J*-algebra contains the minimal sufficient J*-algebra.
    This shows the claim since $\mathrm{Fix}(\vartheta)=L(\H)$ implies $\vartheta = \id$.

    \ref{it:J-symmetries3} $\Rightarrow$ \ref{it:J-symmetries2}:
    By \cref{lem:unitary-symmetries}, item \ref{it:J-symmetries31} implies that the minimal sufficient *-algebra is $L(\H)$.
    According to \cref{thm:luczak}, this entails that $J_{(\rho,\sigma)}$ generates $L(\H)$ as a *-algebra.
    Thus, $J_{(\rho,\sigma)}\subset L(\H)$ is a J*-factor (if it had a center, it would generate a *-algebra with a center) with trivial multiplicity (i.e., $J_{(\rho,\sigma)}$ is not of the form $\hat J \ox \1$ for some decomposition $\H = \K\ox\L$).
    By the classification of J*-factors (see \cref{sec:structure-theory}), there are three possibilities: $J_{(\rho,\sigma)}= L(\H)$, (ii) there is a basis such that $J = L(\CC^d)^t$ with $d=\dim \H$, or (iii) $\H$ is even-dimensional and there is a basis such that $J=L(\CC^{2d})^\beta$.
    The cases (ii) and (iii) are ruled out by items \ref{it:J-symmetries32} and \ref{it:J-symmetries33}, respectively.
\end{proof}

\begin{corollary}
    Let $(\rho,\sigma)$ be an irreducible faithful dichotomy on an odd-dimensional Hilbert space $\H$.
    If there is no basis in which both $\rho$ and $\sigma$ are real matrices, then $L(\H) = J_{(\rho,\sigma)}$ is the minimal sufficient J*-algebra.
\end{corollary}

\subsection{Sufficiently many examples}\label{sec:many-examples}

We say that a J*-algebra $J$ is \emph{2-generated} if there exist hermitian elements $a,b\in J$ such that $J = \Jstaralg(a,b)$.
2-generated J*-algebras are classified in \cref{sec:2-gen}: A J*-algebra is 2-generated if and only if it is J*-isomorphic to a direct sum of the J*-factors $L(\CC^d)$ ($d\ge 1$), $L(\CC^d)^t$  ($d\ge2$), and $L(\CC^{2d})^\beta$ ($d\ge 4$).
We conjecture that 2-generatedness is equivalent to being minimal sufficient for some dichotomy.
We are only able to prove this for J*-factors.

\begin{proposition}\label{prop:2-gen}
    Let $J$ be a J*-factor on a Hilbert space $\H$.
    The following are equivalent:
    \begin{enumerate}[(a)]
        \item $J=J_{(\rho,\sigma)}$ is the minimal sufficient J*-algebra of a faithful dichotomy $(\rho,\sigma)$ on $\H$;
        \item $J$ is 2-generated.
    \end{enumerate}
\end{proposition}

\begin{lemma}\label{lem:minimality-condition}
    Let $J\subset L(\H)$ be a J*-factor and let $(\rho,\sigma)$ be a dichotomy on $\H$ such that $J=\Jstaralg(\rho,\sigma)$.
    If $A=\staralg(J)$ is a factor, then $J$ is the minimal sufficient J*-algebra $J = J_{(\rho,\sigma)}$.
\end{lemma}
\begin{proof}
    Since $J$ is 2-generated, it is universally reversible.
    We have $\H = \L\ox\R$, $A = L(\L) \ox \1$, $J = \hat J\ox \1$ by the general representation theory of (universally) reversible factors. It follows that $\rho = \hat\rho \ox \1, \sigma = \hat\sigma\ox \1$ and $\staralg(\hat\rho,\hat\sigma) = L(\L)$. 
    Passing to $L(\L)$, we can (and will) assume without loss of generality that $A = L(\H) = \staralg(\rho,\sigma)$.

    Set $A_0 = \staralg(J_{(\rho,\sigma)})$ and let $F:L(\H) \to J_{(\rho,\sigma)}$ be the $(\rho,\sigma)$-preserving conditional expectation. 
    From \cref{sec:CEs2}, we find that $\H = \oplus_j \L_j\ox \R_j$, 
    $A_0 = \oplus_j L(\L_j) \ox \1$, $\rho = \oplus_j \rho_j \ox \omega_j , \sigma = \oplus_j \sigma_j\ox \omega_j$. But since $\rho$ and $\sigma$ generate $L(\H)$, there can only be one summand, and the associated state $\omega$ must be trivial. 
    It follows that $\rho,\sigma \in J_{(\rho,\sigma)}$ and hence $J_{(\rho,\sigma)} = \Jstaralg(\rho,\sigma) = J$.
\end{proof}

\begin{proof}
    By the representation theory of J*-factors (see \cref{sec:reps}), there are two possibilities (1) $A=\staralg(J)$ is a factor, or (2) $A= \staralg(J)$ is a direct sum of two factors.

    \emph{Case (1).}
    By \cref{prop:reps}, the generated *-algebra $\staralg(J)$ is a factor. 
    Let $a,b\in J$ be hermitian operators such that $\Jstaralg(a,b)=J$. Pick constants $c_i>0$ such that $\rho = c_1 + c_2 a$ and $\sigma=c_3 + c_4 b$ are density operators.
    The claim follows from \cref{lem:minimality-condition}.

    \emph{Case (2).} We have $J\Jcong L(\CC^d)$, $d\ge1$.
    By \cref{prop:reps}, there are integers $l_1,l_2\ge1$ and a unitary $u:\H \to (\CC^d\ox\CC^{l_1})\oplus (\CC^d\ox\CC^{l_2})$ such that $J = \{(a \ox \1) \oplus (a^t\ox\1) : a\in L(\CC^d)\}$.
    Thus, $(\pi,\H)$ with $\pi(a) = u^*((a\ox\1)\oplus (a^t\ox\1)) u$ is a J*-representation of $L(\CC^d)$ with $J=\pi(L(\CC^d))$.
    By case (1), there are density matrices $\rho_0,\sigma_0$ on $\CC^d$ with $J_{(\rho_0,\sigma_0)}=L(\CC^d)$.
    Set $\rho = (\rho_0\ox \frac1{2l_1}\1) \oplus (\rho_0^t \ox \frac1{2l_2}\1)$ and define $\sigma$ analogously.
    Then $(\rho,\sigma)$ are invariant under the trace-preserving conditional expectation $E$ onto $J$.
    The UP maps $\pi$ and $\pi^{-1}\circ E$ establish s PTP-equivalence $(\rho,\sigma)\leftrightarrow(\rho_0,\sigma_0)$, so that \cref{thm:ptp-inter} implies $J = \pi(L(\CC^d)) = T(J_{(\rho_0,\sigma_0)}) = J_{(\rho,\sigma)}$.
\end{proof}

A naive attempt to generalize the result to J*-algebras goes as follows:
Let $J=\oplus J_k$ be the direct sum decomposition into J*-factors of a 2-generated J*-algebra $J$ on a Hilbert space $\H$, then each $J_k$ is a 2-generated J*-factor.
By \cref{prop:2-gen}, there exist dichotomies $(\rho_k,\sigma_k)$ on subspaces $\H_k$ such that $J_k = J_{(\rho_k,\sigma_k)} \subset L(\H_k)$, $\H = \oplus_k \H_k$.
To get a dichotomy on $\H$, take a probability distribution $(p_k)$ with $p_k>0$ and set $\rho = (\oplus_k p_k\rho_k$, $\sigma=\oplus_k p_k\sigma_k)$.
It follows that $J = \oplus J_k$ is sufficient for $(\rho,\sigma)$ (see \cref{lem:direct-sum-MSJA-inclusion} below), but minimality is false in general.
This can, for instance, be seen in \cref{exa:trp-doubling}, where the minimal sufficient J*-algebra of a weighted direct sum is strictly smaller than the direct sum of the minimal sufficient J*-algebras.
To prove the non-factorial case, one needs to construct the dichotomies $(\rho_k,\sigma_k)$ in such a way that no new symmetries are introduced by the direct summation.

\begin{lemma}\label{lem:direct-sum-MSJA-inclusion}
    Let $(\rho_k,\sigma_k)$ be faithful dichotomies on Hilbert spaces $\H_k$, $k=1,\ldots n$, and let $p_k>0$ be such that $\sum_k p_k=1$. 
    Consider the faithful dichotomy $(\rho,\sigma) = (\oplus_kp_k\rho_k,\oplus_kp_k\sigma_k)$ on $\H=\oplus_k \H_k$.
    Then 
    \begin{equation}\label{eq:direct-sum-MSJA-inclusion}
        J_{(\rho,\sigma)} \subseteq \oplus_k J_{(\rho_k,\sigma_k)}.
    \end{equation}
\end{lemma}
\begin{proof}
    Denote by $F_k$ the $(\rho_k,\sigma_k)$-preserving conditional expectation onto $J_{(\rho_k,\sigma_k)}$.
    Then $F = \oplus F_k$ is a $(\rho,\sigma)$-preserving conditional expectation onto the direct sum $\oplus_kJ_{(\rho_k,\sigma_k)}$.
    Hence, the latter is sufficient, which is equivalent to \eqref{eq:direct-sum-MSJA-inclusion}.
\end{proof}

\section{2-generated J*-algebras}\label{sec:2-gen}

\begin{proposition}\label{prop:2-gen-classification}
    An abstract J*-algebra $J$ is 2-generated if and only if it is J*-isomorphic to a direct sum of the following J*-factors:
    \begin{itemize}
        \item $L(\CC^d)$, $d\ge 1$;
        \item $L(\CC^d)^t$, $d\ge 2$;
        \item $L(\CC^{2d})^\beta$, $d\ge 4$.
    \end{itemize}
\end{proposition}

We see that the class of 2-generated J*-algebras is almost the same as the class of universally reversible J*-algebras, except that the latter also allows the J*-factor $L(\CC^{2d})^\beta$ with $d=3$.
It is known that 2-generated J*-algebras are reversible in any representation \cite[Cor.~2.3.8]{hanche-olsen_jordan_1984} and, hence, universally reversible.

We say that a real Jordan algebra is 2-generated if it is generated by two of its elements and the identity.

\begin{lemma}\label{lem:symmetric-matrices-are-2-gen}
    The Jordan algebra $\mathrm{Sym}_d(\RR)$ of symmetric real matrices $d\times d$ is generated by two elements for all $d\ge1$.
\end{lemma}

\begin{proof}
    We let $e_{l,m}$ with $l,m\in\{0,\ldots,d-1\}$ denote the standard matrix units.
    We set $f_{m,n} = e_{m,n}+e_{n,m}$, $m,n=0,\ldots,d-1$, where $e_{m,n}$ denote the standard matrix units, and note the following identity
    \begin{equation}\label{eq:kjsdnfs}
        \jp{f_{l,m}}{f_{m,n}} = \frac12 f_{l,n}, \quad \text{if $l,m,n$ are distinct}.
    \end{equation}
    Clearly, $\mathrm{Sym}_d(\RR) = \lin \{ f_{l,m} \}_{l,m}$.
    We set
    \begin{equation}\label{eq:gens-sym}
        a =\sum_{l=0}^{d-1} l\, e_{l,l}, \qquad b = \sum_{l=0}^{d-2} f_{l,l+1}
    \end{equation}
    and claim that $(a,b)$ generate $\mathrm{Sym}_d(\RR)$.
    Let $J$ be the Jordan algebra generated by $a$ and $b$.
    Clearly $a$ and $b$ are symmetric, so that $J\subset \mathrm{Sym}_d(\RR)$.
    We have $e_{l,l} = \frac12 f_{l,l}\in J$, $l=0,\ldots,d-1$, because these are the spectral projections of $a$.
    We have to show $f_{l,m} \in J$ for $l< m<d$. 
    We have
    \begin{align*}
        \trip{e_{l,l}}{b}{e_{l+1,l+1}} 
        &= f_{l,l+1}.
    \end{align*}
    Thus, we have $f_{l,l+1}\in J$ for $0\le l \le d-2$.
    For $d=2$, this finishes the proof. Assume now $d>2$.
    By \eqref{eq:kjsdnfs}, we have 
    \begin{equation}\label{eq:conjugation-generators}
        f_{l,l+2} = 2\jp{f_{l,l+1}}{f_{l+1,l+2}} \in J.
    \end{equation}
    Iterating (the next step would be $f_{l,l+3} = 2 \jp{f_{l,l+1}}{f_{l+1,l+2}} \in J$) shows $f_{l,m} \in \Jstaralg(a,b)$ for all $l<m<d$. This finishes the proof.
\end{proof}

\begin{lemma}\label{lem:quaternions-mats-2-gen}
    For $d\ge4$, the Jordan algebra $\mathrm{Herm}_d(\mathbb H)$ of hermitian quaternionic $3\times 3$ matrices is 2-generated.
\end{lemma}
\begin{proof}
    For a quaternion $q\in \mathbb H$, we define
    \begin{equation}
        f_{l,m}(q) = qe_{l,m} + \bar qe_{m,l}.
    \end{equation}
    Note that $f_{l,m}(1)$ is the matrix $f_{l,m}$ in the proof of \cref{lem:symmetric-matrices-are-2-gen}.
    The relation \eqref{eq:kjsdnfs} generalizes to
    \begin{equation}\label{eq:product-formula}
        \jp{f_{l,m}(q)}{f_{m,n}(p)} = \frac12 f_{l,n}(qp)\quad \text{if $l,m,n$ are distinct},
    \end{equation}
    where $p,q\in\mathbb H$.
    Then $\mathrm{Herm}_d(\mathbb H)$ is spanned by $f_{l,m}(1)$, $f_{l,m}(i)$, $f_{l,m}(j)$ and $f_{l,m}(k)$.
    Let $J$ be the Jordan algebra generated by the matrices
    \begin{equation}\label{eq:symplectic-generators}
        a = \sum_{l=0}^{d-1} l \,e_{l,l}, \qquad b = \sum_{l=0}^{d-2} f_{l,l+1} + f_{0,2}(i) +  f_{1,3}(j) 
    \end{equation}
    As in the proof of \cref{lem:symmetric-matrices-are-2-gen}, we have $f_{l,m}(1)\in J$ for all $l,m$.
    In particular, $e_{l,l}\in J$ for all $l$.
    Moreover, we have
    \begin{equation}\label{eq:osidu}
        f_{0,2}(i) = 2\trip{e_{0,0}}{b}{e_{2,2}} \in J, \qquad f_{1,3}(j) = 2\trip{e_{1,1}}{b}{e_{3,3}}\in J.
    \end{equation}
    We use \eqref{eq:osidu} and \eqref{eq:product-formula} to obtain
    \begin{equation}
        f_{l,m}(i) = 4 \jp{f_{l,0}(1)}{ \jp{f_{0,2}(i)}{f_{2,m}(1)}} \in J.
    \end{equation}
    Analogously, we get $f_{l,m}(j)\in J$.
    To see $f_{l,m}(k)\in J$, pick some $n \ne l,m$. We can now apply \eqref{eq:product-formula} to get $f_{l,m}(k) = 2 \jp{f_{l,n}(i)}{f_{n,m}(j)} \in J$.
    This completes the proof.
\end{proof}

\begin{lemma}\label{lem:cx-mats-2-gen}
    For $d\ge1$, the Jordan algebra $\mathrm{Herm}_d(\CC)=L(\CC^d)_h$ of hermitian complex $d\times d$ matrices is 2-generated.
\end{lemma}
\begin{proof}
    For $z\in \CC$, we again write $f_{l,m}(z) = z e_{l,m}+\bar z f_{l,m}$. Note that we can use the product formula \eqref{eq:product-formula} with $p,q\in\CC$.
    We set
    \begin{equation}
        a=\sum_{l=0}^{d-1} l e_{l,l},\qquad b=\sum_{l=0}^{d-2} f_{l,l+1}(1) + f_{0,1}(i).
    \end{equation}
    It then follows, as in the proof of \cref{lem:quaternions-mats-2-gen}, that the Jordan algebra generated by $a,b$ is the full Jordan algebra of hermitian complex matrices.
\end{proof}

\begin{lemma}\label{lem:H3-not-2-gen}
    The Jordan algebra $H_3(\mathbb H)$ of hermitian quaternionic $3\times 3$ matrices is not 2-generated.
\end{lemma}
\begin{proof}
    We show that the Jordan algebra $J$ generated by an arbitrary pair $a,b\in H_3(\mathbb H)$ and the identity matrix is always contained in an isomorphic copy of the hermitian complex matrices.
    Indeed, by the spectral theorem for quaternionic matrices \cite{zhang_quaternions_1997}, there is a unitary quaternionic matrix $u\in M_3(\mathbb H)$ diagonalizing $a$, i.e., $u$ is such that $uau^*=\mathrm{diag}(\alpha_1,\alpha_2,\alpha_3)$ for $\alpha_1,\alpha_2,\alpha_3\in\RR$.
    Set $a'=uau^*$ $b' = ubu^*$.
    Note that the map $u(\placeholder)u^*$ is a Jordan isomorphism on $H_3(\mathbb H)$.
    Let 
    \begin{equation}
        b' = \begin{pmatrix}
            \beta_1 &x &y \\
            \bar x & \beta_2 & z\\
            \bar y&\bar z & \beta_3 
        \end{pmatrix}, \qquad \beta_1,\beta_2,\beta_3\in\RR,\ x,y,z\in \mathbb H.
    \end{equation}
    We set $v = \mathrm{diag}(1,\abs x^{-1}x,\abs y^{-1}y)$ (setting the ratio to 1 if $x$ or $y$ is zero). 
    Then $va'v^*= a'$ because $a'$ is diagonal and real-valued and 
    \begin{equation}
        vb'v^* = \begin{pmatrix}
            \beta_1 &\abs x &\abs y \\
            \abs x & \beta_2 & q\\
            \abs y& \bar q & \beta_3 
        \end{pmatrix}, \qquad q := \abs x^{-1}\abs y^{-1} xz\bar y.
    \end{equation}
    We now pick an imaginary unit $s$, i.e., a quaternion with $s^2=-1$, such that $q \in \CC_s = \RR+s\RR$. 
    By construction, both $vb'v^*$ and $va'v^*=a'$ are $\CC_s$-valued matrices. 
    The hermitian $\CC_s$-valued matrices $H_3(\CC_s)$ are closed under the Jordan product (they are a Jordan-isomorphic copy of the complex hermitian matrices).
    Therefore, $J$ is contained in $u^*v^*H_3(\CC_s)uv$, which is Jordan isomorphic to $H_3(\CC_s)$, which is Jordan isomorphic to $H_3(\CC) = L(\CC^3)_h$.
\end{proof}

\begin{lemma}\label{lem:2-gen-oplus}
    Let $J_1,\ldots J_n$ be abstract J*-algebras, then $\oplus_k J_k$ is 2-generated if and only if each $J_k$ is 2-generated.
\end{lemma}
\begin{proof}
    Let $e_k\in \oplus_k J_k$ be the central projection corresponding to the unit $1\in J_k$.
    If $\oplus_k J_k$ is generated by the unit and a pair of hermitian elements $a,b$, then each $J_k$ is generated by the unit and  $a_k= e_kae_k$, $b_k = e_k be_k$.
    Conversely, assume that $J_k$ is generated by $(e_k,a_k,b_k)$ for each $k$.
    Pick weights $\lambda_k>0$ such that the spectra of the elements $\lambda_k a_k$, $k=1,\ldots, n$, are distinct and the spectra of the elements $\lambda_k b_k$, $k=1,\ldots, n$, are distinct. 
    Then, for each $k$, the elements $a_k$ and $b_k$ (now regarded as elements of the direct sum) can be obtained via the functional calculus from $a=\oplus_l a_l$ and $b=\oplus_k b_k$, respectively. Therefore, the J*-algebra generated by $(1,a,b)$ contains the direct sum of the J*-algebras generated by $(e_k,a_k,b_k)$, which is $J$.
\end{proof}

\begin{proof}[Proof of \cref{prop:2-gen-classification}]
    By \cref{lem:2-gen-oplus}, we only have to show that a J*-factor $J$ is 2-generated if and only if it is J*-isomorphic to one of the listed J*-factors. 
    We use the classification of J*-factors (see \cref{thm:classification}).
    By definition, a J*-factor $J$ is 2-generated if and only if the hermitian part $J_h$ is a real Jordan algebra generated by two elements.
    If $J$ is J*-isomorphic to $L(\CC^d)$ ($d\ge1$), $L(\CC^d)^t$ ($d\ge2$), or $L(\CC^{2d})^\beta$ ($d\ge4$), it is 2-generated by claim follows from \cref{lem:cx-mats-2-gen,lem:quaternions-mats-2-gen,lem:symmetric-matrices-are-2-gen}.
    If $J$ is J*-isomorphic to $L(\CC^{2d})^\beta$, it is 2-generated by \cref{lem:H3-not-2-gen}.
    In all other cases, $J$ is J*-isomorphic to a spin factor $V_n$ with $n=4$ or $n\ge 6$. 
    By \cref{prop:reversible}, this entails that $J$ is irreversible, which contradicts 2-generatedness.
    This finishes the proof.
\end{proof}

\section{Proof of Frenkel's integral formula for approximately finite-dimensional von Neumann algebras}
\label{app:frenkel}

For normal states $\omega,\varphi$ on a von Neumann algebra $M$, the Hockey stick divergence is defined as
\begin{equation}
    E_t(\omega\|\varphi) =(\omega-t\varphi)^+(1) = \frac12 \big( \norm{\omega-t\phi} +1-t \big), \qquad t>1,
\end{equation}
where $\chi^+$ denotes the positive part of a normal linear functional $\chi\in M_*$.

\begin{proposition}\label{prop:frenkel-vna}
    Let $M$ be an approximately finite-dimensional von Neumann algebra and let $\omega, \varphi$ be normal states on $M$.
    Then
    \begin{equation}\label{eq:vN-Frenkel}
        D(\omega \| \varphi) = \int_1^\oo \left(\frac1t E_t(\omega\|\varphi) + \frac1{t^2} E_t(\varphi\|\omega) \right) dt.
    \end{equation}
\end{proposition}

In the particular case $M = L(\H)$ with $\dim \H=\oo$, the statement was shown in \cite{jencova_recoverability_2024}, albeit with a more complicated proof. 

\begin{proof}
    Frenkel showed his formula for density operators on finite dimensional Hilbert spaces \cite{frenkel_integral_2023}.
    Algebraically phrased, this means that the formula holds for states on matrix algebras.
    By standard procedure, it extends to direct sums of matrix algebras $M\cong L(\CC^n)$.
    Since both sides of \eqref{eq:vN-Frenkel} are invariant under *-isomorphism, the formula holds for all states on finite-dimensional unital *-algebras.

    If $M$ is approximately finite-dimensional, there is an increasing net $(M_\alpha)_\alpha$ of finite-dimensional *-subalgebras $M_\alpha \subset M$ with
    \begin{equation}
        M = \overline{\bigcup_\alpha\, M_\alpha}^{\mathrm{uw}}.
    \end{equation}
    We define projective nets of states $(\omega_\alpha)_\alpha$, $(\varphi_\alpha)_\alpha$ via $\omega_\alpha = \omega|_{M_\alpha}$, $\varphi_\alpha=\varphi|_{M_\alpha}$.
    By the approximation property of the relative entropy \cite[Cor.~II.5.12]{ohya_quantum_1993}, we have
    \begin{equation}\label{eq:sdasd}
        D(\omega\|\varphi) = \lim_\alpha D(\omega_\alpha\|\varphi_\alpha).
    \end{equation}
    Using $\omega_\alpha = \omega_\beta|_{M_\alpha}$, $\varphi_\alpha = \varphi_{\beta}|_{M_\alpha}$ for $\beta>\alpha$ and the monotonicity (data processing inequality) of the hockey stick divergence, we see that 
    \begin{equation}\label{eq:2i3ueh2}
        \frac1t E_\alpha(\omega_\alpha\|\varphi_\alpha) + \frac1{t^2} E_t(\varphi_\alpha\|\omega_\alpha) \qquad \text{is non-decreasing in $\alpha$.}
    \end{equation}
    The density of $\bigcup_\alpha M_\alpha$ in $M$ implies
    \begin{equation}\label{eq:sdu213}
        E_t(\omega\|\varphi) = \frac12( \norm{\omega-t\phi}-1-t ) =  \lim_\alpha \frac12( \norm{\omega_\alpha-t\phi_\alpha}-1-t ) = \lim_\alpha E_t(\omega_\alpha\|\varphi_\alpha).
    \end{equation}
    Using the monotone convergence theorem, Frenkel's integral formula for states on finite-dimensional unital *-algebras, and equations \eqref{eq:sdasd}, \eqref{eq:2i3ueh2}, and \eqref{eq:sdu213}, we find
    \begin{align*}
        \int_1^\oo \left(\frac1t E_t(\omega\|\varphi) + \frac1{t^2} E_t(\varphi\|\omega) \right) dt
        &= \int_1^\oo \lim_\alpha \left(\frac1t E_t(\omega_\alpha\|\varphi_\alpha) + \frac1{t^2} E_t(\varphi\|\omega) \right) dt\\
        & =\lim_\alpha \int_1^\oo \left(\frac1t E_t(\omega_\alpha\|\varphi_\alpha) + \frac1{t^2} E_t(\varphi\|\omega) \right) dt\\
        & = \lim_\alpha D(\omega_\alpha\|\varphi_\alpha) = D(\omega\|\varphi).
    \end{align*}
\end{proof}

\printbibliography

\end{document}